%% file: coNext-arxiv.tex
\documentclass[sigconf,nonacm]{acmart}

\acmYear{2021}
\copyrightyear{2021}
\setcopyright{acmcopyright}
\acmConference{CoNEXT '21}{December 7--10, 2021}{Virtual Event, Germany}
\acmPrice{TBA} % Use the data from your email from ACM rights management
\acmDOI{TBA} % Use the data from your email from ACM rights management
\acmISBN{TBA} % Use the data from your email from ACM rights management

\usepackage{amsfonts}
\usepackage{amsthm}
\usepackage[utf8]{inputenc}
\usepackage{tikz,forest}
\usepackage{comment}
\usepackage{algorithm}
\usepackage[noend]{algpseudocode}
\usepackage[english]{babel}
\usepackage{url}
\usetikzlibrary{shapes.geometric, arrows}
\usepackage{paralist}
\usepackage{caption}
\usepackage{subcaption}
\usepackage{xspace}
\usepackage{footnote}
\usepackage{epigraph}
\makesavenoteenv{algorithm}

\algrenewcommand\algorithmicindent{1.0em}%

\newtheorem{theorem}{Theorem}[section]

\newtheorem{lemma}[theorem]{Lemma}

\theoremstyle{remark}

\theoremstyle{definition}
\newtheorem{definition}{Definition}[section]

\usepackage{mathtools}
\allowdisplaybreaks

\newcommand{\abs}[1]{\left\vert#1\right\vert}
\newcommand{\set}[1]{\left\{#1\right\}}
\newcommand{\reals}{\mathbb{R}}

\newcommand{\naturals}{\mathbb{N}}

\DeclareMathOperator*{\argmin}{arg\,min}

\newcommand{\topalg}{Top}
\newcommand{\maxalg}{Max}
\newcommand{\levelalg}{Level}
\newcommand{\numswitches}{n}
\newcommand{\switchset}{\mathcal{S}}
\newcommand{\switch}{s}
\newcommand{\serverset}{\mathcal{W}}
\newcommand{\server}{w}
\newcommand{\network}{T}
\newcommand{\vertexset}{V}
\newcommand{\linkset}{E}
\newcommand{\rootswitch}{r}
\newcommand{\link}{e}
\newcommand{\load}{L}
\newcommand{\parent}{p}
\newcommand{\numblue}{k}
\newcommand{\blueset}{U}

\newcommand{\blue}{B}
\newcommand{\red}{R}
\newcommand{\weight}{\omega}
\newcommand{\rate}{\rho}
\newcommand{\aggcap}{a}
\newcommand{\msg}{\mathrm{msg}}
\newcommand{\childnum}{C}
\newcommand{\bic}{\ensuremath{BIC}}
\newcommand{\bica}{\ensuremath{\phi}-\bic}
\newcommand{\msgsize}{M}
\newcommand{\msgcost}{\phi}
\newcommand{\destination}{d}
\newcommand{\alg}{SOAR}

\newcommand{\reduce}{Reduce\xspace}

\DeclareMathOperator{\minsplit}{mSplit}
\DeclareMathOperator{\mincost}{mCost}
\DeclareMathOperator{\nodecolor}{color}

\newcommand{\distance}{D}
\newcommand{\height}{h}
\newcommand{\wcapp}{WC}
\newcommand{\psapp}{PS}
\DeclareMathOperator{\binarytree}{BT}
\DeclareMathOperator{\scalefree}{SF}
\newcommand{\btnet}[1]{\binarytree(#1)}
\newcommand{\sfnet}[1]{\scalefree(#1)}

\newcommand{\rpa}{RPA}
\newcommand{\Avilabilty}{\Lambda}

\newcommand{\Path}{\tau}

\newcommand{\aggc}{\pi}

\newcommand{\dpelli}{X}
\newcommand{\dpellicol}{Y}

\newcommand{\Gather}{Gather}
\newcommand{\Color}{Color}
\newcommand{\alggather}{\alg-\Gather}
\newcommand{\algcolor}{\alg-\Color}

\newcommand{\myleftp}[1]{\left(\rule{0 cm}{#1}\right.}
\newcommand{\myrightp}[1]{\left.\rule{0 cm}{#1}\right)}
\newcommand{\mylefts}[1]{\left[\rule{0 cm}{#1}\right.}
\newcommand{\myrights}[1]{\left.\rule{0 cm}{#1}\right]}
\newcommand{\mylp}{\myleftp{0.6cm}}
\newcommand{\myrp}{\myrightp{0.6cm}}
\newcommand{\myls}{\mylefts{0.7cm}}
\newcommand{\myrs}{\myrights{0.7cm}}

\tikzset{
  treenode/.style = {align=center, inner sep=0pt, text centered},
  node_b/.style = {treenode, circle, white, draw=black, font=\Huge, fill=blue, text width=3em},
  node_r_full/.style = {treenode, circle, white, draw=black, font=\Huge, fill=red, text width=3em},
  node_r/.style = {treenode, circle, red, draw=red, font=\Huge, text width=3em, very thick},
  node_ws/.style = {treenode, circle, draw=white, font=\Huge, text width=3em, very thick},
  node_p/.style = {treenode, rectangle, draw=black, font=\Huge, minimum height = 1.1cm, text width=3em, very thick},
  node_h/.style = {treenode, rectangle, white, draw=black, font=\Huge, minimum height = 1.1cm, text width=3em, fill=gray, very thick},
  node_wh/.style = {treenode, rectangle, draw=white, font=\Huge, minimum height = 1.1cm, text width=3em, very thick},
  edge from parent/.style = {font=\Huge, line width=1mm, draw, <-, >=stealth'},
  level 2/.style = {sibling distance=40mm},
  level 3/.style = {sibling distance=20mm},
  level/.style = {level distance=2.5cm},
}

\usepackage{soul}

\newcommand{\revision}[1]{#1}
\newcommand{\newrevision}[1]{#1}

\date{}

\title[\alg: Minimizing Network Utilization with Bounded In-network Computing]{\alg: Minimizing Network Utilization with Bounded In-network Computing
}

\author{Raz Segal,  Chen Avin,  Gabriel Scalosub}

\affiliation{%
    \department{School of Electrical and Computer Engineering}
    \institution{Ben-Gurion University of the Negev, Israel}
}
\begin{document}

% ==================================================
\begin{abstract}
In-network computing via smart networking devices is a recent trend for modern datacenter networks.
State-of-the-art switches with near line rate computing and aggregation capabilities are developed to enable, e.g., acceleration and better utilization for modern applications like big data analytics, and large-scale distributed and federated machine learning.
We formulate and study the problem of activating a limited number of in-network computing devices within a network, aiming at reducing the overall network utilization for a given workload.
Such limitations on the number of in-network computing elements per workload arise, e.g., in incremental upgrades of network infrastructure, and are also due to requiring specialized middleboxes, or FPGAs, that should support heterogeneous workloads, and multiple tenants.

We present an optimal and efficient algorithm for placing such devices in tree networks with arbitrary link rates, and further evaluate our proposed solution in various scenarios and for various tasks.
Our results show that having merely a small fraction of network devices support in-network aggregation can lead to a significant reduction in network utilization.
Furthermore, we show that various intuitive strategies for performing such placements exhibit significantly inferior performance compared to our solution, for varying workloads, tasks, and link rates.
\end{abstract}

\maketitle
% ==================================================

\section{Introduction}

% Paragraph 1
Datacenter networks and their distributed data processing capabilities are the driving force behind
% our digital world and its day-to-day applications
% like social networking, on-line shopping and others.
leading applications and services, including search engines, content distribution, social networks and eCommerce.
Recent work has shown that for many of the tasks performed by such applications, the network (and not server computation) is the actual bottleneck hindering the ability to optimize 
% to father increase the overall 
computation efficiency and performance~\cite{chowdhury11managing,mai14netagg,viswanathan20network}.
% \revision{NEED CITATION}.
Such bottlenecks occur, e.g., in
% distributed data processing using {\em Reduce} and {\em allReduce} operations in
distributed and federated machine learning (e.g., {\em AllReduce}), and in solutions employing the MapReduce methodology for big data tasks,
and more generally in scenarios giving rise to the {\em incast} problem~\cite{alizadeh10dctcp,wu13ictcp}.
%Reduce and allReduce are a basic building blocks of such computations e.g., stochastic gradient decent (and others)
%For many of such computations the network is the bottleneck 

% Paragraph 2
In order to tackle these deficiencies, recent research has been pushing the concept of {\em in-network computing}~~\cite{ports19when,sapio17innetwork}, which suggests
offloading a considerable portion of the computation onto ``smart'' networking elements, thus relieving end-hosts and servers from some of the computational tasks, resulting in improved efficiency and performance.
%study in-network processing (computation) to optimize reduce like operation, under limited "budget"
In proposing this paradigm, attempts were made to characterize the types of computation that could potentially benefit from such an approach~\cite{costa12camdoop}.
% \revision{MORE GOOD CITATION NEEDED}.
Indeed, recent works showed that modern switches can perform local computation on packets, like reduce operations, even at line rate \cite{graham20sharp,gebara21innetwork}.
Such {\em computing switches} can be implemented, for example, using  
SDN and programmable network elements (e.g., using P4)~\cite{bosshart14p4}, and have been shown 
% Using such aggregation switches has been shown 
to significantly improve network utilization,
% and increase the network capacity,
which in turn improves overall application performance, and resource usage efficiency~\cite{graham20sharp,gebara21innetwork}.
It should be noted that the question of whether such offloading approaches are beneficial or advised is not without controversies~\cite{murphy19thoughts}.
However, data aggregation, as performed in, e.g., big-data tasks based on MapReduce, and distributed ML, which are the main use cases considered in our work, are well within consensus, especially when implemented using programmable switches with co-located accelerators (such as FPGAs).
% \revision{CITATION \cite{}}.

Bearing these potential benefits in mind, one should note that employing in-network computing comes at a cost (in the form of, e.g., hardware or availability limitations), and such capabilities might not be ubiquitous throughout the network.
For example, this could be the case in an incremental upgrade of the network, where some (but not all) legacy switches are replaced by more capable network elements.
In addition, in many cases such in-network computing require specialized middleboxes, or FPGAs, which might call for independent, possibly partial, deployment.
Lastly, even if such in-network computing capabilities are indeed available throughout the network, the available resources required to support the various workloads requiring such computation might not be sufficient for servicing all such workloads.
In such a case, one would need to allocate in-network computing resources sparingly to the various workloads, so as to optimize overall system performance.
% or in a scenario where the network performs many parallel reduce-like operations and 
% \emph{each} reduce can therefore use only a limited number of aggregation switches.
We therefore focus our attention on in-network computation tasks, while using a {\em limited} number of in-network processing devices.

In particular, we consider the task of {\em data aggregation} as it occurs in, e.g., MapReduce frameworks, or distributed machine learning using a parameter server.
% We therefore focus, in this paper, the case of a network-wide reduce operation, but with a limited number of available "super" switches.
%To the best of our knowledge, the setup we study here was not formally examined before. 

% Paragraph 3
%\chen{we do need to say something short so he will right away know we checked. imo.}
We study such in-network computing paradigms in tree-based topologies where
% reduce operation as supported e.g., in distributed machine learning
%the NVIDIA Collective Communication Library (NCCL) 
% \cite{nvidia19doubletree,sanders2009two}.
%\gabi{I don't understand why this is an example to what we do. I think there are other works that are much more relevant.}
%\chen{motivation to mention machine learning, tree and Nvidia, but this actually an overlay tree, so I'm not sure}
given a tree network of switches, each connected to some number of servers (e.g., as Top-of-Rack switches), our goal is to perform data aggregation by means of a {\em Reduce} operation; We wish to send the aggregated data from all the servers in the network, towards a special {\em destination} server $d$.
% denoted as the root (or the destination $d$) (this aka as converge-cast in the literature \revision{CITATION \cite{}}.
We note that such tree-based topologies are becoming increasingly popular for distributed machine-learning use cases, implementing, e.g., AllReduce operations~\cite{nvidia19doubletree,sanders2009two,gebara21innetwork}.
% \gabi{add Manya, Gil Bloch (appears in Manya's paper)}

A simple example of our problem is depicted in Fig.~\ref{fig:intro}, where the destination server $d$ is connected via a tree 
to six servers.
Initially each server $i$ holds a value $x_i$ and $d$ needs to compute a function $f(x_1, x_2, \dots, x_6)$ over all the values available at the servers.
%\gabi{I don't think the figure really contributes much -- MapReduce / data aggregation are pretty "standard" concepts, and the figure doesn't really add to that, I think.}
To perform this task more efficiently, we have at our disposal a limited {\em budget} of $\numblue$ aggregation switches, which should be deployed (or activated) in some $\numblue$ locations in the tree network.
While several different metrics of interest could be considered, in the current work we focus on optimizing the {\em utilization complexity}, where one strives to minimize the {\em total transmission time} throughout the network while performing the Reduce operation.
This is equivalent to minimizing the average transmission time over all links.
When link rates are the same across the network (e.g.,  rate 1), the utilization complexity is proportional (or even identical) to the message complexity \cite{peleg2000distributed}, the total number of messages sent during the operation.
In this work, we consider the more general case of having  arbitrary rates at the links.
The utilization complexity therefore serves as a generalization of message complexity, and can be considered as a basic metric for the performance of network algorithms, where we apply it to studying the efficiency of the Reduce operation.
We note that for a given network capacity induced by the link rates,
the ability to maintain a low utilization complexity is expected to allow supporting more workloads.

The benefits of having an aggregation switch at some location is that such a switch can perform local aggregation of messages, i.e., aggregating multiple incoming messages onto a {\em single} outgoing message. Hence, the judicious allocation of these aggregation switches can assist in significantly reducing the utilization complexity.
It should be noted that the utilization complexity is closely correlated with the actual {\em bandwidth consumption} (in Bytes) of the system, while performing a Reduce operation (as we further demonstrate in Sec.~\ref{sec:evaluation}).

To better illustrate the notion of utilization complexity, assume for example that server $d$ needs to compute an aggregate function $f$ (e.g., sum) of the $x_i$'s in Fig.~\ref{fig:intro}.
Two extreme in-network allocations are
\begin{inparaenum}[(i)]
\item the {\em all-red} solution, where none of the switches serve as aggregation switches, requiring no (i.e., $\numblue=0$) aggregation switches, and
\item the {\em all-blue} solution, where all switches are aggregation switches, thus requiring the allocation of $\numblue=5$ aggregation switches.
\end{inparaenum}
If we assume for simplicity that all link rates are 1, the all-red solution translates to having a utilization complexity of 14, as there is an overall of 14 messages traversing the network, where the all-blue solution will require merely 5 -- the number of edges in the tree.
%Note that the number of transmitted bytes depends also on the input size.

% This scenario can be used for example to compute map-reduce wordcount \revision{CITATION \cite{}} or distributed ML, e.g., stochastic gradient decent, via a parameter server \revision{CITATION \cite{}}. To perform {\em allReduce} the root can simply broadcast the final aggregation it has. 

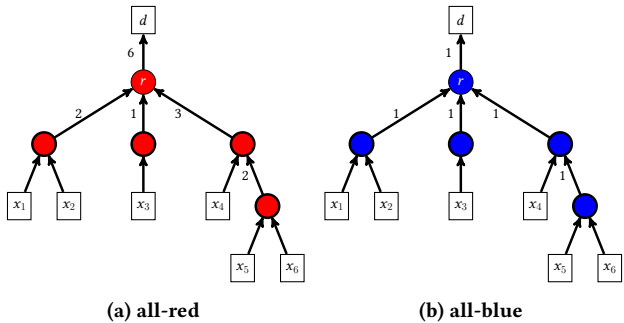
\begin{figure}[t]
    \centering
    \subcaptionbox{all-red}{
        \resizebox{0.47\columnwidth}{!}{\input{aggregation_allred}}
    }
    \subcaptionbox{all-blue}{
        \resizebox{0.47\columnwidth}{!}{\input{aggregation_allblue}}
    }
    \caption[The number of messages on each edge in \emph{all-red} and \emph{all-blue} trees]{The number of messages on each edge in \emph{all-red} and \emph{all-blue} aggregation trees. Destination computes a function $f$ on $x_1,\ldots,x_6$. \revision{The number of messages sent on each edge between two switches is denoted.}} 
    \label{fig:intro}
\end{figure}

As it turns out, for non extremal cases of $\numblue$, finding the optimal placement of the aggregation switches is not a trivial task, even for trees, which is the case being studied in this work.
This follows from the fact that the optimal placement of the aggregation switches depends both on the (possibly complex) tree topology and links rates, as well as on the (possibly complex) load distribution at the servers.
Moreover, multiple aggregation switches on the unique path from a server to the tree root introduce dependencies between the switches, which render standard approaches, like greedy, or divide and conquer, inapplicable.

% Paragraph 4
We believe that our problem setup could also be used, for example, by cloud providers that can offer such a service as part of their Network-as-a-Service (NaaS) offerings, where each client can choose its required amount of aggregation switches based on the performance it needs.
% Additionally, in some cases our setup may be used to overcome the in-cast problem in datacenters \revision{CITATION \cite{}}.
% \gabi{This was already mentioned earlier}

We note that our work focuses primarily on reducing the bandwidth footprint, thus maximizing the effective utilization of the networking resources.
More recently, the focus has also been given to highlighting networking bottlenecks that are due to transport-level deficiencies, which hinder exploiting the full potential of distributed applications such as big data tasks and ML~\cite{zhang20is}.
Our proposed approach can be applied alongside any solutions being thus developed for other layers of the networking stack.

% Paragraph 5

% ==================================================
\subsection{Our Contribution}

\revision{We formulate the {\em Bounded In-network Computing} (BIC) problem, aiming at minimizing the utilization complexity}, and present an optimal and time efficient algorithm for solving the problem on tree networks with arbitrary, heterogeneous, link rates.
Such topologies are common in datacenter networks, e.g., fat-tree topologies \cite{al2008scalable}.
Our algorithm uses dynamic-programming with a non-trivial parameterized potential function.

While our mathematical formulation is for a single workload (or tenant), we extend it to support multiple 
workloads that arrive in an {\em online} manner, each requiring the allocation of (some) in-network aggregation switches.
In such a scenario each switch has a 
limited capacity of workloads it can support.
We discuss and present various properties of our resulting solution, and evaluate its performance for various server load distribution, network sizes, and network topologies.
In our study, we further consider two main {\em use cases}:
\begin{inparaenum}[(i)]
\item MapReduce (using word-count as an illustration), and
\item gradient aggregation in distributed machine learning using a parameter server.
\end{inparaenum}
%both regular, random, and dynamically generated networks.
% Finally, we study the performance of our algorithm in two big-data, distributed data processing, use cases: one within the MapReduce framework and a second within the distributed ML framework.
We further show the benefits of using our algorithm when compared with several natural allocation strategies.
Our results indicate that a small fraction of aggregation switches can already significantly diminish the utilization  complexity of data aggregation tasks.

While we use an abstract mathematical model for scatter-gather type applications, we believe the model, and our algorithmic approach, alongside the structural properties it uncovers, may well be suited for further studying other objectives like minimizing the load on bottleneck links, or minimizing the latency of completing the data-transfers of a workload.

The rest of the paper is structured as follows. 
In Sec.~\ref{sec:model} we introduce our formal system model.
Sec.~\ref{sec:example} provides a motivating example highlighting various aspects of the \bic\ problem.
Sec.~\ref{sec:algorithm} presents an overview of our optimal algorithm \alg\, and the main theoretical results. 
We evaluate our algorithm experimentally in Sec.~\ref{sec:evaluation}.
The formal algorithms and analysis are presented in Section \ref{sec:analysis}.
We conclude the paper with related work and discussion in Secs. \ref{sec:relatedwork} and \ref{sec:disussion_future_work}, respectively.
\revision{For readability some of the proofs are deferred to the appendix.}
% ==================================================
\section{Preliminaries \& System Model}\label{sec:model}

We consider a system comprising a set of $\numswitches$ switches $\switchset$, a set of servers (workers) $\serverset$, and a special destination server $\destination \notin \serverset$.
We assume there exists a pre-specified {\em root} switch $\rootswitch \in \switchset$, and a {\em weighted tree network}
% $\network=(\switchset \cup \set{\destination},\linkset \cup \set{(\rootswitch,\destination)})$, 
$\network=(\vertexset ,\linkset,\weight)$,
where
$\vertexset=\switchset \cup \set{\destination}$ and $\linkset=\linkset' \cup \set{(\rootswitch,\destination)}$ for some $\linkset' \subseteq \switchset^2$ forming a tree over the set of switches $\switchset$. 
Let $\weight:\linkset \mapsto \reals^+$ be the rate function of the links (in messages per second).
For $e\in \linkset$ let $\rate(e)=\frac{1}{\weight(e)}$.
The tree $\network$ thus consists of the underlying network topology connecting the switches, and connecting the root $\rootswitch$ to the destination $\destination$.

We further assume that all links in $\linkset$ are directed towards $\destination$.
Let $\Path(u,v)$ denote the unique directed path from $u$ to $v$, if such exist.
In particular, every switch $\switch \in \switchset$ has a unique {\em parent} switch $\parent(\switch) \in \switchset$ defined as the neighbor of $\switch$ on the unique path from $\switch$ to the $\destination$, $\Path(\switch,\destination)$.
In such a case we say $\switch$ is a {\em child} of $\parent(\switch)$, and we let $\childnum(\switch)$ denote the number of children of switch $\switch$.  
When $v$ is an ancestor of $u$ let $\rate(v,u)=\sum_{e \in \Path(v,u)} \rate(e)$,
and for convenience we let $\rate(\switch) = \rate((\switch,\parent(\switch)))$.

We further let $\distance(\switch)$ denote the distance between switch $\switch$ and the root $\rootswitch$, and let $\height(\network)=\max_{\switch}\distance(\switch)$ denote the {\em height} of the tree $\network$.
 
We assume each server $\server \in \serverset$ is connected to a single switch $\switch(\server) \in \switchset$, and let $\load:\switchset \mapsto \naturals$ be the function matching each switch $\switch$ with the number of servers conected to $\switch$.
We refer to $\load$ as the {\em network load}.
Each server $\server$ produces a single message, $x_{\server}$, which is forwarded to $\switch(\server)$, where we assume every message has size at most $\msgsize$, for some (large enough) constant $\msgsize$.
The destination server $\destination$ needs to compute some function $f(\overline{x})$.
% \raz{The network is tasked with set of "work loads",$\workload$. Eatch work load denote it one $\load_i\in \workload$.}
% \raz{Eatch switch can be assigned a finite number of work}
Each switch $\switch$ can be of one of two types, or operates at one of two modes:
\begin{enumerate}[(i)]
\item an {\em aggregating} switch (blue), which can aggregate messages arriving from its children (each of size at most $\msgsize$), to a {\em single} message (also of size at most $\msgsize$) and forwards it to its parent switch $\parent(\switch)$,\footnote{\newrevision{We assume $\msgsize$ is large enough to hold the value of the function $f$ being computed at every node.}}
or
\item a {\em non-aggregating} switch (red), which cannot aggregate messages, and simply forwards each message arriving from any of its children to its parent switch $\parent(\switch)$.
\end{enumerate}

We denote by $\Avilabilty \subseteq \switchset$ the set of switches that are \emph{available}
to serve as aggregation switches.

Our assumption on the aggregation capabilities of aggregating switches is satisfied by systems computing, e.g., separable functions~\cite{mosk2006computing};
\revision{A separable
function of two independent values can be expressed as the product-operator of the values of two individual functions, each one applied to a distinct operand value.
In particular, this holds true for aggregation functions computing, e.g., the count, sum, or max/min of the values contained in the messages being sent by the servers~\cite{goda2019separability}. We leave for future work the study of more complex functions.}

In what follows we will be referring to aggregating switches as {\em blue} nodes in $\network$, and to non-aggregating switches as {\em red} nodes in $\network$.
We denote by a non-negative integer $\numblue$ our \emph{budget}, which serves as an upper bound on the number of blue nodes allowed in $\network$.
We will usually refer to $\blueset \subseteq \Avilabilty$ as the set of blue nodes in $\network$ and require that $\abs{\blueset} \le \numblue$. 

Given a weighted tree network $\network=(\vertexset,\linkset,\weight)$ with a network load $\load:\switchset \mapsto \naturals$, and a set of blue nodes $\blueset \subseteq \Avilabilty$, we consider a simple \reduce operation on $T$ as detailed in Algorithm \ref{alg:reduce}.
Every switch in the tree processes all messages received from its children and forwards message(s) to its parent.
Every blue node (i.e., a node in $\blueset$) is an aggregation switch and all other switches (i.e., nodes not in $\blueset$) are non-aggregation switches.
The operation ends when the {\em destination} receives the overall (possibly aggregated) information from all the nodes that have a strictly positive load.

\begin{algorithm}[t]
\caption[Algorithm]{\reduce$(\network,\load,\blueset)$}
\label{alg:reduce}
\begin{algorithmic}[1]
\Require tree $\network$,  network load $\load$, set of blue nodes $\blueset$
\Ensure aggregate information at destination $\destination$
\State for each node $v$ in $\network$ do:
\While{not received all messages from all children}
        \State process incoming message (by switch type: $\blue, \red$)  
        \State if needed send message to $\parent(v)$ (by switch type: $\blue, \red$)
\EndWhile
\end{algorithmic}
\end{algorithm}

For every link $\link=(\switch,\parent(\switch))$ in $\linkset$, we then define the {\em link message cost} $\msg_{\link}(\network,\load,\blueset)$ as the number of messages traversing link $\link$, given the \reduce operation on $\network$, $\load$, and $\blueset$.
% We further define the {\em network message cost} (or the message complexity) associated with the \reduce operation on $\network$, $\blueset$ and $\load$ to be 
We further define the {\em network utilization cost} (or the utilization complexity) as the total transition time associated with the \reduce operation on $\network$, $\blueset$ and $\load$ to be
\begin{align}
\label{eq:msgcost_def}
% \msgcost(\network,\load,\blueset)=\sum_{\link \in \linkset^*}\msg_{\link}(\network,\load,\blueset).   
\msgcost(\network,\load,\blueset)=\sum_{\link \in \linkset}\msg_{\link}(\network,\load,\blueset)\cdot\rate(e).
\end{align}

The network utilization complexity measures the total (or equivalently, the average) transmission time of all links in preforming the Reduce operation.
% The smaller this value is, the network is more available to preform other operations.

\revision{
In this paper we study the {\em Bounded In-network Computing} (\bic) allocation problem which tries to minimize the network utilization cost. We refer to this problem as the \bica\ problem, which is formally defined as follows: 
} % end revision

\begin{definition}[\bica]
Given a weighted tree network $\network=(\vertexset,\linkset,\weight)$, a network load $\load:\switchset \mapsto \naturals$, a set of available switches $\Avilabilty$, and a budget $\numblue$, the \bica\ problem is finding a set of switches $\blueset \subseteq \Avilabilty$ of size at most $\numblue$ that minimizes the utilization cost $\msgcost(\network,\load,\blueset)$. Formally,
\begin{align}\label{eq:mindef}
    \mbox{\bica}(\network,\load, \Avilabilty, \numblue)=\min_{\blueset \subseteq \Avilabilty, \abs{\blueset}=\numblue} \msgcost(\network,\load,\blueset).
\end{align}
\end{definition}

\revision{
Clearly, one can use a brute-force approach, and consider all possible $\Theta(n^k)$ subsets of size $k$.
}
Although such an approach may work well for a small constant $k$, such an enumeration would result in exorbitant running time for arbitrary values of $k$.
In what follows we will describe and discuss our {\em efficient} solution, \alg, to the \bica\ problem.

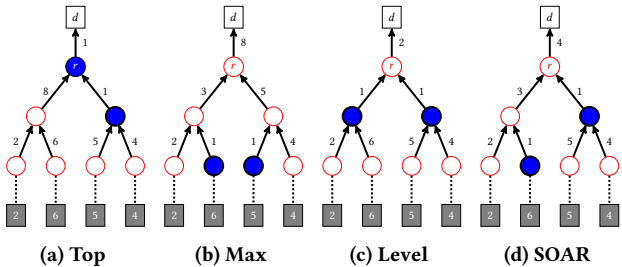
\begin{figure}[t!]
    \centering
    \subcaptionbox{\topalg\label{fig:toy_example_1:top}}{
        \resizebox{0.22\columnwidth}{!}{\input{toy1_1_top}}
    }
    \subcaptionbox{\maxalg\label{fig:toy_example_1:max}}{
        \resizebox{0.22\columnwidth}{!}{\input{toy1_2_max}}
    }
    \subcaptionbox{\levelalg\label{fig:toy_example_1:level}}{
        \resizebox{0.22\columnwidth}{!}{\input{toy1_3_level}}
    }
    \subcaptionbox{\alg\label{fig:toy_example_1:alg}}{
        \resizebox{0.22\columnwidth}{!}{\input{toy1_4_alg}}
    }
    \caption{Example of solutions produced by 4 allocation algorithms for a simple load over a weighted tree network, with constant rates of $1$, $\Avilabilty=\switchset$ and $\numblue=2$ aggregation switches (blue nodes).
    Switches in $\blueset$ are marked by circles, the servers connected to each of the leaf switches are depicted by a gray square noting the load of the switch, and the destination is marked by a white square.
    The tree network $\network$ is defined over $\switchset \cup \set{\destination}$, and edges are directed towards $\destination$.
    In each solution, %the set of blue nodes are marked in blue, and 
    each link in $\network$ is marked by its utilization value. %the number of messages traversing the link.
    Subfigure~\eqref{fig:toy_example_1:top} shows the solution produced by \topalg, with network utilization cost of 27.
    Subfigure~\eqref{fig:toy_example_1:max} shows the solution produced by \maxalg, with network utilization cost of 24.
    Subfigure~\eqref{fig:toy_example_1:level} shows the solution produced by \levelalg, with network utilization cost of 21.
    Subfigure~\eqref{fig:toy_example_1:alg} shows the solution produced by our proposed algorithm, \alg, which obtains the optimal network utilization cost of 20.}
    \label{fig:toy_example_1}
\end{figure}

\section{Motivating Example}
\label{sec:example}
We now turn to consider a motivating example highlighting the fact that simple, yet reasonable, approaches might fall short of finding an optimal solution to the \bica\ problem.
Specifically, we consider the following three allocation strategies for determining the set of blue nodes:
\begin{inparaenum}[(i)]
\item The \topalg\ strategy, which picks the set of $\numblue$ blue nodes as the set closest to the root.
This approach targets reducing the number of messages transmitted in the topmost part of the network, and is motivated by the fact that failing to aggregate messages close to the root may lead to a large number of messages being forwarded from the root to the destination.
\item The \maxalg\ strategy, which picks the set of blue nodes as the $\numblue$ switches with the largest 
% number of neighbors, including servers and other switches.
load.

\item The \levelalg\ strategy, defined for complete binary trees, which aims at partitioning the network into subtrees of similar size, where all the messages within a subtree are aggregated. This is done  by picking a whole level in the complete binary tree as the set of blue nodes.

\end{inparaenum}

We consider a tree network with $n=7$ switches which induces a complete binary tree topology on the set of switches which are all available and support aggregation, and all links have a constant rate of $1$.
Servers are connected only to leaf switches.
% \chen{update weights}
% \raz{updated weights and availability}
Such a topology can be viewed as if the leaf switches are effectively top-of-rack (ToR) switches in a small datacenter topology, where each rack accommodates a distinct number of servers (or VMs).
Fig.~\ref{fig:toy_example_1} provides an illustration of the network and the load being handled.
Each leaf switch is connected to a rack of several servers where the number of servers in the rack is marked in the gray square depicting the rack.
In particular, the load handled by the 4 leaf switches is $(2,6,5,4)$ (from left to right).
% \chen{should we explicit write $L$ for this example?}
% \gabi{done.}
In our example the maximum number of blue switches allowed is set to $\numblue=2$.
Each link $e$ is marked with the utilization cost of this link, $\msg_{\link}(\network,\load,\blueset) \cdot 1$.
% \raz{NOTE: I changed message cost to utilization cost}

Figs. ~\eqref{fig:toy_example_1:top}, \eqref{fig:toy_example_1:max}, and~\eqref{fig:toy_example_1:level} show the results of applying strategies \topalg, \maxalg, and \levelalg, respectively, to such a network and load.
The optimal approach, which is obtained by our proposed algorithm, \alg\ (formally described and analyzed in Sec. ~\ref{sec:algorithm}), ends up picking a non-trivial set of blue nodes as can be seen in Fig. ~\eqref{fig:toy_example_1:alg}.
This allocation strictly outperforms all three contending strategies.

Fig.~\ref{fig:toy_example_2} provides examples of the optimal sets of blue nodes for increasing values of $\numblue$.
We note that in general, optimal solutions need not be unique, and for such cases ($\numblue=1,4$, in Figs.~\eqref{fig:toy_example_2:k_1} and~\eqref{fig:toy_example_2:k_4}) we provide one of these solutions.
However, for some cases ($\numblue=2,3$, in Figs.~\eqref{fig:toy_example_2:k_2} and~\eqref{fig:toy_example_2:k_3}) the optimal solutions are unique.
Considering the specific optimal solutions provided for these cases, we observe that the optimal sets of blue nodes, for increasing values of $\numblue$, are not necessarily monotone. Namely, adding even one more blue node to the set can change the set of blue nodes completely.

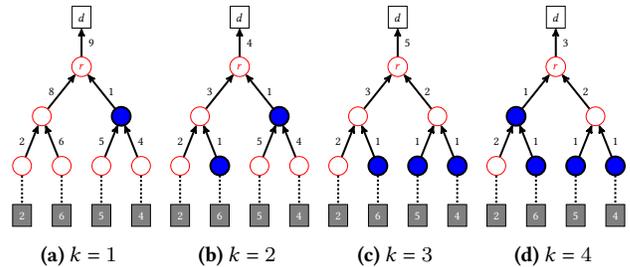
\begin{figure}[t!]
    \centering
    \subcaptionbox{$\numblue=1$\label{fig:toy_example_2:k_1}}{
        \resizebox{0.22\columnwidth}{!}{\input{toy_2_k1}}
    }
    \subcaptionbox{$\numblue=2$\label{fig:toy_example_2:k_2}}{
        \resizebox{0.22\columnwidth}{!}{\input{toy_2_k2}}
    }
    \subcaptionbox{$\numblue=3$\label{fig:toy_example_2:k_3}}{
        \resizebox{0.22\columnwidth}{!}{\input{toy_2_k3}}
    }
    \subcaptionbox{$\numblue=4$\label{fig:toy_example_2:k_4}}{
        \resizebox{0.22\columnwidth}{!}{\input{toy_2_k4}}
    }
    \caption{Example of optimal solutions (produced by \alg) for distinct bounds on the number of allowed blue switches $\numblue=1,2,3,4$, with network utilization costs of $35,20,15,11$, respectively.
    Switches in $\switchset$ are marked by circles, the servers connected to each of the leaf switches are depicted by a gray square noting the load of the switch, and the destination is marked by a white square.
    The tree network $\network$ is defined over $\switchset \cup \set{\destination}$, edges are directed towards $\destination$ and constant rates of 1.
    In each solution, %the set of blue nodes are marked in blue, and 
    each link in $\network$ is marked by its link utilization cost. %the number of messages traversing the link.
    The solutions for $\numblue=2$ (Subfigure~\eqref{fig:toy_example_2:k_2}) and for $\numblue=3$ (Subfigure~\eqref{fig:toy_example_2:k_3}) are unique.
    These two cases serve as an example for the fact that the optimal sets of blue nodes for increasing values of $\numblue$ are not necessarily monotone.}
    \label{fig:toy_example_2}
\end{figure}

% ==================================================

\section{\alg: An Optimal Algorithm}
\label{sec:algorithm}
\epigraph{``Those who sow in tears will reap with songs of joy.''}{Psalm 126:5}

In this section we describe our algorithm, \alg, that produces an optimal solution to the \bica\ problem.\footnote{\alg\ stands for SOw-And-Reap.}
The intuition underlying our algorithm is that if we are able to optimally sow blue nodes in the right locations, then it is possible to reap significantly improved performance in terms of the system's utilization complexity.
The main technical contribution of the paper is the following theorem.
%The following theorem is the main result of this section.
\begin{theorem}
\label{thm:alg_is_optimal}
Given a weighted tree network $\network$,rates $\weight$,a load $\load$, availability $\Avilabilty$, and a bound $\numblue$ on the number of allowed blue switches, algorithm \alg\ solves the \bica\ problem in time $O(\numswitches \cdot \height(\network) \cdot \numblue^2)$.
\end{theorem}

Before describing the algorithm and proving the theorem we first provide some insight as to the structure induced by any solution, and the long-ranging effect of having a sequence of red nodes along a path. These serve to provide a better understanding of our objective function $\msgcost(\network,\load,\blueset)$ capturing the utilization cost of the system.
% ==================================================
\subsection{Re-formulating the Utilization Complexity: A Barrier Perspective} % Message Complexity: A Barrier Perspective}
\label{sec:alg:barrier}
When considering the \bica\ problem, one can view any solution $\blueset$ as inducing a {\em tree partitioning}, such that while scanning the nodes from the leaves towards the root, for every blue node $v \in \blueset$ that has no blue nodes in its subtree, we can {\em detach} the subtree rooted at $v$ from the tree.
Such a blue node $v$ then becomes a leaf in the remaining tree, where its load is set to $1$ in that tree.
The overall utilization cost in the original tree is simply the sum of utilization costs in all the subtrees thus produced.
The reason for this equality is that every blue node in the tree effectively forms a {\em barrier} between the subtree rooted at the node, and the remaining tree above the node.
% Furthermore, every such blue node can effectively be viewed as a leaf with a load of 1 in the remaining tree above it.
Fig.~\ref{fig:barrier} provides an illustration of such a decomposition, and the breakdown of the utilization complexity as the sum of the utilization complexity over the subtrees.
Note that in each subtree, each node that was blue in the original network $\network$ is either a leaf (with load 1), or a destination.%load 1), or a destination.

\begin{figure}[t]
    \centering
    \subcaptionbox{Network $\network$\label{fig:barrier:complete_network}}{
        \resizebox{0.25\columnwidth}{!}{\input{toy_2_k2}}
    }
    \subcaptionbox{Tree decomposition of $\network$\label{fig:barrier:decomposition}}{
        \hspace{-0.35cm}
        \resizebox{0.25\columnwidth}{!}{\input{fig_barrier_3}}
        \resizebox{0.25\columnwidth}{!}{\input{fig_barrier_2}} 
        \hspace{-0.5cm}
        \resizebox{0.22\columnwidth}{!}{\input{fig_barrier_1}}
        % \hspace{0.8cm}
    }
    \caption{Example of the barrier perspective using tree-decomposition.}
    \label{fig:barrier}
\end{figure}
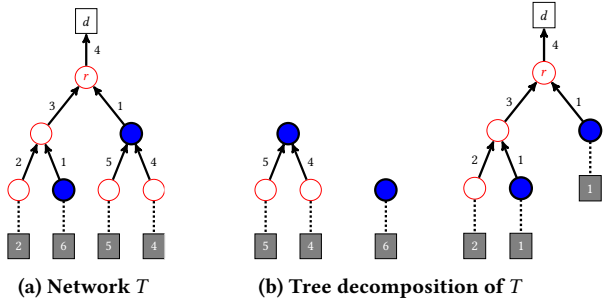

An alternative way to view the system's utilization complexity, which serves as the fulcrum in our proposed algorithm, \alg, is considering the distance of any node $v$ from its {\em closest} blue ancestor (or the destination $d$, if no such blue ancestor exists).
Formally, the next lemma, which follows directly from the definition of $\msgcost$ in Eq.~\ref{eq:msgcost_def}, provides an alternative characterization of $\msgcost(\network,\load,\blueset)$.
%\chen{update weights}\raz{Done}
\begin{lemma}
\label{lem:decompose}
Consider a tree network $\network$ with load $\load$, and consider any set $\blueset \subseteq \network$ of blue nodes.
For each node $v \in T$ let $\parent^*_v=\parent^*_v(\network,\blueset)$ denote $v$'s closest blue ancestor, if one exists, or $d$, otherwise. Then,
\begin{align}
\msgcost(\network,\load,\blueset)
&= \sum_{v \in \blueset} 1 \cdot  \rate(v,\parent^*_v)
+ \sum_{v \notin \blueset} \load(v) \cdot \rate(v,\parent^*_v)
\label{eq:decompose}
\end{align}
\end{lemma}
%\raz{Updated lemma 5.2 please go over}
% Since the sum $\sum_{v \in \network}\load(v)$ is independent of $\blueset$ and depends only on $\load$, we ignore it in our calculations.
To illustrate Eq. \eqref{eq:decompose}, consider Fig.~\eqref{fig:toy_example_2:k_2}.
By Eq. \eqref{eq:msgcost_def}, the sum over all edges (from left to right, bottom-up) is
$(2+1+5+4)+(3+1)+4=20$. Alternatively, by Eq. \eqref{eq:decompose} (considering the relevant nodes from left to right, bottom-up) is $(3+2) + (2 \cdot 3+5 \cdot 1+4 \cdot 1)=20$.

% In turn, Lemma~\ref{lem:decompose} will be the fulcrum which we use to solve the \bica\ problem.

We note that the closest blue ancestor of a node is equivalent to the blue node serving as barrier in our description of the tree-decomposition induced by any set of blue nodes $\blueset$.

\revision{The alternative formulation of our objective in Eq.~\ref{eq:decompose} lays at the core of our proposed algorithm, \alg.
In particular, as we show in the sequel, this formulation will serve to evaluate the potential effect, in terms of utilization, of having a node colored red or blue.
}

% ==================================================

\subsection{Overview of \alg}

In this section we provide a high-level overview of \alg, formally defined in Algorithm~\ref{alg:alg}, which solves the \bica\ problem.
Our solution is based on dynamic programming, and is split into two phases. In the first phase we apply algorithm \alggather, formally defined in Algorithm~\ref{alg:alg:gather} (in Sec.~\ref{sec:analysis}), for gathering the information required for computing an optimal solution. This is followed by the second phase where we apply algorithm \algcolor, formally defined in Algorithm~\ref{alg:alg:color} (in Sec.~\ref{sec:analysis}), which traces back the actual allocation of blue nodes along the breadcrumbs produced in the first phase.
We now provide further details as to each of the phases, and discuss their design criteria.

\begin{algorithm}[t!]
\caption[Algorithm]{\alg$(\network,\load,\Avilabilty,\numblue)$}
\label{alg:alg}
\begin{algorithmic}[1]
\Require A tree $T$, load $\load$, availability $\Avilabilty$, $\numblue$ $\#$ of blue nodes
%\Require A weighted tree $T$, A load $\load$, $\numblue$ \# of blue nodes
\Ensure $\mbox{\bica}(\network,\load,\numblue)$
    \State run \alggather$(\network,\load, \Avilabilty, \numblue)$ at each node $v$
    \Statex
    \Comment{gather $\dpelli_v$ and $\dpellicol_v$ in bottom-up order}
    \State wait until destination receives $\dpelli_r$
    \State run \algcolor$(k)$ at each node $v$
    \Statex
    \Comment{determine node colors in top-down order}
    % \State \alggather$(\network,\load,\numblue)$  
    % \State at the destination $d$
    % \State \algcolor$(k,0)$
    % \State return...
\end{algorithmic}
\end{algorithm}

\begin{figure*}
    \centering
    \subcaptionbox{\alggather\ running example. \label{fig:run_gather_example}}{        \includegraphics[width=1.85\columnwidth]{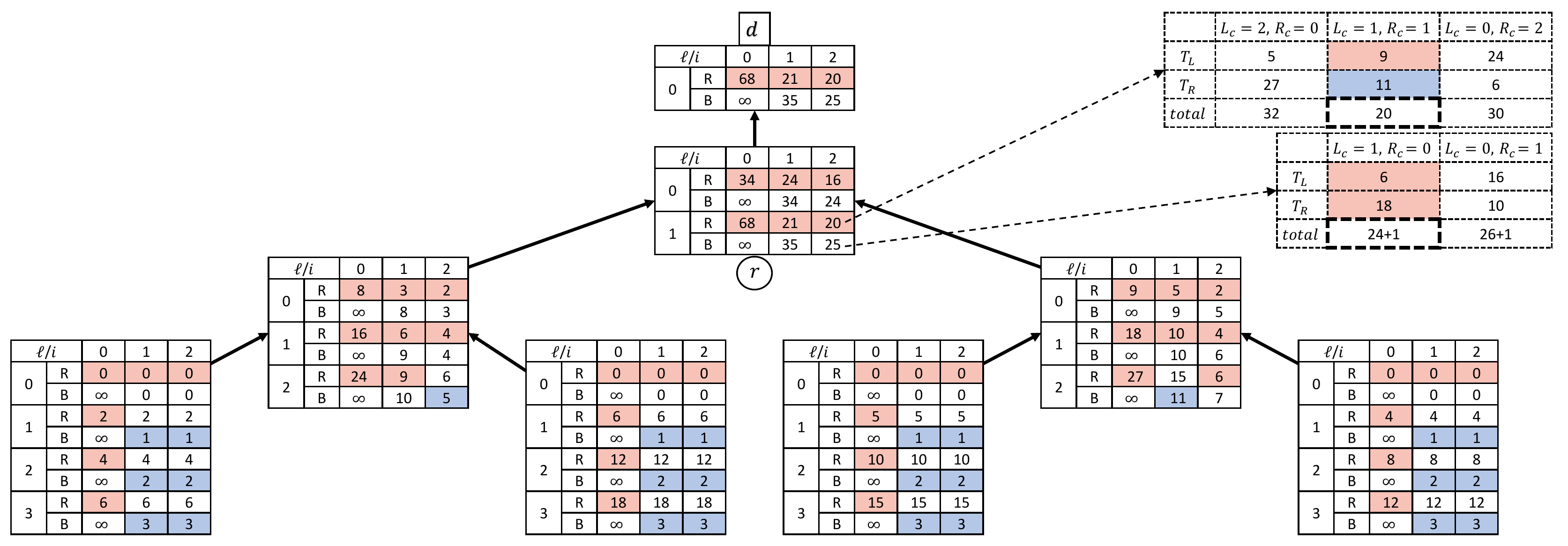}}\\
    \subcaptionbox{\algcolor\ running example.\label{fig:run_color_example}}{
        \includegraphics[width=1.85\columnwidth]{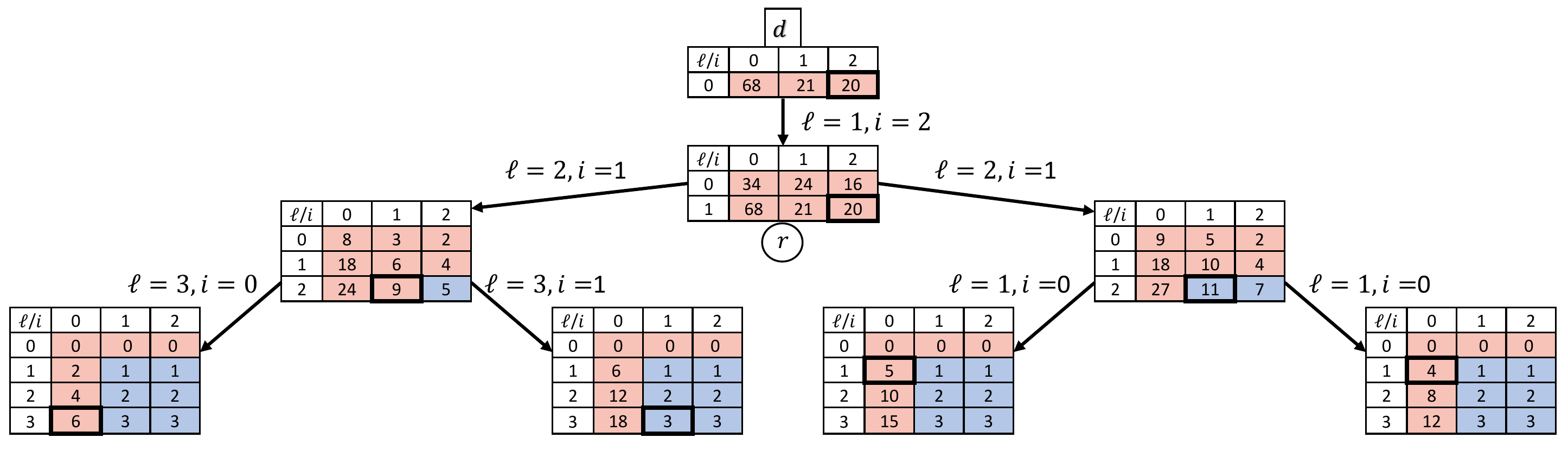}}
    \caption{\alg\ running example}
    \label{fig:run example}
\end{figure*}

\paragraph{\alggather.} Algorithm \alggather, which
\revision{effectively builds the dynamic programming table,}
% while scanning the nodes of the network, using a depth-first-scan (DFS), starting from the root.
% During this phase the algorithm makes use of
uses parameterized potential functions, which account for the long-range effect of having red nodes in the subtree rooted at some given node.
% This phase is formally defined in Algorithm~\ref{alg:alg:gather}.
In particular, a parameter $\ell$ used in these potential functions corresponds to the possible distance of a node in the network from the closest blue ancestor (or the root, if there is no such ancestor) in the network.
Since we don't know the coloring at this phase we compute the potential function for node $v$ and each possible value $\ell$ between zero and $\distance(v)$. 
The key observation that enables us to do this efficiently is that conditioning on a parent $v$ having some specific color (either red or blue) and the value of parameter $\ell$, the subtree $\network_v$ can be independently optimized from the rest of the tree.  

The information gathered during the first phase provides the breadcrumbs required for determining the allocation of blue nodes within the network in the second phase.
In particular, each node gathers two sets of values:
\begin{inparaenum}[(i)]
\item $\dpelli_v$, which prescribes the utilization that would potentially be added in links further up in the tree (for any amount $i$ of blue nodes being distributed in the subtree rooted at $v$) until the closest blue ancestor (at distance $\ell$ from $v$), and
\item $\dpellicol_v$, which registers the distinct partitioning of blue nodes to children of $v$, for any combination of number of blue nodes $i$ that should be distributed in the subtree rooted at $v$, any distance $\ell$ from $v$ to its closest blue ancestor, and for both cases of whether $v$ is blue or red.
\revision{As we show in our proofs in Sec.~\ref{sec:analysis}, these partitions can be computed efficiently.}
\end{inparaenum}
In \alggather\ information is gathered while scanning the nodes of the network from the leafs upwards.

\paragraph{\algcolor.} Algorithm \algcolor, which determines the color of each node, either blue or red, traces
\revision{an optimal path in the dynamic programming table calculated by \alggather.}
% This phase is formally defined in Algorithm~\ref{alg:alg:color}.
Initially, a switch is set to be red, and this is altered only if setting the switch to being blue implies a smaller value of the potential function computed during the first phase, for the specific parameter corresponding to the distance of that node from its nearest blue ancestor, or from the destination $\destination$ (in case the node has no blue ancestor).
The nodes are assigned colors in this manner while scanning them from the root downwards, where each node further alerts each of its children as to the number of blue nodes that should be distributed in the subtree rooted at that child.
\revision{The number of blue nodes assigned to the subtree rooted at each child, i.e., the partitioning of blue nodes across the subtrees of the node, is also available as one of the outputs of \alggather.}

We note that both \alggather\ and \algcolor\ are described as distributed, asynchronous, algorithms, where synchronization between nodes is maintained by waiting for specific messages, or signals, to be received from either the children of a node, or the parent of a node.
For \alggather, the leaves of the tree network $\network$ initiate the messages carrying information upwards in the network, all the way up to the destination server $\destination$.
For \algcolor, the first node to initiate the flow of information is the destination server $\destination$, which sends the bound on the number of allowed blue switches, $\numblue$, to the root node $\rootswitch$.

% ==================================================

\revision{
\subsection{\alg\ Running Example}

In this section we demonstrate a running example of \alg.
We show how the optimal solution in Fig.~\eqref{fig:toy_example_1:alg} was obtained.
Figure~\eqref{fig:run_gather_example} presents the data structures that \alg\ manages in each node during the \alggather\ phase.
Every node $v$ maintain a table with three dimensions, denoting possible parameters of the potential function:
\begin{inparaenum}[(i)]
    \item The number $i$ of possible blue nodes in its subtree, ranging from $0$ to $k$ (columns),
    \item the possible distance $\ell$ from its \emph{closest blue ancestor}, if one exists, or $d$, otherwise (rows), and
    \item the node \emph{color} ($R$ or $B$).
\end{inparaenum}
Fig.~\eqref{fig:run_color_example} shows the how during the \algcolor\ phase the algorithm tracks an optimal path along the tables computed during the \alggather\ phase.

\paragraph{The \alggather\ phase:}
For every node, and every combination of these parameters, the table maintains the minimal total utilization that would be incurred by the subtree $T_v$.
This includes the effect of the color of $v$ on the utilization of links above $v$ ($\ell$ levels up).
These tables are calculated during the \alggather\ phase of the algorithm which proceeds from the leaves towards the root, where each node calculates the values based on the tables available at its children.
Note that a node reports only one value, for every combination of $i$ and $\ell$, taken as the {\em minimum} of being red or blue (the minimum appears with its appropriate color in Fig.~\eqref{fig:run_gather_example}). 

For example, consider the root $r$ and how $r$ calculates the entry for $\ell=1$ (i.e., it is at distance 1 from its nearest blue ancestor, or $d$ -- where in this case it is its distance from $d$) and it has $i=2$ blue nodes to distribute within its subtree.
Assume its children calculated their tables correctly.

First $r$ considers the case where it is colored red ($R$) (top dotted table in Fig.~\eqref{fig:run_gather_example}).
In this case, it has 2 blue nodes to distribute in the subtrees rooted at its children.
Therefore it takes the minimum of three possible cases,
$(L_C=2, R_C=0), (L_C=1, R_C=1)$ or $(L_C=0, R_C=2)$, where $L_C$ and $R_C$ denote the number of blue nodes to be used in the left and right child of $r$, respectively.
By checking the tables of its left and right children, $r$ can find that the minimum is obtained in the case where $(L_C=1, R_C=1)$. where its left child will be colored red and will contribute 9 to the overall utilization, and the right child will be colored blue and will contribute 11 to the overall utilization.
When checking its children tables $r$ considers the entries corresponding to $\ell=2$ since $r$ is assumed to be red, and it is its own value of $\ell$ is 1.
So overall $T_r$ will contribute 20 to the utilization under these settings.

The other alternative is the case where $r$ is colored blue (bottom dotted table in Fig.~\eqref{fig:run_gather_example}).
This leads to only two possible partitions of remaining blue nodes across the subtrees rooted at its children, since $r$ has already ``used-up'' one of the $i=2$ blue nodes available in its subtree.
The minimum configuration is when the left child of $r$ gets to distribute the remaining blue node, $(L_C=1, R_C=0)$, in which case the subtree rooted at this child contributes 6 to the overall utilization (note that we consider the child's table for $\ell=1$, since the child is at distance 1 from its closest blue ancestor, being the root in this case).
The right child gets to distribute no blue nodes, and thus contributes 18 to the overall utilization. The utilization contributed by $T_r$ thus totals 25, where $6+18=24$ are due to the subtrees rooted at the children of $r$, and 1 more contributed by $r$ since its distance to $d$ is one and $r$ is assumed to be blue.
% In this case when checking its children tables $r$ set $\ell=1$ since $r$ is assumed to be blue.
Taking the minimum of 20 and 25, $r$ will report to its parent that for this setting, $(\ell=1, k=2)$, $T_r$ will contribute 20 to the utilization, and its color will be red.

\paragraph{The \algcolor\ phase:}
In the next phase, the coloring is done using \algcolor\ by tracing an optimal path over the tables generated during the \alggather\ phase, from the root to leaves.
For our example, the destination, $\destination$, needs to place $k=2$ blue nodes in the {\em network}.
The utilization in this case is 20.
$\destination$ passes the values $\ell=1$ and $i=2$ (from which the minimal utilization was derived) to its child, $r$.
At this point $r$ looks up the color corresponding to these values in its table, and determines its own color, red in this case. Furthermore, $r$ knows the number of blue nodes available for distribution in its subtree, and its distance from its closest blue ancestor (or $\destination$, in this case).
$r$ can then determine the amount of blue nodes it needs to pass on to each of its children, for distribution in their subtrees.
With this information it can recursively determine the color of its children.
Fig.~\eqref{fig:run_color_example} shows the color and configuration that was selected at each node (bold square) in the optimal solution provided by \alg.
} % end revision

% ==================================================

\revision{
\subsection{\alg: Practical Aspects and Limitations}
}
\revision{
While \alg\ minimizes the network utilization, namely the overall transmission time over all links, using bounded in-network computing, there are various {\em practical} aspects that are related to the implementation of our approach, and the benefits it provides.
% The main contribution of the paper is studying a simple model for bounded  in-network computing and providing an optimal algorithm for tree networks.
% These leave many open questions regarding the practicality of our solution.
For example, how and where should switches store the aggregated values? What are the effects of packet-loss and latency (affecting the delivery of messages)?
% What will happen if messages fail to arrive? 
For line-rate aggregation, what synchronization mechanisms are required?
How can one use our solution in a system handling multiple tenants and workloads, and what would be the overhead of using our approach?
% What could we expect about the actual running time of the algorithm?
% how do we handle multiple workloads and what we can assume on them? 

For the most part, these questions are applicable to most in-network computing environments performing in-network aggregation. However, some aspects are specifically more pronounced in our model, namely, the distinction between aggregating nodes (which wait for all incoming information before forwarding a message), and non-aggregating nodes, which simply follow a store-and-forward regime.
} % end revision
\newrevision{
We plan to study these specific aspects in future work, as we note that these may well affect the performance whenever significant variable delay is manifested (e.g., by inducing an overall lower rate of information flow), and for some aggregating functions the memory tolls may be non-negligible.
} % end new revision

\newrevision{
Another significant aspect related to our approach is the fact that our model assumes that any message being transmitted throughout the system is of size at most $\msgsize$.
For some aggregation functions (e.g., bitwise-functions, or max/min), assuming such a bound is quite reasonable as it can be determined by the maximum size of a message generated by the servers. The optimality of \alg\ relies on this assumption.
However, for other functions, performing aggregation, and furthermore doing so repeatedly, might result in a message size increase that may be proportional to, or at least monotone with, the size of the workload (e.g., sum or product functions).
In such cases, {\em \alg\ is not guaranteed to ensure optimal performance}.
However, we do evaluate such effects in Sec.~\ref{sec:evaluation:compare_applications}, where we study the performance of \alg\ in terms of the overall number of {\em bytes} being transmitted (which essentially take into account the effect of increasing message size while doing in-network aggregation).
Our results show that in some cases the decisions made by \alg\ allow it to come close to a {\em lower bound} on the optimal performance possible.

% While we aware of these problems and plan to study them in future work, 
Bearing the above limitations in mind, one should note that practical solutions that address various of the above issues, are already being deployed in real systems and datacenters (e.g., Nvidia's SHARP~\cite{graham20sharp} protocol).
However, we are not aware of any such solutions which  handle bounded in-network computing capabilities as in our model, nor of any solutions that are optimized for multiple workload.
Bridging the gaps between our model and solutions, and real-life deployments, remains a significant challenge.
} % end newrevision

\revision{In the following section we provide the results of our evaluation study of \alg. We defer the formal analysis and proof of optimality of \alg\ to Sec. ~\ref{sec:analysis}.}

\section{Evaluation}
\label{sec:evaluation}
%\gabi{describe the two use cases: MapReduce and ParameterServer}

\begin{figure*}[t!]
    \centering
    \begin{tabular}{cccc}
        % & constant rates ($\weight=1$) & linear increasing rates ($\weight=i$) & exponentially increasing rates ($\weight=2^i$)\\
        \rotatebox[origin=lc]{90}{Power-law load dist.} \hspace{0.1cm} &
        % \subcaptionbox{A \label{fig:WC-MsgByte_subfigA}}{
            \includegraphics[width=0.29\textwidth]{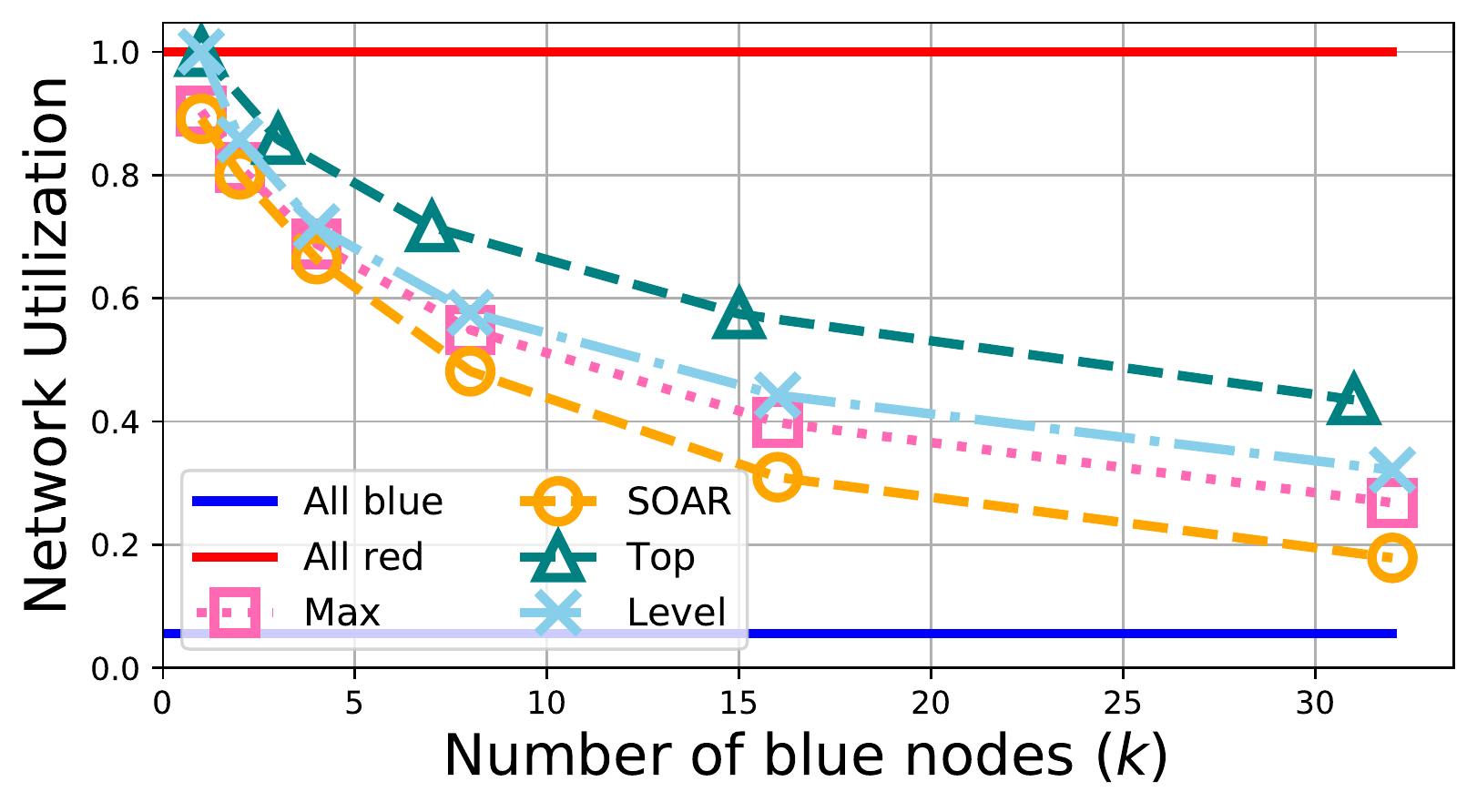} &
            %\includegraphics[width=0.29\textwidth]{Figs/WeightedAlgsComper/singleJob/Wieghted_Util_Multi_wieght_uniform_count_distribution_PowerLaw1_scale_allRed_SOAR_bold.pdf} &
        % }&
        % \subcaptionbox{B \label{fig:WC-MsgByte_subfigB}}{
            \includegraphics[width=0.29\textwidth]{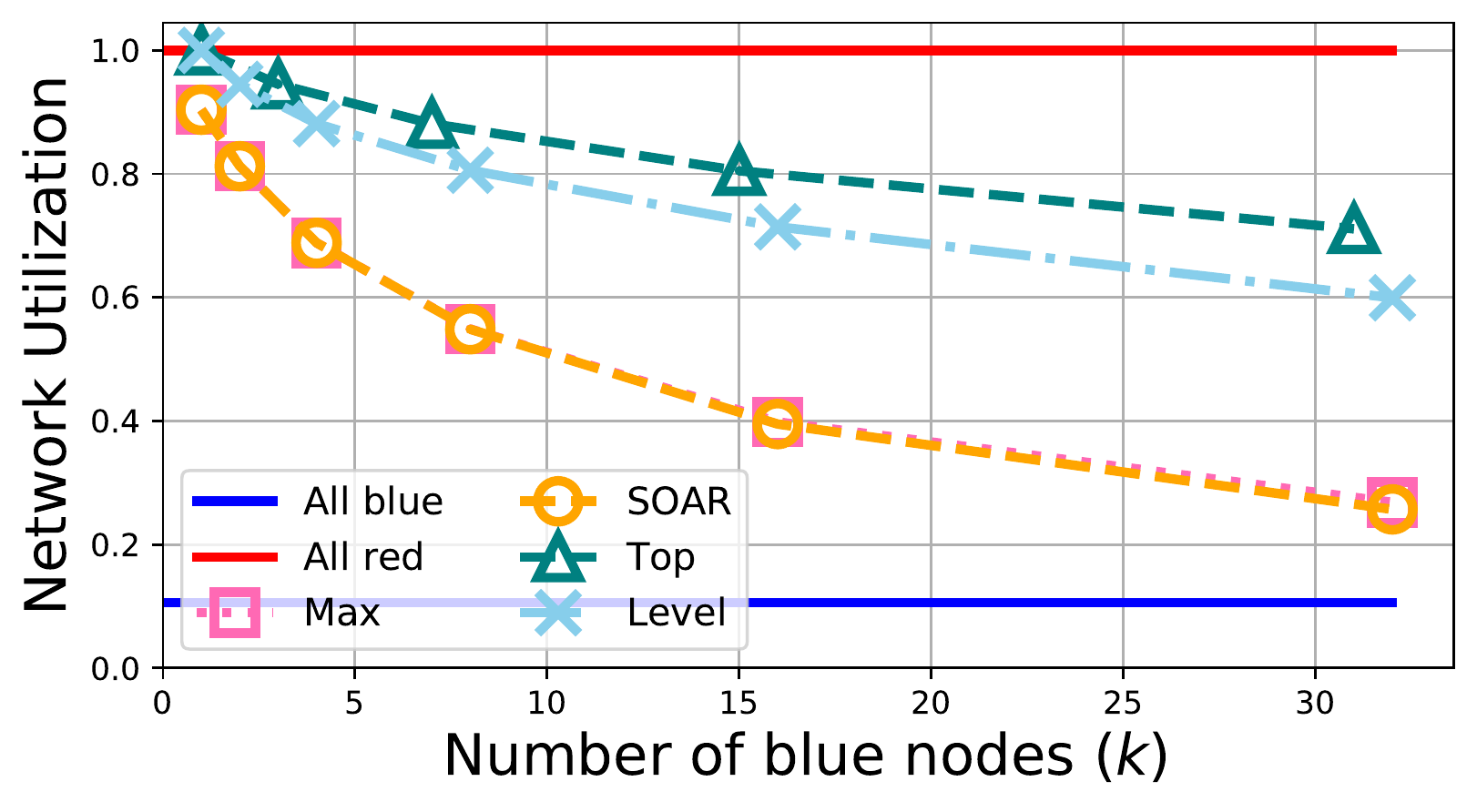} &
        % }&
        % \subcaptionbox{C \label{fig:WC-MsgByte_subfigC}}{
            \includegraphics[width=0.29\textwidth]{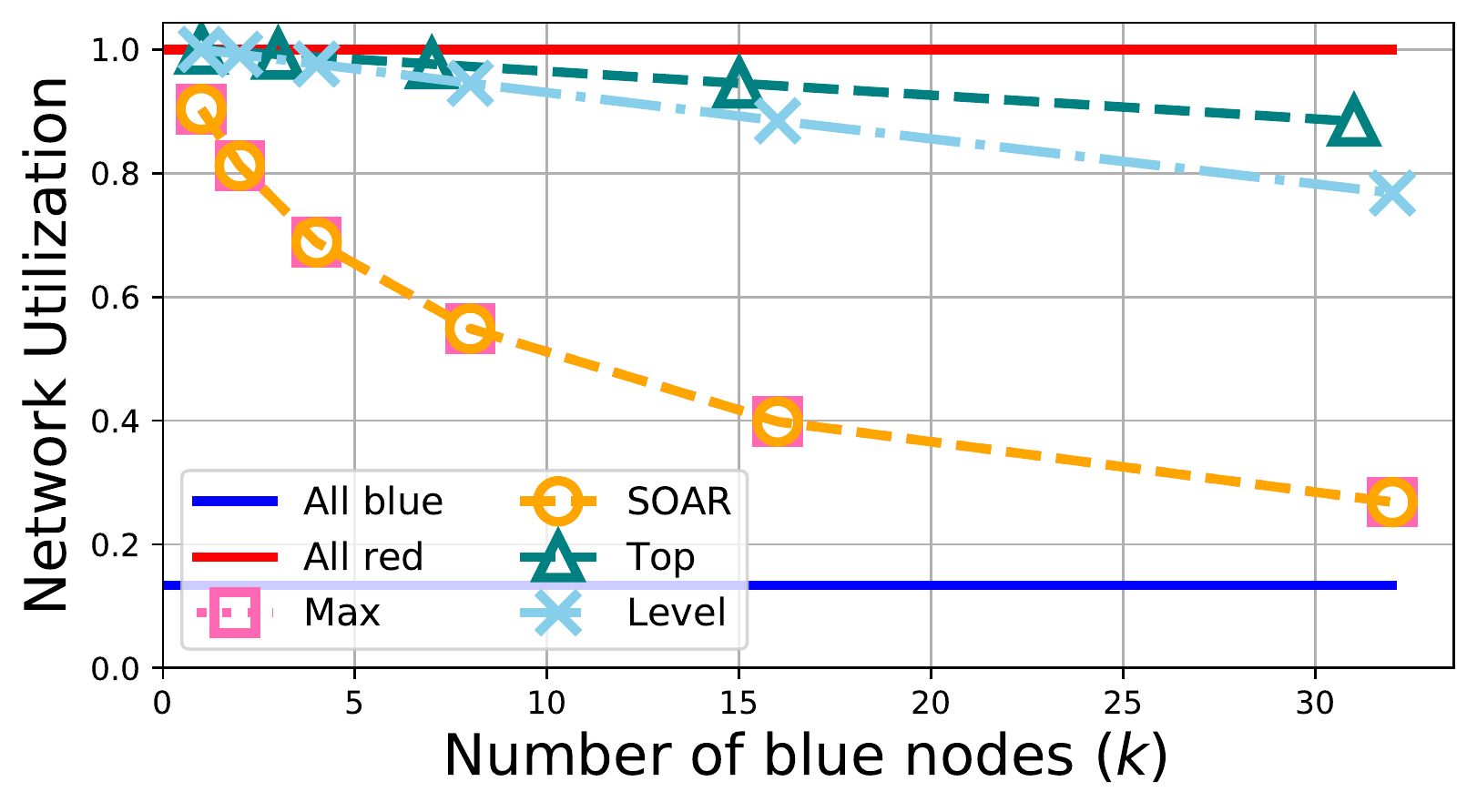} \\
        % }\\
        \rotatebox[origin=lc]{90}{Uniform load dist.} \hspace{0.1cm} &
        \subcaptionbox{constant ($\weight=1$) \label{fig:WC-MsgByte_constant}}{
            \includegraphics[width=0.29\textwidth]{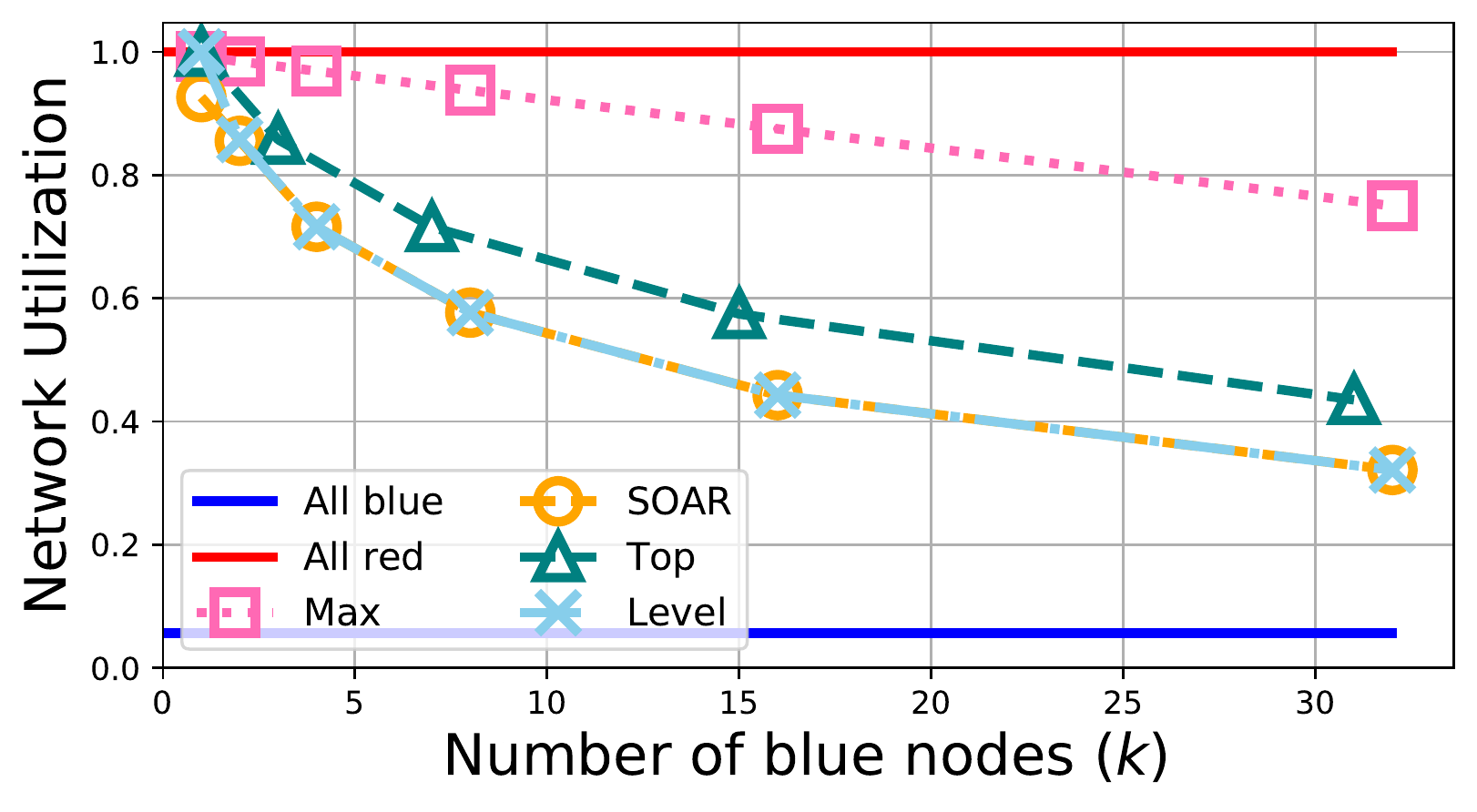}
        }&
        \subcaptionbox{linear increasing ($\weight=i$) \label{fig:WC-MsgByte_linear}}{
            \includegraphics[width=0.29\textwidth]{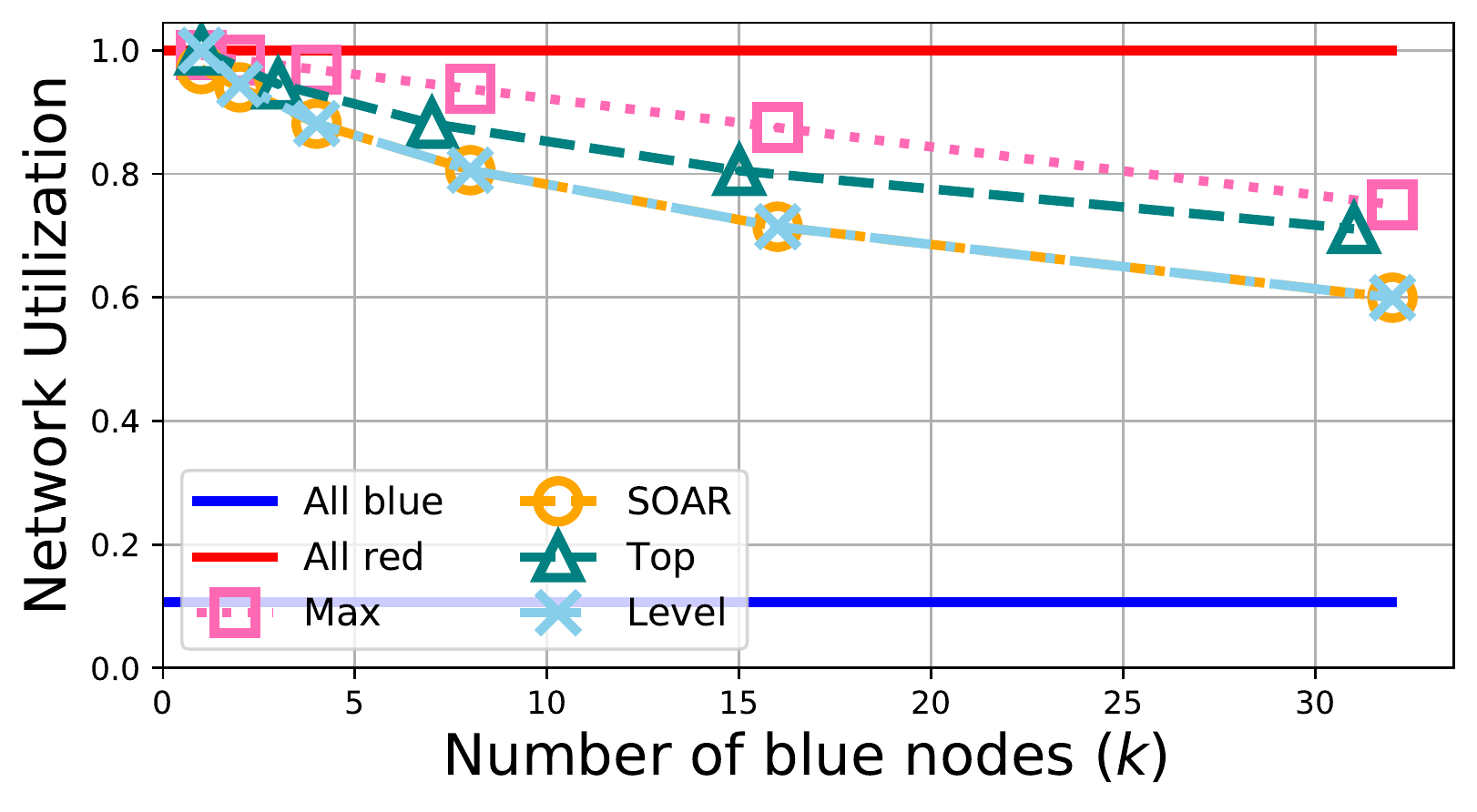}
        }&
        \subcaptionbox{exponentially increasing ($\weight=2^i$) \label{fig:WC-MsgByte_exponential}}{
            \includegraphics[width=0.29\textwidth]{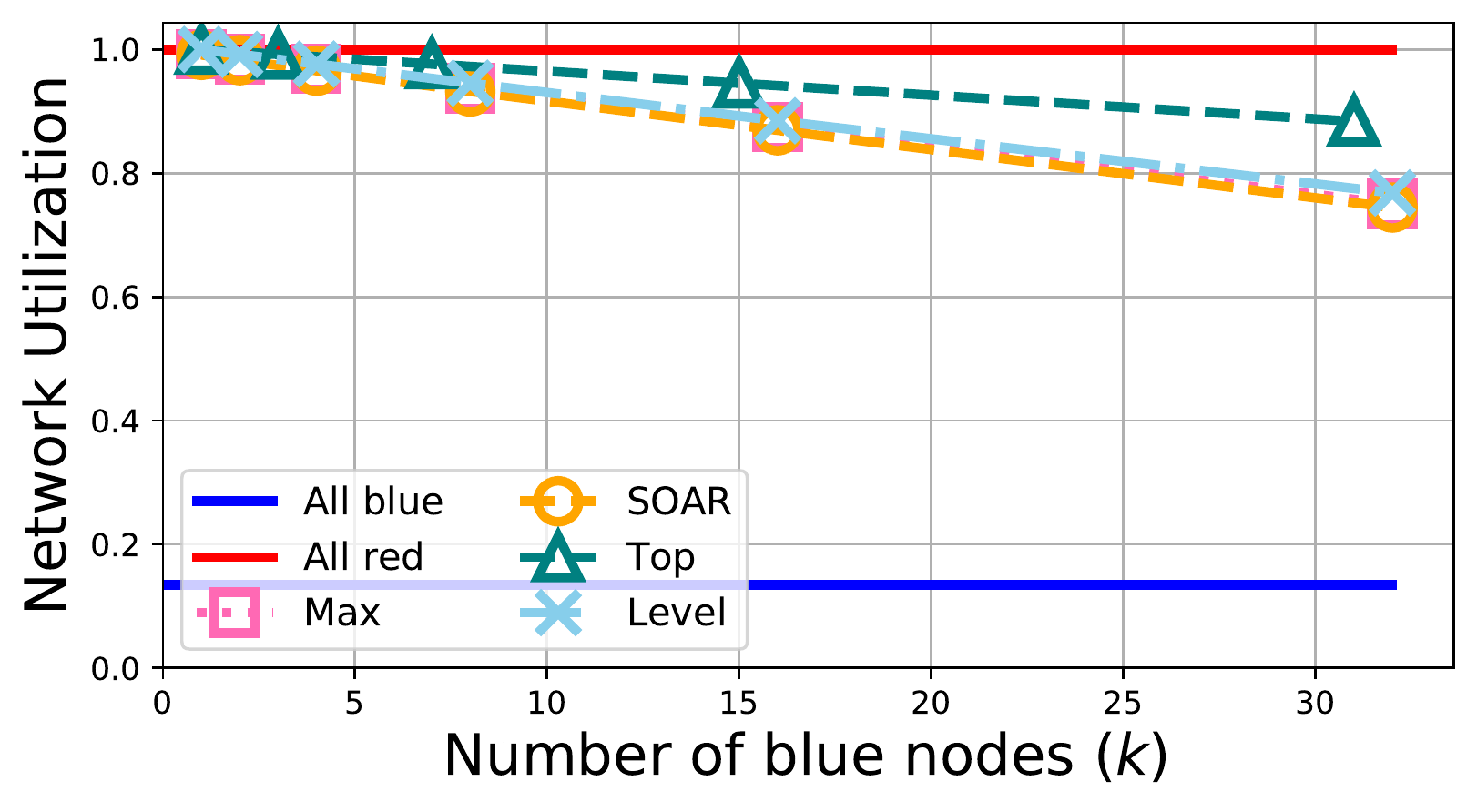}
        }
    \end{tabular}
\caption{\alg\ vs. other strategies for distinct schemes of rates (Fig.~\ref{fig:WC-MsgByte_constant}-\ref{fig:WC-MsgByte_exponential}), and distinct load distributions (power-law in the top plot, uniform in the bottom plot).}
\label{fig:WC-MsgByte}
\end{figure*}

In this section we describe the results of our evaluation of \alg, where we performed extensive simulations which provide further insight as to its performance.
In our evaluation, we examine both the utilization complexity induced by \alg\ (and at times additional contending strategies), and also the {\em byte complexity} which is the actual network load, in bytes, imposed by performing the Reduce operation.

Most of our evaluation makes use of the following system characteristics (unless explicitly stated otherwise).
We consider complete binary weighted trees as the underlying network, denoted by $\btnet{n}$, where $n$ is the number of nodes in the network, \revision{including the destination server.}
We allow non-zero load to be placed only in the leaves of the tree.
These leaves serve as top-of-rack switches connected to servers which generate load, whereas the remaining network serves to model the higher levels of a datacenter network which facilitates the flow of information from the various worker servers, to the destination server which acts as the aggregator.
We consider two distributions for the load at the leaves of the network:
\begin{inparaenum}[(i)]
\item uniform, where the integer load of each node is picked u.a.r. in some range $[x,y]$, and
\item power-law, where the integer load of each node is picked from a power-law distribution. %, and
%\item mixed, where half the nodes have their load drawn from a uniform distribution, and half from the power law distribution.
\end{inparaenum}
The distributions characteristics are as follows; The mean of both distributions is $5$, the  variance is $0.65625$ and $97.1$ for the uniform and the power-law, respectively. The (min, max) values are $(4,6)$ and $(1,63)$ for the uniform and the power-law distributions, respectively. 
% Table~\ref{tbl:distribution_characteristics} provides a summary of the distribution characteristics for the two distributions we consider. \chen{do we need the table?}
We consider three different rate schemes:
\begin{inparaenum}[(i)]
\item constant rates, were all link rates are equal to $1$,
\item linear rates, were $\weight(e)$ increases linearly, by adding $1$, from leaf edges (rate $1$) towards the root, and
\item exponential rates, were $\weight(e)$ increases by doubling (i.e., a power of 2), from leaf edges (rate $1$) towards the root.
\end{inparaenum}

Each experiment was repeated ten times and we present the average performance for each such set of experiments. For clarity we present error bars only where we encountered significant variance in the results.

In most of our results, we present the {\em normalized} performance of an algorithm, where normalization is usually with respect to the {\em all-red} scenario.
This essentially shows the cost reduction of the specific scenario, compared to the {\em all-red} solution.
I.e., 
if the performance of an algorithm is $\alpha \in [0,1]$ in some scenario, this means that the algorithm incurs an $\alpha$ fraction of the cost 
% (either utilization, or bytes) 
of the all-red solution when preforming \reduce in that scenario.

Additionally, we consider two use cases for evaluating the system:
\begin{inparaenum}[(i)]
\item {\em big-data}, using a word-count task~\cite{apache}, where we make use of a wikipedia dump~\cite{wiki}, with an overall of 54M words, out of which 800K are unique. We refer to this use case as the {\em word count (\wcapp)} use case.
\item {\em distributed ML}, using distributed gradient aggregation with a parameter server~\cite{li14scaling}, where worker servers independently perform neural-network training, over a 10K feature space, using 0.5 dropout rate~\cite{srivastava2014dropout}\footnote{We used dropout in order to obtain more diverse network utilization results in terms of bytes, as using all features would render the utilization complexity, and the number of bytes sent, the same.}, and send their updated gradients to a parameter server aggregating the information.\footnote{We note that our work considers solely the network load produced by such tasks, and not the quality of the model produced, which may depend on a variety of problem characteristics. We therefore do not implement the actual neural network, but rather consider the messages sent by the worker servers, and the aggregation of these messages.} We refer to this use case as the {\em parameter server (\psapp)} use case.
\end{inparaenum}

% ==================================================
\subsection{Comparing \alg\ with Other Strategies}
\label{sec:evaluation:compare_algorithms}

In this section we consider the performance of \alg\ compared to the performance of several contending strategies for solving the \bica\ problem.
Specifically, we focus our attention on the simple strategies described in our motivating example in Sec.~\ref{sec:example}, namely,
\begin{inparaenum}[(i)]
\item {\em \topalg},
\item {\em \maxalg}, and
\item {\em \levelalg}.
\end{inparaenum}

Fig. \ref{fig:WC-MsgByte} presents the performance of \alg\ alongside the performance of the contending strategies in distinct rate regimes (subfigures~\ref{fig:WC-MsgByte_constant}-\ref{fig:WC-MsgByte_exponential}), for different workload distribution (top and bottom), using $\btnet{256}$.
We consider distinct values of $\numblue=1,2,4,8,16,32$, and performance is normalized to the all-red strategy.
We further plot the performance of the {\em all-blue} solution for reference.
As would be expected, all strategies exhibit improved performance for increasing values of $\numblue$, which allows for more in-network aggregation, translating to reduced utilization complexity.

\begin{figure*}[t!]
    \centering
    \begin{tabular}{cccc}
        \rotatebox[origin=lc]{90}{Number of workloads} \hspace{0.1cm} &
        \includegraphics[width=0.29\textwidth]{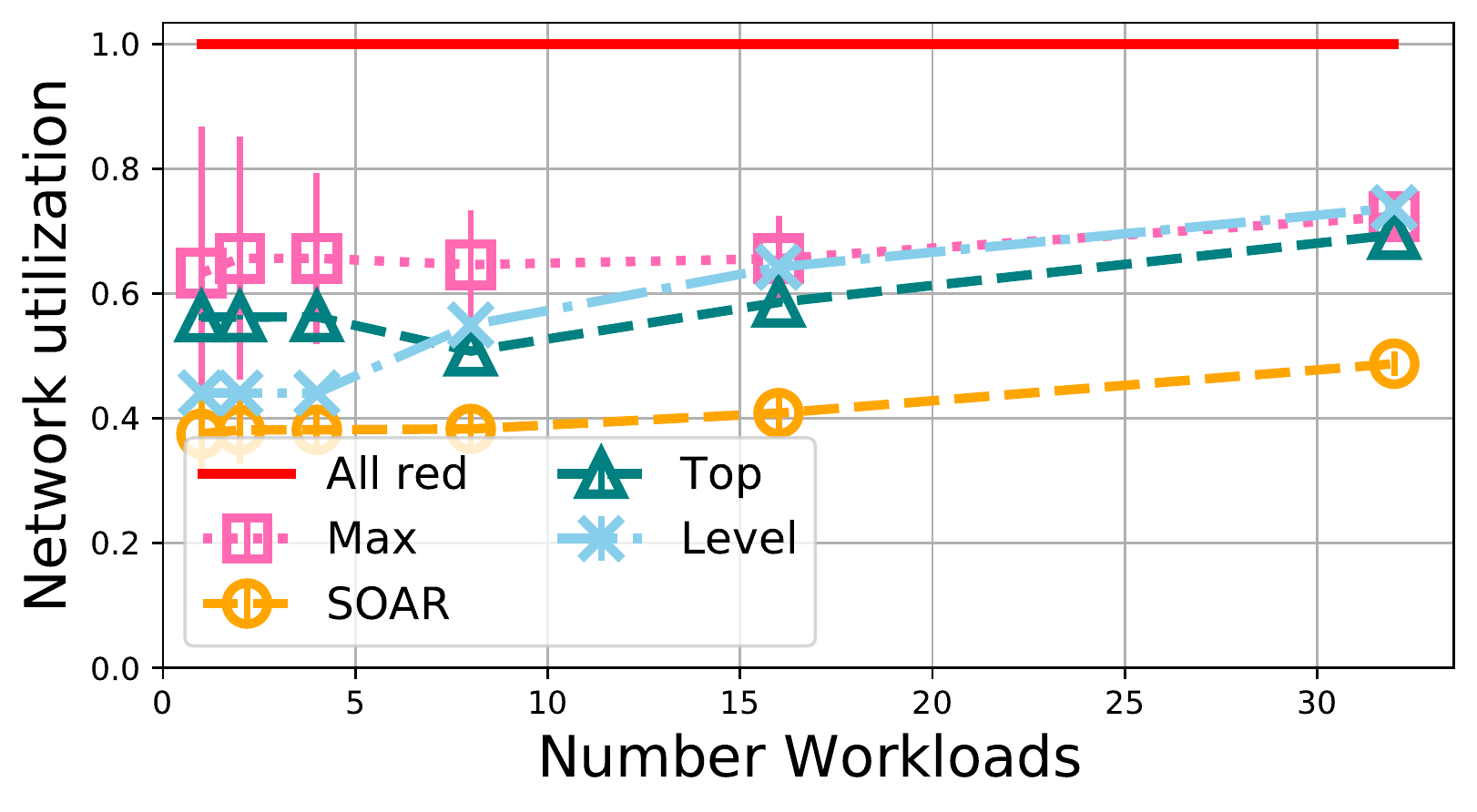} &
        \includegraphics[width=0.29\textwidth]{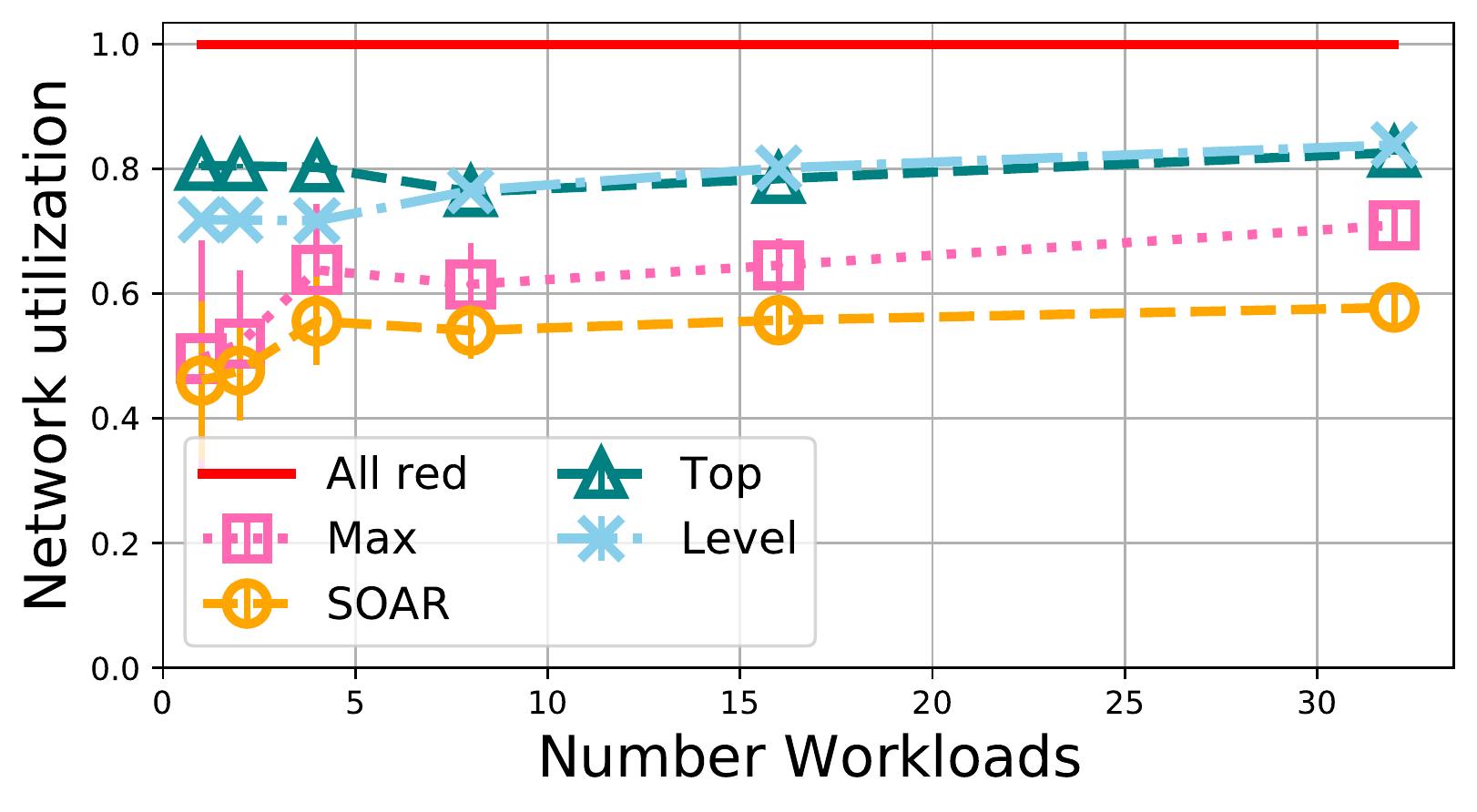}&
        \includegraphics[width=0.29\textwidth]{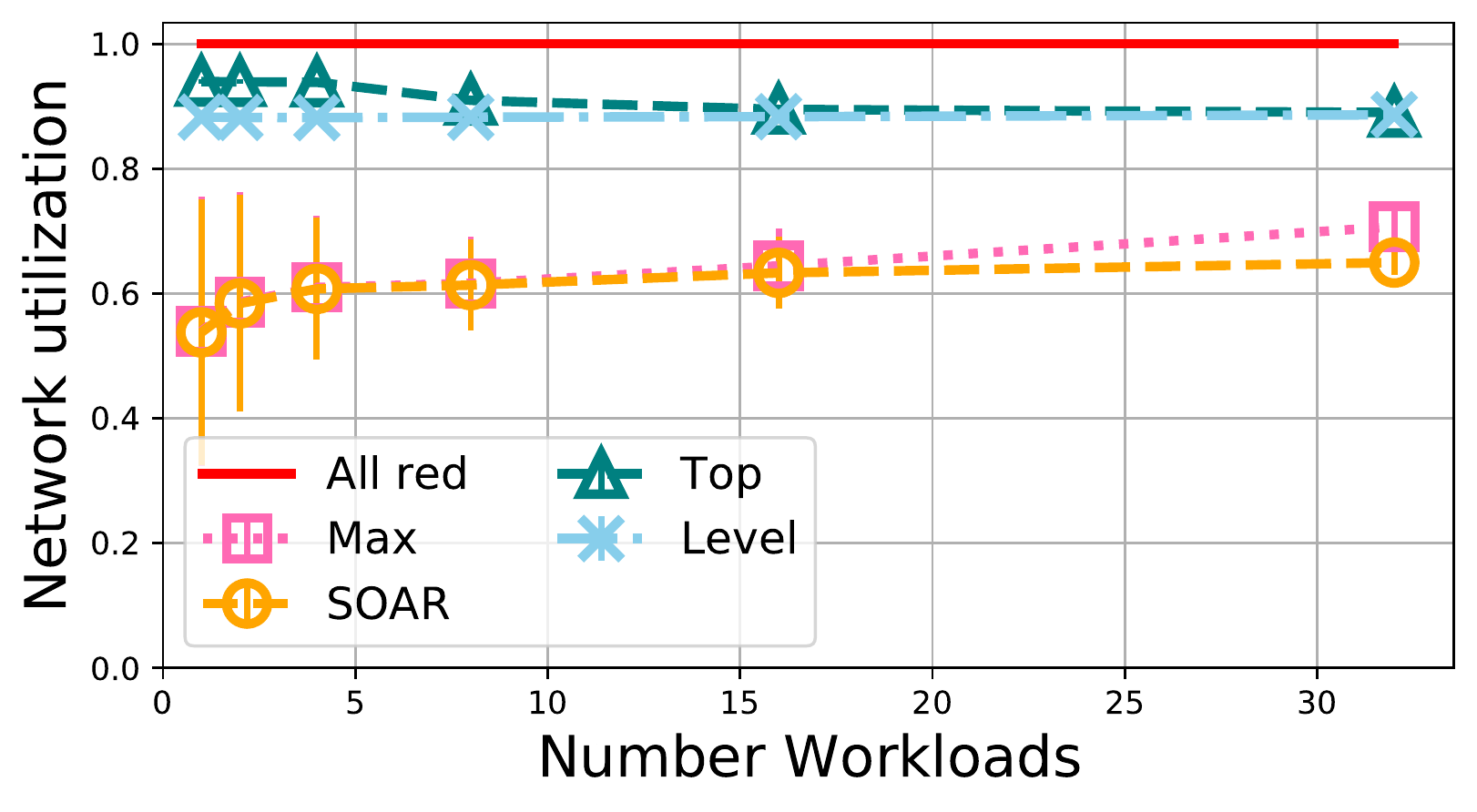}\\
        \rotatebox[origin=lc]{90}{Aggregation capacity} \hspace{0.1cm} &
        \subcaptionbox{constant ($\weight=1$) \label{fig:multiple_workloads_constant}}{
            \includegraphics[width=0.29\textwidth]{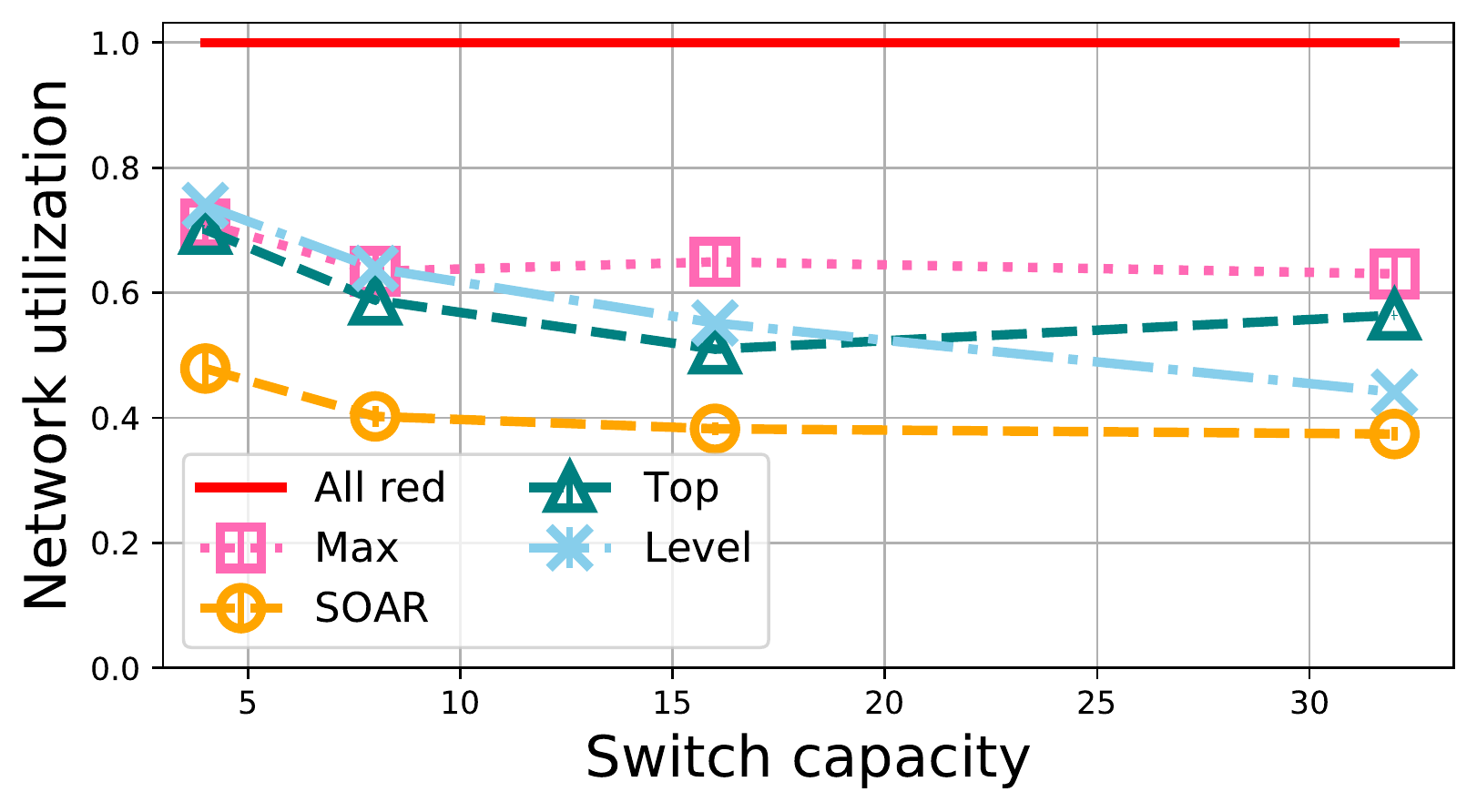}
            }&
        \subcaptionbox{linearly increasing ($\weight=i$) \label{fig:multiple_workloads_linear}}{
            \includegraphics[width=0.29\textwidth]{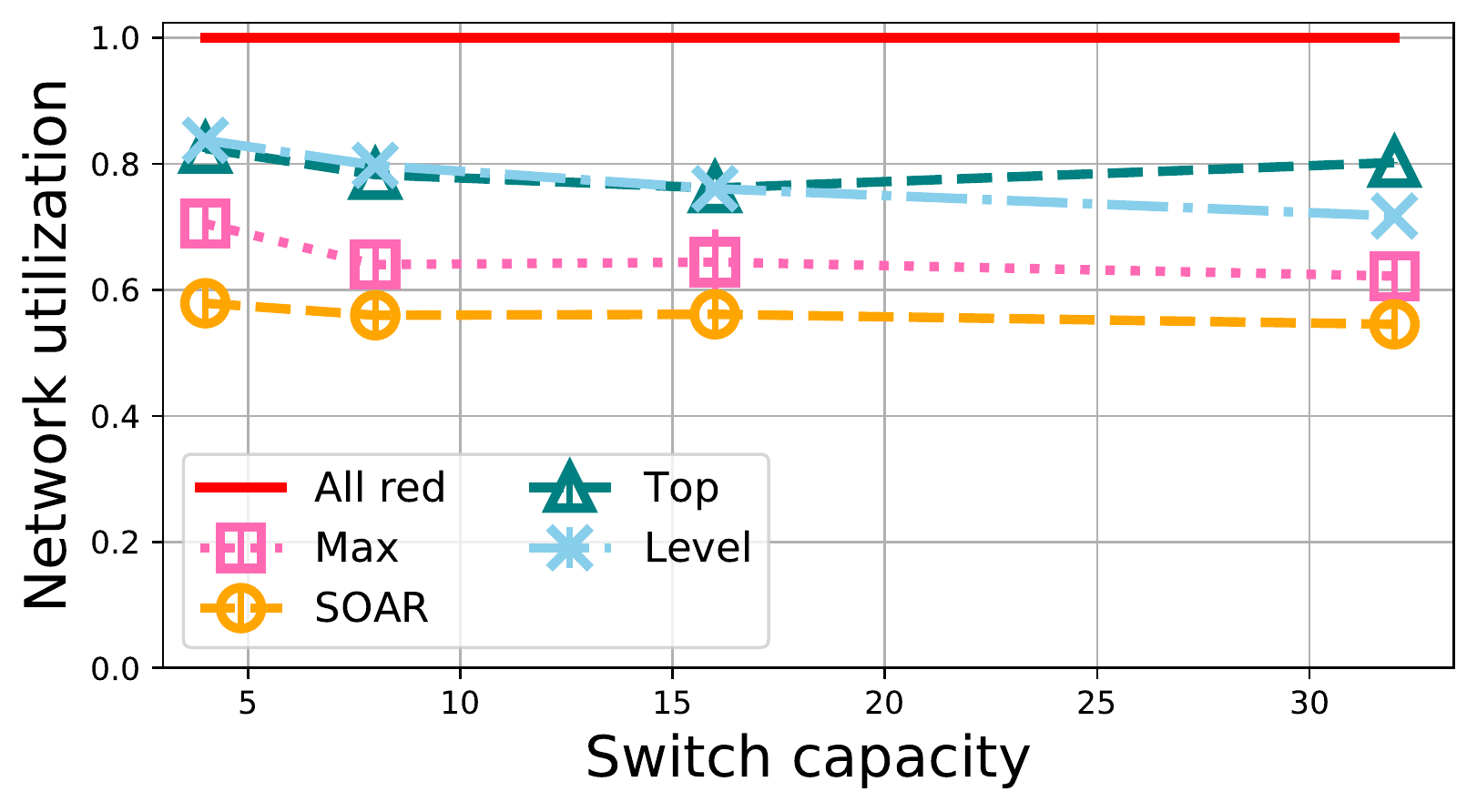}
            }&
        \subcaptionbox{exponentially increasing ($\weight=2^i$) \label{fig:multiple_workloads_exponential}}{
            \includegraphics[width=0.29\textwidth]{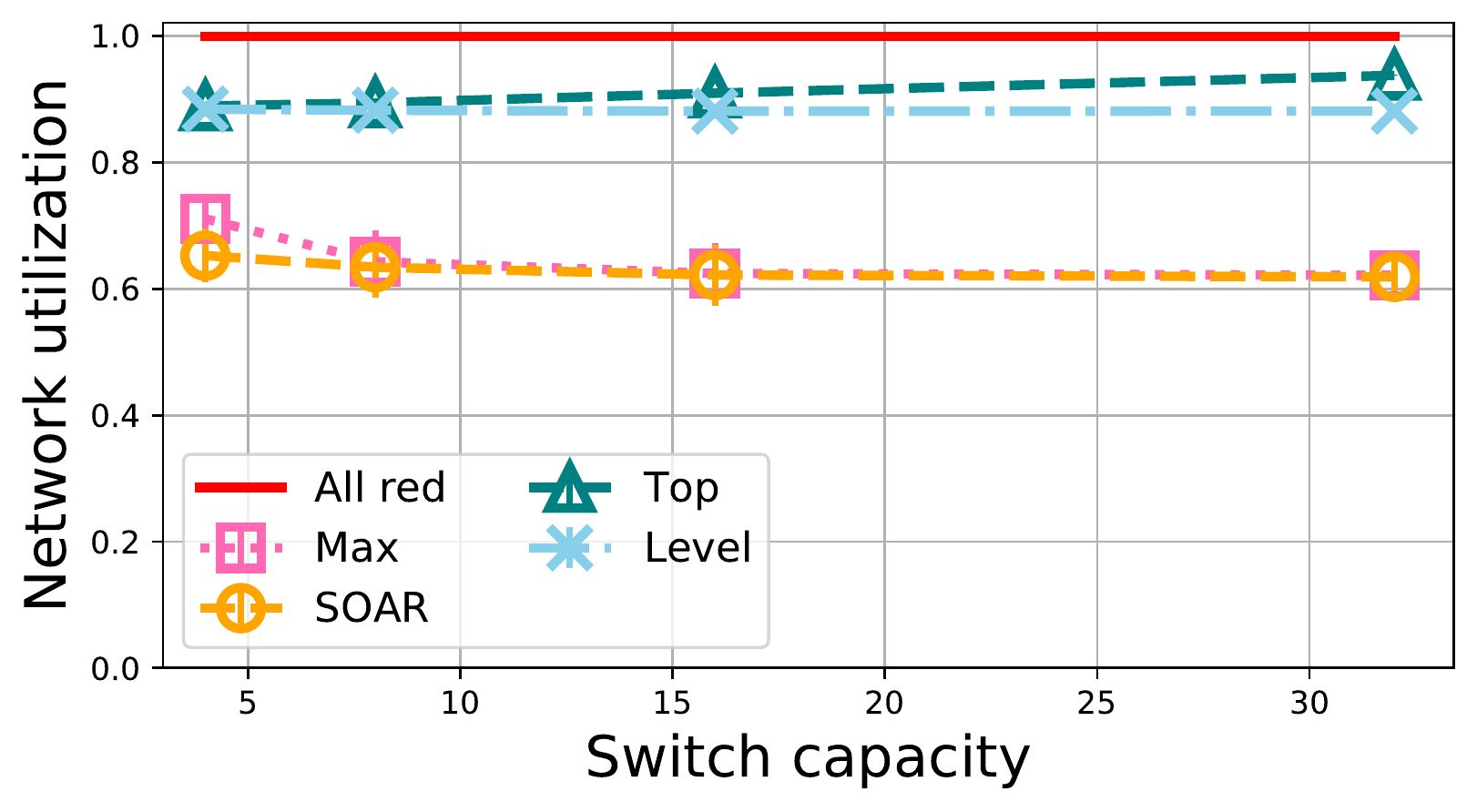}
            }\\
    \end{tabular}
\caption{\alg\ vs. other strategies when aggregating multiple workloads online. Subfigures represent distinct schemes of rates (Fig.~\ref{fig:multiple_workloads_constant}-\ref{fig:multiple_workloads_exponential}), and the effect of increasing the number of workloads (for switch capacity 4, top plots), and increasing the aggregation capacity (for 32 workloads, bottom plots).}
\label{fig:multiple_workloads}
\end{figure*}

Since \alg\ is optimal, it exhibits the best performance in all scenarios. This serves to show that using \alg\ ensures robustness regardless of load distribution or link rates.
However, the second-best strategy strongly depends on the load distribution, and the link rates.
The power-law load distribution favors the \maxalg\ strategy, since high-load leaf-switches that perform aggregation induce a significant reduction in overall utilization complexity.
For the uniform distribution, however, the \levelalg\ strategy fares best, since it implies load balancing the uniform loads at the leaf-switches throughout the network. For such scenarios, the \levelalg\ strategy essentially mimics the ``barrier'' approach underlying the design of \alg, as described in Sec.~\ref{sec:alg:barrier}.
The \topalg\ strategy is the most sensitive to the link rates, where having higher rates towards the root of the network implies that performing in-network aggregation further up
% in the network 
provides far lesser benefits than doing so closer to the leaves (or closer to the middle of the network).

\noindent \newrevision{\textbf{Takeaways}: \alg\ can significantly outperform other strategies across different workloads and link rates functions. A small fraction of nodes with in-network processing capabilities is enough to reduce network utilization substantially.}

\subsection{Multiple Workloads}

In this subsection we consider the problem of handling multiple workloads, and determining where aggregation should take place for each such workload.
Each workload is determined by its $\load_t$, $t=0,1,2,\ldots$.
We consider the workloads as arriving in an {\em online} fashion, such that determining the aggregating switches for workload $\load_t$ should be settled before handling workload $\load_{t+1}$.

We further assume each switch $\switch$ has a predetermined {\em aggregation capacity} $\aggcap(\switch)$ which bounds the number of workloads for which $\switch$ can be assigned as an aggregating switch.
We let $\aggcap_t(\switch)$ denote the {\em residual aggregation capacity} available at $\switch$ before handling workload $\load_t$.
If switch $\switch$ is designated as an aggregation switch when handling workload $\load_t$, then $\aggcap_{t+1}(\switch)=\aggcap_t(\switch)-1$, and $\aggcap_{t+1}(\switch)=\aggcap_t(\switch)$ otherwise.

We consider the performance of the various strategies used in Sec.~\ref{sec:evaluation:compare_algorithms},
%\raz{not sure this is the right reference,it should be \ref{sec:example}},
when applied repeatedly to the sequence of workloads $\load_0, \load_1, \ldots$, given as input.
The set of switches available for aggregation when handling workload $\load_t$ is defined by $\Avilabilty_t=\set{\switch \mid \aggcap_t(\switch) > 0}$.

We generate our sequence of workloads in an online fashion, by drawing each workload from either the uniform load distribution, or the power-law load distribution, each with probability 1/2.

In our evaluation we consider the effect of varying the aggregation capacity, and the number of workloads.
As a baseline we consider the topology $\btnet{256}$, with $k=16$, $\aggcap(\switch)=4$ for every switch $\switch$, and 32 workloads.
Fig.~\ref{fig:multiple_workloads} shows the performance of \alg\ compared to the performance of the various strategies described in Sec.~\ref{sec:example}.
Similarly to our results presented in Sec.~\ref{sec:evaluation:compare_algorithms}, our evaluation considers 3 scaling laws for link rates: constant (in Fig.~\ref{fig:multiple_workloads_constant}), linearly increasing (in Fig.~\ref{fig:multiple_workloads_linear}), and exponentially increasing (in Fig.~\ref{fig:multiple_workloads_exponential}).

\begin{figure*}
    \centering
    \begin{tabular}{ccc}
        \subcaptionbox{Utilization \label{fig:apps_bytes_utilization}}{
            \includegraphics[width=0.31\textwidth]{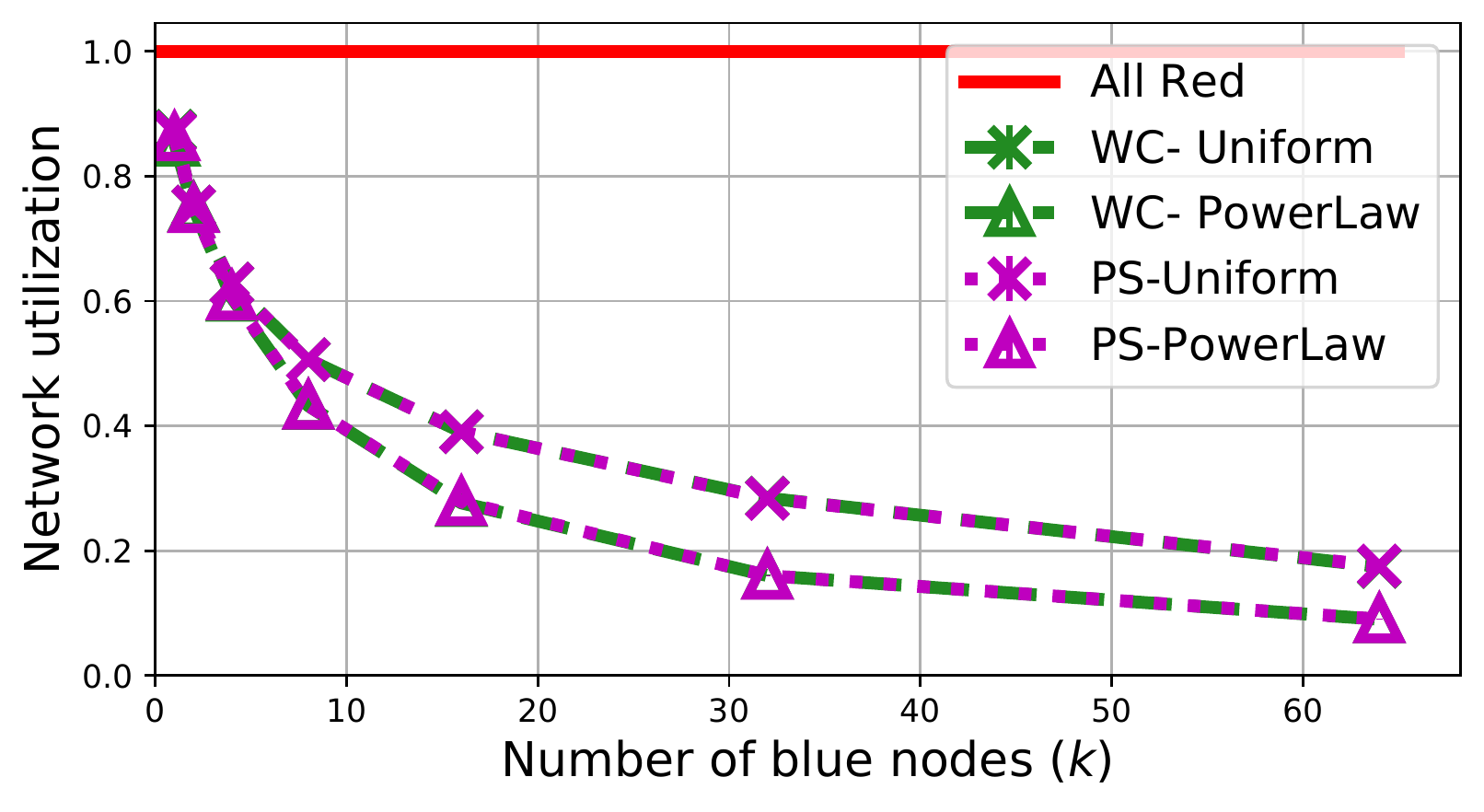}
            } &
        \subcaptionbox{Bytes (normalized to all-red) \label{fig:apps_bytes_red}}{
            \includegraphics[width=0.31\textwidth]{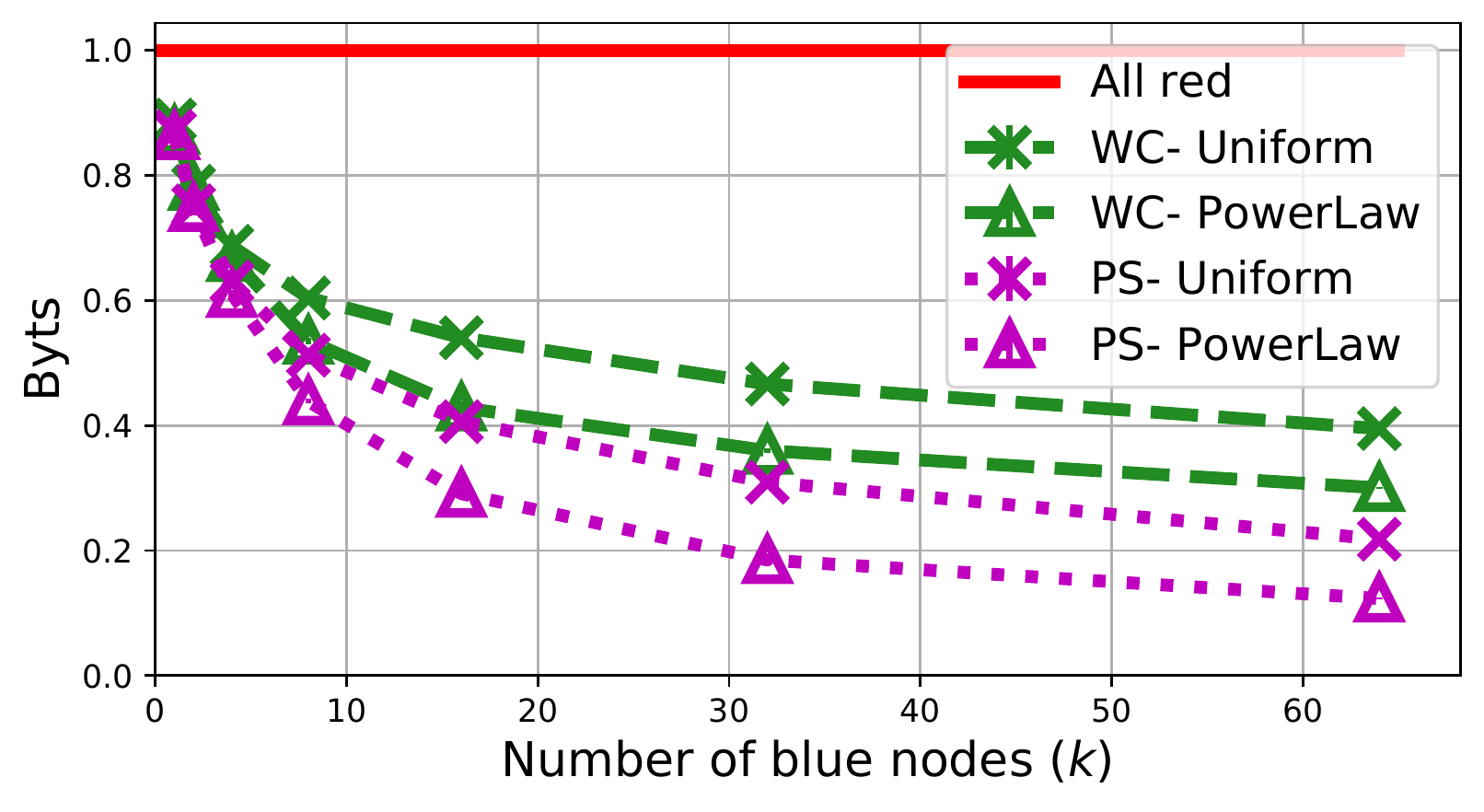}
            } &
        \subcaptionbox{Bytes (normalized to all-blue) \label{fig:apps_bytes_blue}}{
            \includegraphics[width=0.31\textwidth]{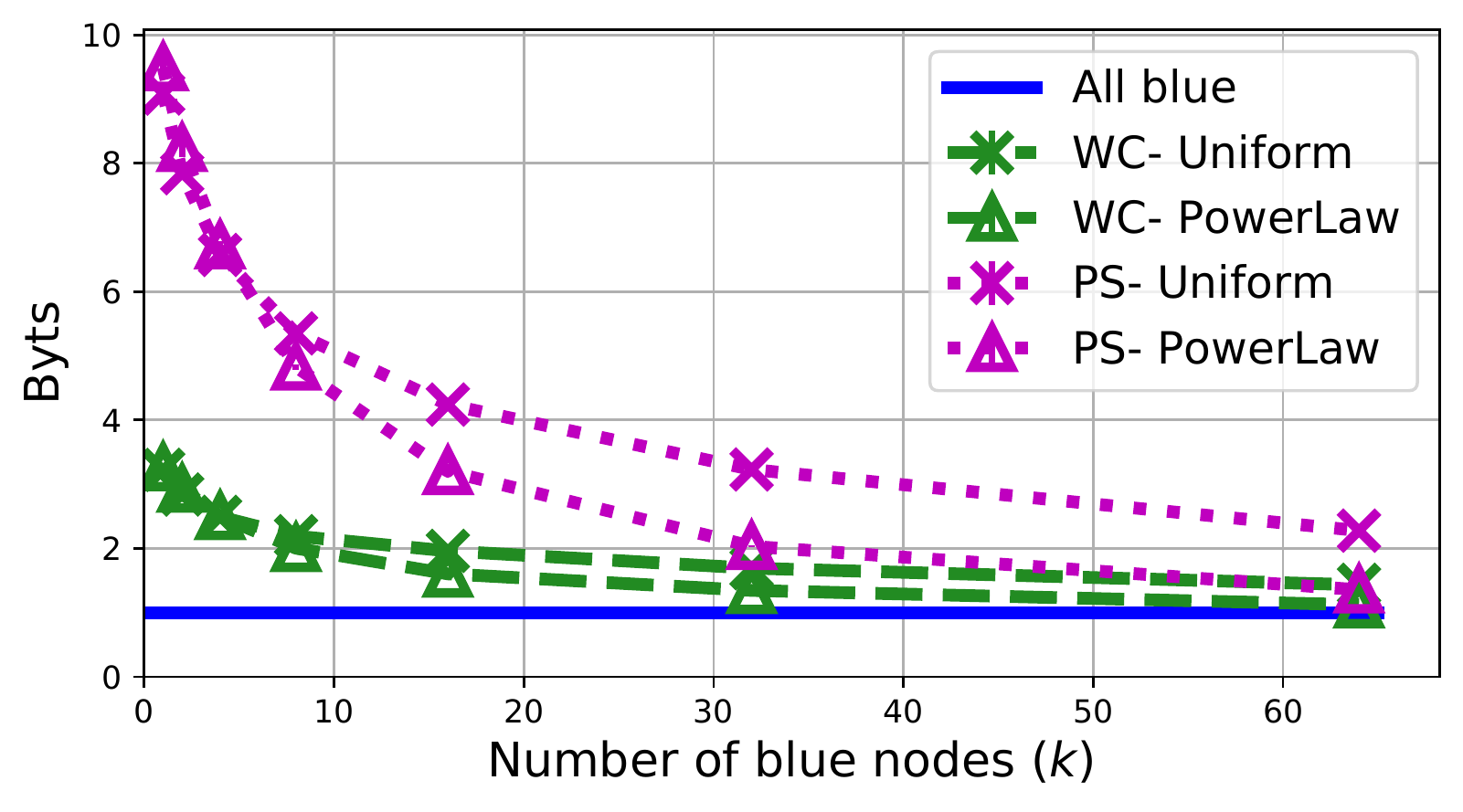}
            }
    \end{tabular}
\caption{\alg\ performance for the \wcapp\ and \psapp\ use cases.}
\label{fig:apps}
\end{figure*}

When considering the effect of handling more workloads (Top plot in each column), the normalized utilization ratio (compared to that of the all-red solution) tends to increase as we handle more workloads, and the improved performance demonstrated by \alg\ compared to the best contending strategy is more pronounced as the weight differences across layers are smaller.
We note that as the number of workloads increases, the performance would converge to that of the all-red configuration, regardless of the strategy being used.
This follows from the fact that the aggregation capacity is bounded, implying that once the number of workloads is large enough, further workloads cannot benefit from any aggregation, and the initial benefits of aggregating the prefix of the workload arrival sequence become marginal compared to the toll imposed by the entire sequence.

It is instructive to note that the second-best strategy varies significantly, where for exponentially increasing rates the performance of the Max strategy is closest to that of \alg\, while for constant rates either the performance of the Level strategy or the Top strategy come closest to that of \alg.

When considering the effect of increasing the aggregation capacity at each switch, one can see that most strategies exhibit improved performance as the aggregation capacity increases, and \alg\ exhibits the best performance across all scenarios, where the differences are again more pronounced as the differences in rates across levels is smaller.
An exception to this performance is exhibited by the Top strategy, which actually fares worse as aggregation capacity increases. This is due to the fact that the larger capacity enables the strategy to handle more workloads closer to the root, which accentuates its sub-optimality.
Finally, we note that when aggregation capacity is unbounded, \alg\ would produce the optimal solution possible (for any given $k$) {\em even in the online setting}, since it is optimal for every workload, and workloads are handled separately and independently. % by any strategy.

\noindent \newrevision{\textbf{Takeaway}, \alg\ exhibits the best performance compared to other strategies in the online settings (although it is not proven to be optimal). Furthermore, for small switch capacity and many workloads, \alg\  obtain more considerable gains.
}
\subsection{\alg\ for Different Applications}
\label{sec:evaluation:compare_applications}

We now turn to consider the performance of \alg\ for distinct use cases, namely, \wcapp, and \psapp.
We focus our attention on the case of constant rates, which better emphasizes the differences in the performance.
We distinguish between our utilization metric (which in the constant rate case is equivalent to the number of messages traversing the network), and the {\em byte complexity}, where we take into account the actual message size, and evaluate the overall number of bytes being transmitted throughout the network over all links.
We note that our problem formulation, and our algorithm, do not target minimizing the byte complexity.

Fig.~\ref{fig:apps} shows the results of our evaluation for the two use cases, in the $\btnet{256}$ topology, where we consider
% the performance of \alg\ on two different applications: word-count, \wcapp, and parameter server, \psapp.
% We present the result for $\btnet{256}$ and 
both the uniform and the power-law load distributions, and the results are normalized to the all-red scenario.
Not surprisingly, the network utilization of both use cases is {\em independent} of the specific reduce task being performed, as can be seen in Fig.~\ref{fig:apps_bytes_utilization}.
This is due to the fact that our model, and in turn, \alg, do not distinguish between the concrete details of the use case, and considers all messages as equal.
As could be expected, the load distribution does bear an effect of the utilization, where the performance of \alg\ improves as the distribution is more skewed (as is the case for the power-law distribution).
This is attributed to the fact that for highly asymmetric load distributions, \alg\ identifies the key points with severe load, and places blue nodes at (or close to) such points.
On the other extreme, as the load is more evenly distributed, the judicious choices made by the algorithm have a lesser effect on overall utilization.

Fig.~\ref{fig:apps_bytes_red} present the normalized cost reduction in terms of the byte complexity, where normalization is again done compared to the all-red scenario.
We see that both the load distribution, and the actual application use case, affect the performance.
The byte complexity in the \psapp\ use case is very similar to the utilization. 
This is due to the fact that we are using a non-negligible dropout rate of 0.5, as is mostly advised in distributed ML.
% since we consider a standard dropout (for distributed ML) of $50\%$,
For this case the sizes of messages traversing distinct links in the network do not vary significantly, and message sizes increase very mildly as we approach the root of the network.
% are similar and large and the message size growing slowly as we clime the tree. In word-count the message size...
For the \wcapp\ use case, the effect of increasing message sizes is more pronounced (as also discussed in Sec.~\ref{sec:evaluation:compare_algorithms}), leading to diminished improvement in terms of byte complexity, when compared to the utilization. However, the general trends across distributions are still apparent.
%\raz{NOTE:This is not write ant more, we don't mention \wcapp\ at this sec.}

Lastly, Fig.~\ref{fig:apps_bytes_blue} shows the effect of having more blue nodes, when compared with the all-blue solution.
These results highlight the effect of the message sizes on byte complexity, where for the \wcapp\ use case the performance of \alg\ comes very close to that of the all-blue solution, already when using but a few blue nodes.
In contrast, for the \psapp\ use case the byte complexity is very closely related to the utilization complexity (as message sizes do not vary significantly).
This is manifested by the fact that significantly more blue nodes are required in order to come close to the performance of the all-blue scenario.
As distributed ML environments become ubiquitous, we believe that our proposed algorithm for data aggregation within the network can have a significant impact on the performance of such systems.
% \revision{We believe that, with growing needs for fast and efficient distributed ML, the last observation is of significant impact.} 
Overall, our results indicate that although message sizes do affect the byte complexity beyond the effects manifested by the network utilization, the ability to determine the optimal location of a bounded number of blue nodes, as done by our algorithm, indeed results in performance that quickly comes close to that obtained by an unbounded solution.

\noindent \newrevision{\textbf{Takeaways:} Minimizing the network utilization reduces significantly also the byte count. The effect of using in-network processing can differ across different applications (e.g., \wcapp\ and \psapp).}

\begin{figure}
    \centering
    \includegraphics[width=0.8\columnwidth]{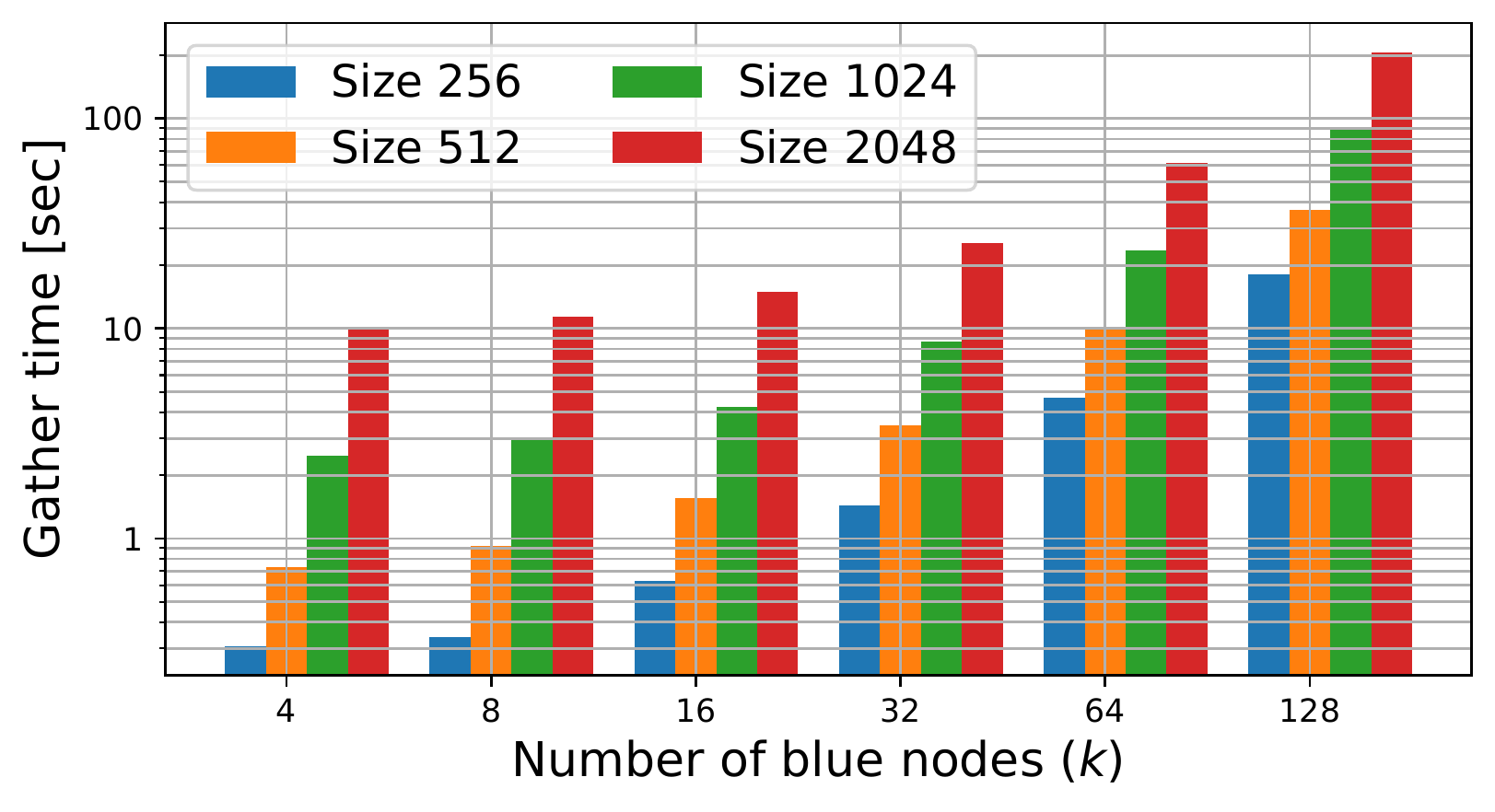}\\
    \caption{\alg\ running time}
    
    \label{fig:run time}
\end{figure}
% ==================================================
\revision{
\subsection{\alg\ Run-Time Evaluation}
In this section we evaluate the running-times of \alg, \alggather\ and \algcolor.
We implemented our simulation in python 3.8 and the evaluation was done on a Laptop equipped with an Intel core i7(10875H) CPU and 32GB of RAM.

When measuring the running-time of \alggather\ and \algcolor\ we conclude that the running time of \algcolor\ negligible comped to \alggather, as its operation is much simpler.
The running time  of \algcolor\ is faster by three orders of magnitude, and while \alggather\ runs in seconds, \algcolor\ runs in milliseconds.

In Fig~\ref{fig:run time} we present the average running time of \alggather\ over ten experiments for different network sizes and values $\numblue$.
The figure shows the running time in seconds, for $\numblue= 4,8,16,32,64$ and $128$, and for network sizes ,$256,512,1024$ and $2048$.
Following Theorem \ref{thm:alg_is_optimal} we can observe that the running time is indeed quadratic in $k$ and close to linear in $n$ (where our results prove an upper bound of $n \log  n$).
We note Fig.~\ref{fig:run time} is in $\log$-$\log$ scale. 
%The results show the running-time increase with number of $\numblue$, in a expediently manner.
%There one can implement in more optimized manner the trend will stay the same. How ever there is room for improvement in our implementation that may will reduce the running-times of the algorithm. 
%\textbf{Take away} from this section is that \alg\ can be run in the deployment stage of the job, with out adding a significant delay.}
}

\newrevision{
We highlight that the above results are applicable to a {\em serial and centralized} implementation of \alggather\ on a single host, which may require up to a few minutes for a relatively large network (2048 switches) and many aggregating nodes (128). For moderately sized networks, or when deplyoing far less switches, the running time of \alggather\ is on the order of tens of seconds, or less.
We further note that \alggather\ can also be implemented in a {\em parallel or distributed} manner (along a parallel DFS-scan from leaves to the root), which would result in a significant speedup, while requiring more computing power to be used in parallel.
We leave this topic for future work.} % end newrevision

% ==================================================
\section{Analysis of \alg}
\label{sec:analysis}

\subsection{Notation and Definitions}

We begin by introducing some notation that would be used throughout our proofs.
For every node $v$, we let $c_1,\ldots,c_{C(v)}$ denote the children of $v$ (in some arbitrary fixed order).
For every $m=1,\ldots,\childnum(v)$ we let $\network_v^m$ denote the subtree rooted at $v$ containing only the subtrees rooted at children $c_1,\ldots,c_m$, and let $\tilde{\network}_v^m$ denote the {\em extended subtree} of $\network_v^m$, which is extended by adding the link $(v,\parent(v))$.
We further let $\network_v=\network_v^{\childnum(v)}$ denote the subtree rooted at $v$ (containing all subtrees of all children of $v$), and let $\tilde{\network}_v$ be the extended subtree of $\network_v$.
For a node $v$ and $\ell \le \distance(v)$, let $A_v^\ell$ be the ancestor at distance $\ell$ from $v$.

For every node $v$ and every $m=1,\ldots,\childnum(v)$, given any $\ell=0,\ldots,\distance(v)$ and any set of blue nodes $\blueset \subseteq \network_v^m$, we consider the {\em $(v,m)$-potential} of $\ell$ and $\blueset$, $\aggc_v^m(\ell,\blueset)$, defined by
\begin{align}
\label{eq:potential:non_leaf}
\aggc_v^m(\ell,\blueset)
&= \left( \sum_{\link \in \network_v^m} \msg_{\link}(\network_v^m,\load,\blueset) \cdot \rate(\link) \right) \notag \\*
% the \\* prohibits a page break in this line
& \quad\quad + \msg_{(v,\parent(v))}(\tilde{\network}_v^m,\load,\blueset) \cdot \rate(v,A_v^\ell).
\end{align}
and we further use $\aggc_v(\ell,\blueset)$ to denote $\aggc_v^{C(v)}(\ell,\blueset)$.
We note that the $(v,m)$-potential is only defined for $\blueset \subseteq \network_v^m$. For ease of notation we will omit this explicit requirement in the remainder of this section.
Furthermore, if $v$ is a leaf we take $\childnum(v)=0$ which results in having
\begin{align}
\label{eq:potential:leaf}
\aggc_v^0(\ell,\blueset)= \msg_{(v,\parent(v))}(\tilde{\network}_v^0,\load,\blueset)\cdot \rate(v,A_v^\ell),
\end{align}

Lastly, we note that by the definition of $\aggc_v^m$ in Eq.~\eqref{eq:potential:non_leaf}, and the definition of $\msgcost$ in Eq.~\eqref{eq:msgcost_def} we have
% for every $\ell$,
\begin{equation}
\label{eq:phi_is_pi_at_destination}
\msgcost(\network,\load,\blueset)=\aggc_d(0,\blueset)=\aggc_r(1,\blueset).
\end{equation}
% since there is no outgoing link from destination $d$.
A key property of our potential function, that follows from the above definitions, is that, {\em conditioning} on the color of $v$, we can decompose the calculation of $\aggc_v^m(\ell,\blueset)$ into calculating $\aggc$ for smaller problems, as depicted by the following lemma,
\revision{Proof appears in Appendix~\ref{sec:appendix:proofs}.}
%Appendix~\ref{sec:appendix:proofs}.
\begin{lemma}
\label{lem:pidecompose}
If $v\in \network$ is a non leaf node and $\blueset$ is the set of blue nodes, then
\begin{align}
\aggc_v(\ell,\blueset) &= \sum_{m=1}^{C(v)} \aggc_{c_m}(1,\blueset) + 1 \cdot \rate(v,A_v^\ell)
& \mbox{if $v \in \blueset$}
\label{eq:ellBlue} \\
\aggc_v(\ell,\blueset) &= \sum_{m=1}^{C(v)} \aggc_{c_m}(\ell+1,\blueset) + \rate(v,A_v^\ell) \cdot \load(v)
& \mbox{if $v \notin \blueset$}.
\label{eq:ellRed}
\end{align}
\end{lemma}

\subsection{Optimality of \alg}

We first consider algorithm \alggather, which gathers the required information about {\em potentially optimal} configurations.
The algorithm essentially scans the tree from the leaves towards the root, varifying in every node along the scan the overall message cost of the subtree rooted at that node, assuming it would end up allocating $i$ blue nodes within the subtree, for all possible values of $i$.
In doing so, the algorithm evaluates these costs while taking into account all the possible distances of the node being considered (captured by parameter $\ell$) from the nearest blue ancestor or $d$.

The intuition behind this approach follows Lemma~\ref{lem:decompose} and the partitioning-view of the solution that would eventually be produced;
We begin our analysis of \alg\ by showing the following lemma, which serves as the main technical tool for proving the correctness of \alg.
% The complete proof appears in Appendix~\ref{appendix:proof:gather_correctness:induction}.

\begin{lemma}
\label{lem:gather_correctness:induction}
For every node $v$, every $m=1,\ldots,\childnum(v)$, every $\ell=0,\ldots,\distance(v)$, and every $i=0,\ldots,k$,
if $v$ is not a leaf then $\dpellicol_v^m$ as computed by \alggather\ satisfies
\begin{align}
\dpellicol_v^m(\ell,i,\red)
&= \min_{\abs{\blueset}=i,v \notin \blueset} \aggc_v^m(\ell,\blueset) 
\label{eq:lem:alggather:correctness:red} \\
\dpellicol_v^m(\ell,i,\blue)
&= \min_{\abs{\blueset}=i,v \in \blueset} \aggc_v^m(\ell,\blueset).
\label{eq:lem:alggather:correctness:blue}
\end{align}
Furthermore, for every node $v$, every $\ell=0,\ldots,\distance(v)$, and every $i=0,\ldots,k$,
$\dpelli_v(\ell,i)$ as computed by \alggather\ satisfies
\begin{align}
\dpelli_v(\ell,i)=\min_{\abs{\blueset}=i} \aggc_v(\ell,\blueset).
\end{align}
\end{lemma}

\usetikzlibrary{fit,calc}
%define a marking command
\newcommand*{\tikzmk}[1]{\tikz[remember picture,overlay,] \node (#1) {};\ignorespaces}
%define a boxing command, argument = colour of box
\newcommand{\boxit}[5]{\tikz[remember picture,overlay]{\node[yshift=3pt,fill=#1,opacity=.25,fit={($(A)-(#2\linewidth,#3\baselineskip)$)($(B)+(#4\linewidth,#5\baselineskip)$)}] {};}\ignorespaces}
% EXAMPLE USAGE: \boxit{mode1}{0.82}{-0.1}{0.86}{0.65}
% 
%define some colours according to algorithm parts (or any other method you like)
\colorlet{mode1}{red!40}
\colorlet{mode2}{cyan!60}
\colorlet{mode3}{green!70}
\colorlet{mode4}{gray!70}

\makeatletter
\newcommand\footnoteref[1]{\protected@xdef\@thefnmark{\ref{#1}}\@footnotemark}
\makeatother

\begin{algorithm}[t!]
\caption[Algorithm]{\alggather$(\network,\load,\Avilabilty,\numblue)$ at node $v$}
\label{alg:alg:gather}
\begin{algorithmic}[1]
\Require A tree $T$,load $\load$, availability $\Avilabilty$, $\numblue$ $\#$ of blue nodes
\Ensure Correct potential functions, $\dpelli_v, \dpellicol_v$, at each node $v$
\If{$v$ is a leaf node}
    \tikzmk{A}
    \For{$\ell=0,\dots,\distance(v)$}
    \Comment{$\distance(v)$ distance of $v$ from root}
        \State $\dpelli_v(\ell,0)=\rate(v,A^\ell_v) \cdot \load(v)$
        \label{alg:gather:leaf:i_equals_0}
        % \Comment{Eq.~\eqref{eq:leaf}}
        \For{$i=1,\dots,\numblue$}
            \If{$v \in \Avilabilty$}
            \Comment{$v$ has available capacity}
                \State $\dpelli_v(\ell,i)=\rate(v,A^\ell_v)$
                \Comment{$v$ is blue}            \Else
                \State $\dpelli_v(\ell,i)=\rate(v,A^\ell_v) \cdot \load(v)$
                \label{alg:gather:leaf:i_greater_0}
                \Comment{$v$ is red}
            \EndIf
        \EndFor
    \EndFor
    \State send $\dpelli_v$ to $\parent(v)$ and \Return
    \Comment{inform parent}
    % \State \Return
\EndIf
\tikzmk{B}
\boxit{mode1}{0.45}{-0.15}{0.987}{-0.2}
\State wait to receive $\dpelli_c$ from each child $c$ of $v$
\For{$m=1,\ldots,C(v)$}
    \For{$\ell=0,\ldots,\distance(v)$}
    \label{alg:for:ell:start}
        \For{$i=0,\dots,\numblue$}
            \If{$m=1$}
                \tikzmk{A}
                \If{$v \in \Avilabilty$}
                    \State $\dpellicol_v^m(\ell,i,B)=\dpelli_{c_m}(1,i-1)+\rate(v,A^\ell_v)$ \footnote{\label{note1}When $i=0$ then $\dpellicol_v^m(\ell,i,B)= \infty$.}
                    \Else
                        \State $\dpellicol_v^m(\ell,i,B)= \infty$
                \EndIf
                \label{alg:gather:Y_v^m:first_child:blue}
                % \Comment{Eq.~\eqref{eq:firstB}}
                % \State $\parti_v^m(\ell,i,B)=i-1$
                % \Statex \Comment{number of blue nodes to distribute in $T_1$}
                \State $\dpellicol_v^m(\ell,i,R)=\dpelli_{c_m}(\ell+1,i)+\rate(v,A^\ell_v) \cdot \load(v)$
                \label{alg:gather:Y_v^m:first_child:red}
                % \Comment{Eq.~\eqref{eq:firstR}}
                % \State $\parti_v^m(\ell,i,R)=i$
                % \Statex \Comment{number of blue nodes to distribute in $T_{c_m}$}
                \tikzmk{B}
                \boxit{mode2}{0.422}{-0.15}{0.131}{-0.2}
            \Else
                \tikzmk{A}
                \If{$v \in \Avilabilty$}
                    \State $\dpellicol_v^m(\ell,i,\blue) =\mincost(\ell,i,\dpellicol_v^{m-1},\dpelli_{c_m},\blue)$
                    \footnoteref{note1}
                \Else
                        \State $\dpellicol_v^m(\ell,i,B)= \infty$
                \EndIf
                % \footnote{When i=0 then $\dpellicol_v^m(\ell,i,B)= \infty$.}
                \label{alg:gather:Y_v^m:m_child:blue}
                % \State $\displaystyle{\min_{0\leq j \leq i} \left[ \dpellicol_v^{m-1}(\ell,i-j,B) + \dpelli_{c_m}(1,j)\right]}$
                % \State $\parti_v^m(\ell,i,\blue) = \minsplit(\ell,i,\dpellicol_v^{m-1},\dpelli_{c_m},\blue)$
                % \State $\displaystyle{\argmin_{0\leq j\leq i} \left[\dpellicol_v^{m-1}(\ell,i-j,B)+\dpelli_{c_m}(1,j)\right]}$
                % \Statex \Comment{number of blue nodes to distribute in $T_c$}
                \State $\dpellicol_v^m(\ell,i,\red) = \mincost(\ell,i,\dpellicol_v^{m-1},\dpelli_{c_m},\red)$
                \label{alg:gather:Y_v^m:m_child:red}
                % \State $\displaystyle{\min_{0\leq j \leq i} \left[  \dpellicol_v^{m-1}(\ell,i-j,R) + \dpelli_{c_m}(\ell+1,j)\right]}$
                % \State $\parti_v^m(\ell,i,R) = \minsplit(\ell,i,\dpellicol_v^{m-1},\dpelli_{c_m},\red)$
                % \State $\displaystyle{\argmin_{0\leq j \leq i} \left[  \dpellicol_v^{m-1}(\ell,i-j,R) + \dpelli_{c_m}(\ell+1,j)\right]}$
                % \Statex \Comment{number of blue nodes to distribute in $T_c$}
                \tikzmk{B}
                \boxit{mode3}{0.252}{-0.15}{0.2}{-0.2}
            \EndIf
        \EndFor
    \EndFor
    \label{alg:for:ell:end}
\EndFor
\For{$\ell=0,\ldots,\distance(v)$}
    \For{$i=0,\ldots,\numblue$}
        \State $\dpelli_v(\ell,i)=\min\set{ \dpellicol_v^{C(v)}(\ell,i,\blue), \dpellicol_v^{C(v)}(\ell,i,\red)}$
        \label{alg:gather:X_v:min_red_blue}
    \EndFor
\EndFor
\State send $\dpelli_v$ to $\parent(v)$ and \Return

\Statex \hrulefill
\Procedure{$\mincost(\ell,i,\dpellicol_v^{m-1},\dpelli_{c_m},\nodecolor)$}{}
\label{alg:mincost:start}
    \If{$\nodecolor==\blue$}
        \State \Return $\displaystyle{\min_{0\leq j < i} \left[ \dpellicol_v^{m-1}(\ell,i-j,B) + \dpelli_{c_m}(1,j)\right]}$
        \label{alg:gather:mCost:blue}
    \Else
    \Comment{$\nodecolor==\red$}
        \State \Return $\displaystyle{\min_{0\leq j \leq i} \left[  \dpellicol_v^{m-1}(\ell,i-j,R) + \dpelli_{c_m}(\ell+1,j)\right]}$
        \label{alg:gather:mCost:red}
    \EndIf
\EndProcedure
\label{alg:mincost:end}
\end{algorithmic}
\end{algorithm}

\begin{algorithm}[t!]
\caption[Algorithm]{\algcolor$(k)$ at node $v$}
\label{alg:alg:color}
\begin{algorithmic}[1]
\Require $\dpelli_v$, $\dpellicol_v$
\Ensure Optimal coloring
\If{$v$ is destination}
    \State send $(k,1)$ to $r$ \label{alg:algcolor:dest}
\EndIf
\State color $v$ red and wait for $(i,\ell^*)$ from $\parent(v)$
\Statex \Comment{$\ell^*$: distance from root or closest blue ancestor}
\Statex \Comment{$i$: number of blue noes in $\network_v$~~~~~~~~~~~~~~~~~~~~~~~~~~~~~~} 
\If{$v$ is a leaf node and $i>0$}
\tikzmk{A}
    \State color $v$ blue and \Return
\EndIf
\tikzmk{B}
\boxit{mode1}{0.587}{-0.1}{0.987}{-0.15}
\If{$\dpellicol_v^{C(v)}(\ell^*,i,\blue) < \dpellicol_v^{C(v)}(\ell^*,i,\red)$}
    \label{alg:algcolor:condition_for_blue}
    \State color $v$ blue
    \State $\ell^*=0$
    \Comment{reset distance to closest blue ancestor}
\EndIf
\For{$m=C(v),\ldots,2$}
\tikzmk{A}
\Comment{children in reverse order}
    \State $j=\minsplit(\ell^*+1,i,\dpellicol_v^{m-1},\dpelli_{c_m},\mbox{color of }v)$
    \label{alg:algcolor:msplit}
    \State send $(j,\ell^*+1)$ to $c_m$
    \State $i = i-j$
\EndFor
\tikzmk{B}
\boxit{mode3}{0.42}{-0.15}{0.987}{-0.18}
\If{$v$ is blue}
\tikzmk{A}
\Comment{handle $c_1$ last}
    \State send $(i-1,\ell^*+1)$ to $c_1$
\Else
    \State send $(i,\ell^*+1)$ to $c_1$
\EndIf
\tikzmk{B}
\boxit{mode2}{0.319}{-0.15}{0.985}{-0.055}
\State \Return
% ============================================
\Statex \hrulefill
\Procedure{$\minsplit(\ell,i,\dpellicol_v^{m-1},\dpelli_{c_m},\nodecolor)$}{}
\label{alg:minsplit:start}
    \If{$\nodecolor==\blue$}
        \State \Return $\displaystyle{\argmin_{0\leq j < i} \left[\dpellicol_v^{m-1}(\ell,i-j,\blue)+\dpelli_{c_m}(1,j)\right]}$\label{alg:minsplit:blue}
    \Else
    \Comment{$\nodecolor==\red$}
        \State \Return $\displaystyle{\argmin_{0\leq j \leq i} \left[  \dpellicol_v^{m-1}(\ell,i-j,\red) + \dpelli_{c_m}(\ell+1,j)\right]}$\label{alg:minsplit:red}
        
    \EndIf
\EndProcedure
\label{alg:minsplit:end}
\end{algorithmic}
\end{algorithm}

The lemma follows from a double induction argument on the height of the subtree $\network_v$ rooted at any node $v$, and indices $m=1,\ldots,\childnum(v)$ of the children of a node $v$.
%\gabi{The proof is by double induction on the height of $\network_v$ and the number of children $m$ of node $v$. The proof is omitted due to space constraints.}
%The proof is omitted due to space constraints,
\revision{Proof appears in Appendix~\ref{sec:appendix:proofs}.}
%\revision{and can be found in~\cite{RAS21SOAR_TR}.}

Lemma~\ref{lem:gather_correctness:induction} ensures that the values gathered and computed by the nodes while running \alggather\ indeed correspond to the configurations minimizing $\aggc_v(\ell,\blueset)$.
In particular, by Eq.~\eqref{eq:phi_is_pi_at_destination}, the lemma guarantees that the value computed for $\dpelli_d(0,k)$, where $d$ is the destination server, is indeed the minimal utilization cost possible using $\numblue$ blue nodes.

In the second phase of \alg, \algcolor\ essentially traces back the allocation of blue nodes along the optimal path in the dynamic programming performed by \alggather.
To show that \algcolor\ indeed produces an optimal solution to the \bica\ problem we make use of the following lemma.
\revision{Proof appears in Appendix~\ref{sec:appendix:proofs}.}
% \revision{whose proof is omitted and appears in~\cite{RAS21SOAR_TR}.}
%, thus completing the proof of Theorem~\ref{thm:alg_is_optimal}.

% NEW PROOF

\begin{lemma}
\label{lem:color_correctness}
The set of blue nodes $\blueset$ determined by \algcolor\ minimizes the utilization complexity, and  $\abs{\blueset} \leq k$.
% Let $\blueset^*$ be the set of $\numblue$ blue nodes colored by \algcolor$(k)$. Then,
% \begin{align}
%   \dpelli_d(0,k) =\dpelli_r(1,k)=\aggc_r(1,\blueset^*).  
% \end{align}
\end{lemma}

The proof of Theorem~\ref{thm:alg_is_optimal} now follows immediately from combining Lemma~\ref{lem:gather_correctness:induction} and Lemma~\ref{lem:color_correctness}.

\begin{proof}[Proof of Theorem~\ref{thm:alg_is_optimal}.]
The correctness of the algorithm follows from Lemma \ref{lem:gather_correctness:induction} and Lemma \ref{lem:color_correctness}.
For the running time of \alg, we note that it is dominated by the running time of \alggather, which, in turn, is dominated by the for-loop in lines~\ref{alg:for:ell:start}-\ref{alg:for:ell:end}. This loop is performed once for every edge $(v,p(v))$. This gives an overall running time of $O(n\cdot h(\network)\cdot k)$ for this loop over all edges, where in each iteration the $\mincost$ procedure is performed at most once, implying an overall running time of $O(n\cdot h(\network)\cdot k^2)$.
The result follows.
\end{proof}

\section{Related Work}\label{sec:relatedwork}
% \paragraph{Data aggregation.}
Data aggregation has been studied extensively in various contexts~\cite{jesus15survey}, where significant focus was given to wireless sensor networks~\cite{nakamura07information}, alongside scheduling algorithms for optimizing the induced convergecast tree~\cite{malhotra11aggregation}.
% The impact of data aggregation in wireless sensor networks
% Compressed data aggregation for energy efficient wireless sensor networks
% Structure-Free Data Aggregation in Sensor Networks
% A Survey of Distributed Data Aggregation Algorithms
Some of these works also focused on characterizing the type of functions that can be efficiently aggregated~\cite{yu09distributed,jesus15survey}.

% \paragraph{MapReduce.}

MapReduce~\cite{dean04mapreduce} has proven to be a fundamental paradigm for various applications in distributed environments. 
Aside from being a cornerstone of big data analytics, it is also being adopted and incorporated into additional applications and systems, such as large distributed databases, although at the expense of sometimes non-negligible complexity~\cite{yu09distributed}.
Significant efforts were made to improve the performance of MapReduce, including aspects related to scheduling~\cite{zaharia08improving}, data placement~\cite{cheng17improving}, and data coding~\cite{li18fundamental}.

% \paragraph{Distributed Machine Learning.}

Efficiently performing distributed machine learning, and specifically the task of training deep neural networks, has been a fundamental concern in the past decade.
In particular, network bottlenecks are arguably one of the major concerns when executing such tasks~\cite{li14communication,viswanathan20network}.
Various methods for improving network performance and footprint in such systems have been proposed and implemented, including sparsification, quantization,
% etc.~\cite{xu20compressed,dutta20discrepancy},
and scheduling~\cite{xu20compressed,dutta20discrepancy,wang20geryon}.
Furthermore, aspects pertaining to system (and network) heterogeneity and varying network topologies have also been shown to affect the performance of such systems~\cite{abdelmoniem21impact,wang19impact}.
We refer to a recent survey of methods and strategies for optimizing networks for ML~\cite{ouyang21communication}.
Of particular relevance to our work is the efficient scaling of distributed ML using a parameter server, which aggregates local computations, and distributes updated models during training~\cite{li14scaling}. Using this approach has shown to provide significant improvements of ML training tasks, with an emphasis on reducing the network footprint of these tasks~\cite{li14communication,mai15optimizing,luo18parameter}.
Furthermore, the advent of federated ML~\cite{bonawitz19towards} has further increased the efforts of optimizing network performance for ML tasks.

% \paragraph{In-network Computing}

In-network Computing (INC) has recently gained a lot of attention from researchers and industry alike~\cite{ports19when,sapio17innetwork}.
This paradigm is fueled by the ability to {\em program} the data plane, using, e.g., the P4 programming language~\cite{bosshart14p4}, alongside advances in FPGA design and performance (including SmartNICs). Such devices, which enable performing non-trivial computation within the network elements themselves, with minimal effects on performance (e.g., throughput and latency)~\cite{eran19nica}, are effectively deployed by large-scale providers~\cite{firestone18azure}.
Examples of such application logic implementations include 
MapReduce~\cite{costa12camdoop,mai14netagg,sapio17innetwork,bruschi20offloading},
Paxos~\cite{jin18netchain,dang20p4xos,belocchi20paxos},
ML~\cite{xiong19do,sapio19scaling,gebara21innetwork},
caching~\cite{jin17netcache,liu17incbricks},
key-value stores~\cite{tokusashi18lake},
storage replication~\cite{zhu19harmonia,li20pegasus},
IoT data aggregation~\cite{madureira20supporting},
compression~\cite{vaucher20zipline},
lock management~\cite{yu20netlock},
and packet-level ML~\cite{swamy20taurus}.
Some more recent efforts target generalizing INC to arbitrary functionalities~\cite{zhang20gallium}, most predominantly those related to network functions~\cite{shantharama20hardware}, and also studying aspects of energy efficiency of such solutions~\cite{tokusashi19case}.
Similar efforts are
% also
being performed in HPC environments, with special emphasis on support for large scale ML tasks (e.g., Nvidia's SHARP~\cite{graham20sharp}).
\revision{A recent work ~\cite{blocher2021switches} also studied the problem of bounded
resources in in-network computing, but the focus of the work was resource scheduling 
and not optimal placement as in our work.}
%Resource allocation and scheduling of the new capabilities given to switch have become has become ever more complex.
%There new non trivial ways to find the best allocation.
%As shown in~\cite{blocher2021switches}, where presented a new framework for allocating in-network resources given a clients request.}
Whereas most of these works focus on implementation of concrete functionalities within the network, we consider the orthogonal network-level problem of where should such capabilities be deployed, in order to optimize the cumulative system performance, regardless of the specific implementation and/or task to be performed.

% ==================================================
\section{Discussion and Future Work} 
\label{sec:disussion_future_work}

This work considers the \bica\ problem, where we need to determine the location of a limited number of aggregation switches within a tree network, so as to minimize the overall utilization complexity of reduce operations.
This problem lays at the heart of many distributed computing use cases, most notably big data tasks using the MapReduce paradigm, and distributed and federated machine learning.
Our work describes an optimal algorithm, \alg, for solving the \bica~ problem, and provides further insights as to the performance of \alg\ via an extensive simulation study.

A future challenging task would be to develop solutions that are applicable to {\em general} networks (i.e., not necessarily tree networks), thus supporting multi-path routing. 
% A more interesting case, is to consider general topology networks and not only trees. The challenge in this case involves also building the (overlay) trees or networks by which the aggregation is taking place.
Another interesting open problem is related to the multiple workloads scenario. The main question there is how to distribute the overall aggregation capacity available throughout the network to the various workloads being served. Specifically, every workload might be serviced by a {\em distinct} number of aggregation switches (i.e., there need not be a uniform $k$ for all workloads).

Last but not least, we expect our approach to also be effective in designing algorithms that target minimizing the {\em delay} incurred by the system, or minimizing the load on bottleneck links, while using a bounded number of aggregation switches.
However, our methodology may need to be modified significantly for such objectives, and may require new tools and insights.
We conjecture, however, that these objectives -- that of minimizing the overall utilization complexity, and that of minimizing the overall system delay or bottlenecks, 
are closely related, and a solution minimizing one of these objectives is expected to perform well also for the other objectives.

\label{end_body}

\newrevision{
\section*{Acknowledgments}
The authors would like to thank the anonymous reviewers and our shepherd, Shay Vargaftik, for their valuable feedback which helped improve the paper. 
This project
was partially funded by the European Research Council (ERC) under the European Union’s Horizon 2020 research and innovation program (grant agreement No 864228 - AdjustNet).
}
% ==================================================
%\newpage
\bibliographystyle{plain}
\bibliography{bibliography}

% \end{document}

\appendix

% % PREVIOUS ORIGINAL FIGURE BEFORE ADDING SUBFIG CAPTIONS
% \begin{figure}%[t!]
%     \centering
%     \begin{tabular}{c} 
% %    \includegraphics[width=1\columnwidth]{Figs/FreeScale/preferentialAttachment_k_4_msgNum_1222_top_draw.pdf} \\
%     \includegraphics[width=.7\columnwidth]{Figs/FreeScale/preferentialAttachment_k_4_msgNum_602_max_draw_neato.pdf} \\
%     \small{(a) Max placement in $\sfnet{128}$} \\
%     \includegraphics[width=.7\columnwidth]{Figs/FreeScale/preferentialAttachment_k_4_msgNum_187_truffel_draw_neato.pdf} \\
%     \small{(b) \alg\ placement in $\sfnet{128}$} \\
%     \includegraphics[width=0.92\columnwidth]{Figs/FreeScale/scalefree_Message_count_distribution_present_sqrtN_logN_1precent_scale_allRed.pdf}\\
%     \small{(c) Improvement relative to red, for $1\%, \log n, \sqrt{n}$.} 
%     \end{tabular}
% \caption{\alg\ in scale-free networks.}
% \label{fig:SF-networks}
% \end{figure}

% \begin{figure}
%     \centering
%     \subcaptionbox{Improvement relative to red, for $1\%, \log n, \sqrt{n}$. \label{fig:scaling_utilization}}{
%         \includegraphics[width=0.92\columnwidth]{Figs/Size/PS/AllRedScale/PS_Message_count_distribution_present_sqrtN_logN_1precent_sqrtNlogN_scale_alLRed.pdf}
%         } \\
%     \subcaptionbox{Size of $U$ for $30\%, 50\%, 70\%$ improvement \label{fig:scaling_blue_required}}{
%         \includegraphics[width=0.92\columnwidth]{Figs/Size/PS/AllRedScale/PS_Message_count_distribution_present_30_50_70_scale_allRed.pdf}
%         }
% \caption{Scaling performance of \alg}
% \label{fig:scaling}
% \end{figure}

\section{Scaling of \alg}
\label{sec:evaluation:scaling}

% \chen{after one beer and jin \& tonic with Omer :-)}

In this section  we consider the scaling laws of \alg, when applied to larger networks.
In our evaluation of the performance of \alg, we consider several types of bounds on the allowed number of blue nodes, where we allow these bounds to scale as a function of the network size.
Specifically, we focus on binary tree networks of sizes $n=2^i$, for $i=8,\ldots,12$ with constant
rates $1$, and consider $k=f(n)$ blue nodes, for
% Figure \ref{fig:scaling} (a) presents the preforming of $\numblue=f(n)$ as the network size $n$ grows.
% We consider binary trees of size $\btnet{2^x}, x\in\{8,9,10,11,12\}$ and 
$f(n)=\set{0.01n,\log n, \sqrt{n}}$.
In our evaluation we consider the power-law load distribution.

In Fig.~\ref{fig:scaling_utilization} we consider the normalized utilization compared to the all-red scenario. 
first observe that when the number of blue nodes is $1\%$ of the network size, dealing with larger networks implies an improvement in the utilization complexity reduction.
% its improvement (relatively to all-red) increased with the network size. 
For example, for $\btnet{512}$, using merely a $1\%$ fraction of nodes as blue yields a 35\% reduction in utilization complexity compared to the all-red solution, whereas for $\btnet{4096}$, the same fraction of blue nodes results in savings that are above 50\%.
However, when the fraction of blue nodes tends to zero (compared to the size of the network), like in the cases of $\numblue=\log n$ and $\numblue=\sqrt{n}$, the  trend changes and improvement slowly decreases with size.  

Fig.\ref{fig:scaling}(b) considers the dual perspective, and studies the fraction of blue nodes (in \%) required to reach an $\alpha\%$ cost reduction in utilization complexity compared to the all-red solution.
% in the total number of messages to preform \reduce. 
Our results indicate that as the network becomes larger, the fraction of nodes required to obtain any such level of cost reduction, also reduces.
% As we can see in the figure the fraction of blue nodes needed to achieve 30\%, 50\% or 70\% saving is decreasing with the network size.
For example, 70\% saving in utilization complexity on $\btnet{4096}$ can be obtained by taking less than 3\% of the nodes as blue, whereas achieving merely 50\% saving, requires less than 1\% of the nodes being blue.
We note that these scaling laws are computed using our optimal algorithm, \alg.

\begin{figure}
    \centering
    \subcaptionbox{Improvement relative to red, for $1\%, \log n, \sqrt{n}$. \label{fig:scaling_utilization}}{
        \includegraphics[width=0.92\columnwidth]{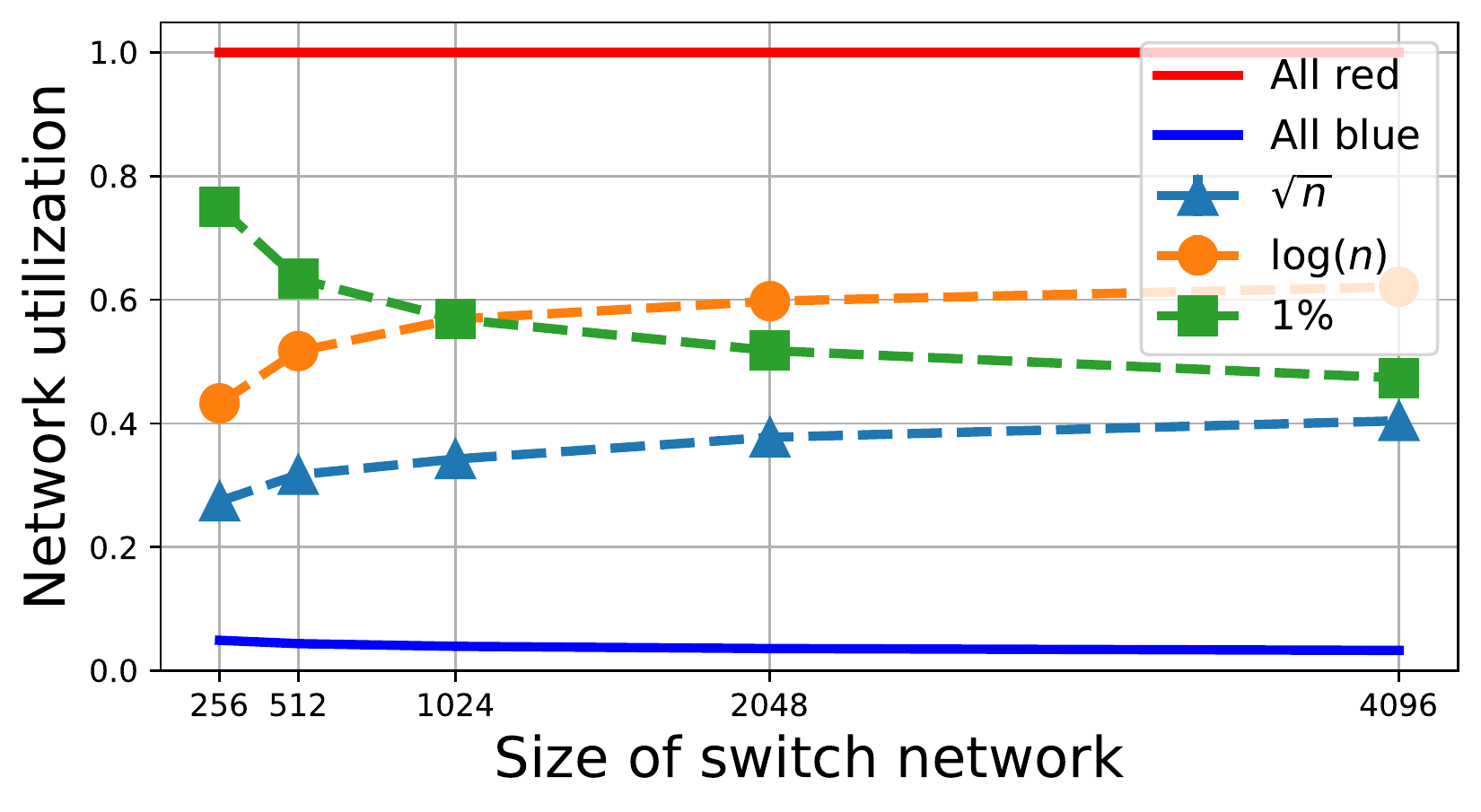}
        } \\
    \subcaptionbox{Size of $U$ for $30\%, 50\%, 70\%$ improvement \label{fig:scaling_blue_required}}{
        \includegraphics[width=0.92\columnwidth]{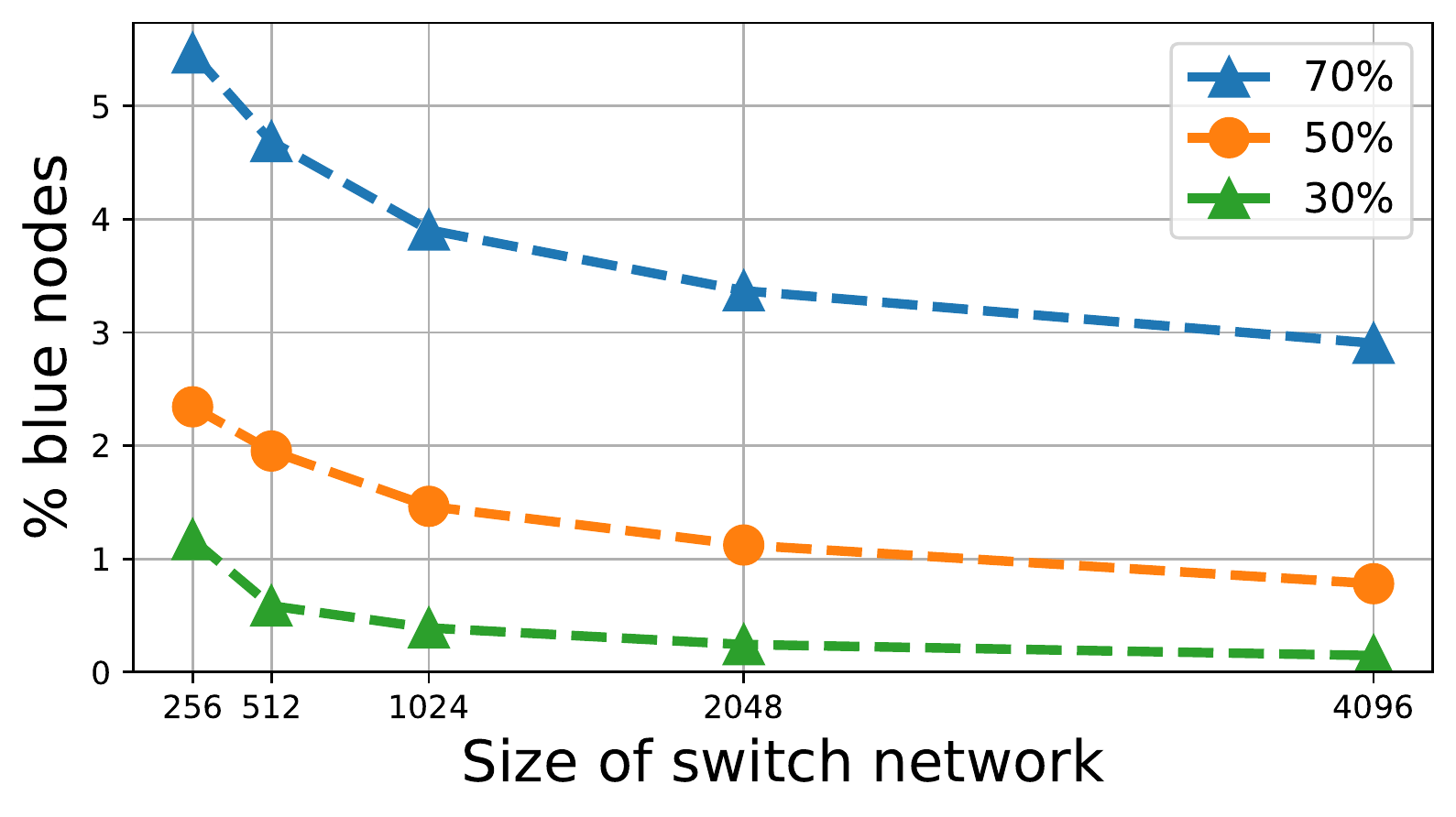}
        }
\caption{Scaling performance of \alg}
\label{fig:scaling}
\end{figure}

\section{Scale-Free Networks}
\label{sec:evaluation:scale-free}

In this subsection we demonstrate the applicability of the \alg\ algorithm to non-regular tree networks.
We study the performance of \alg\ on random preferential attachment (\rpa) trees \cite{barabasi1999emergence} which are known to produce scale-free networks: networks with a degree distribution that follows a power-law distribution \cite{newman2018networks}.
We denote by $\sfnet{n}$ a random network with $n$ nodes produced by the \rpa\ procedure.
When studying scale-free tree networks, in order to avoid introducing a bias into the evaluation, we consider networks where each node has a load of 1.

Fig.~\ref{fig:SF-networks}(a) and~\ref{fig:SF-networks}(b) show an example of a scale-free network, $\sfnet{128}$.
In this example, the degree sequence of the nine highest degree nodes is: $\set{18,15,9,8,6,6,6,5,4}$. 
Since \rpa\ produces scale-free tree networks, a natural strategy for placing blue nodes in such a network is the \maxalg\ algorithm, where blue nodes are placed at the nodes with the highest degree (which are usually closer to the root).
The result of the \maxalg\ policy is depicted in Fig.~\ref{fig:SF-networks}(a), resulting in a utilization complexity of 621.
In contrast, the optimal solution produced by \alg\ is provided in Fig.~\ref{fig:SF-networks}(b).
This solution allocates the blue nodes at nodes with degrees
$18,15,6,4$, incurring a utilization complexity of merely 182, which translates to saving roughly 70\% of the messages, compared to \maxalg.

Finally, Fig.~\ref{fig:SF-networks}(c) presents the scaling of \alg\ for increasing sizes $n=2^i$ of scale-free networks, for $i=8,\ldots,12$.
Using $\numblue=0.01 n$ and $\numblue=\log n$ exhibits similar results to those presented in Fig.~\ref{fig:scaling} for binary trees.
Interestingly for $\numblue=\sqrt{n}$ our results suggest that
as the network size increases the utilization complexity remains close to 40\% of the all-red scenario.

\begin{figure*}
    \centering
    \begin{tabular}{ccc}
        \subcaptionbox{Max placement in $\sfnet{128}$ \label{fig:SF-networks_max}}{
            \includegraphics[width=0.29\textwidth]{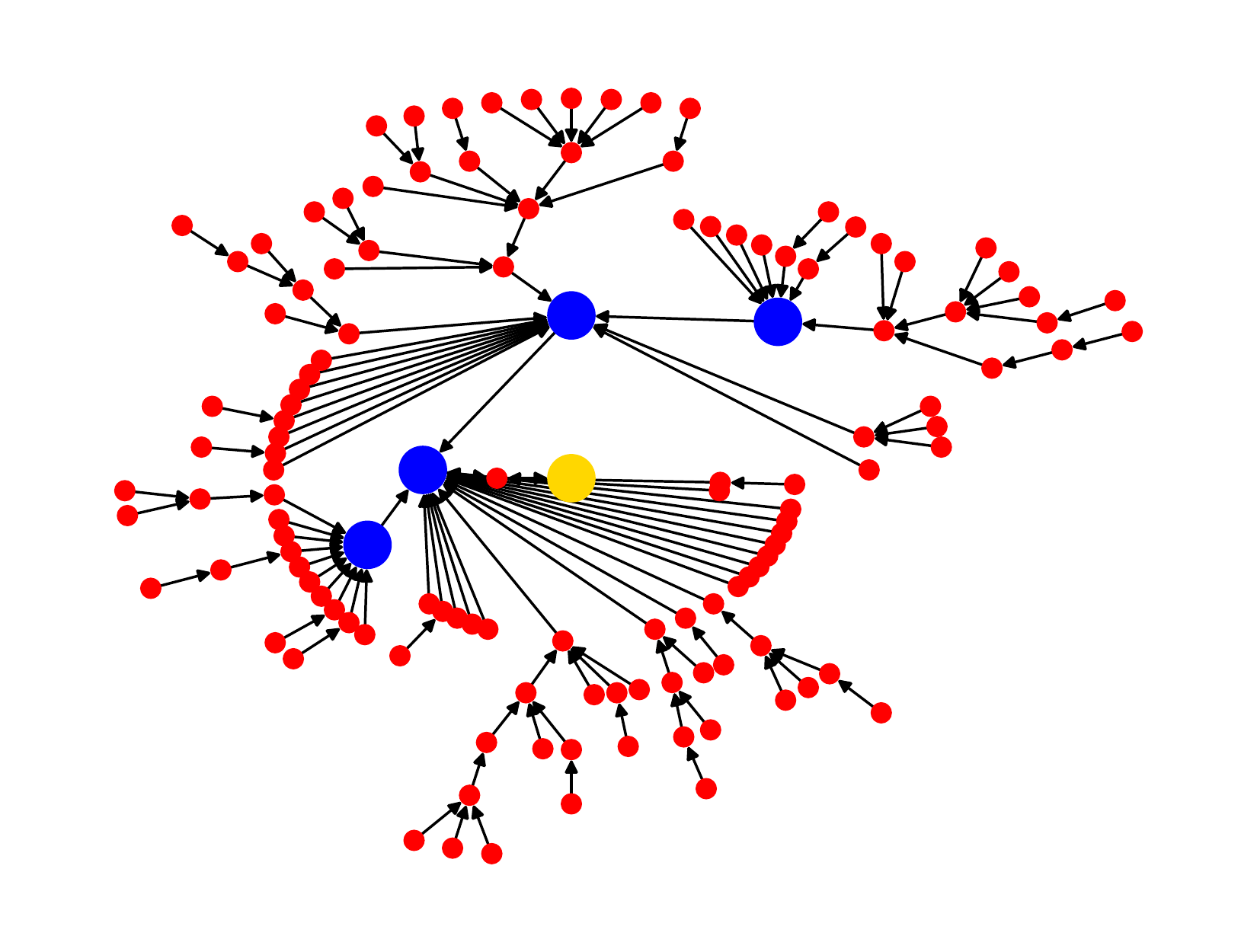}
            } &
        \subcaptionbox{\alg\ placement in $\sfnet{128}$ \label{fig:SF-networks_alg}}{
            \includegraphics[width=0.29\textwidth]{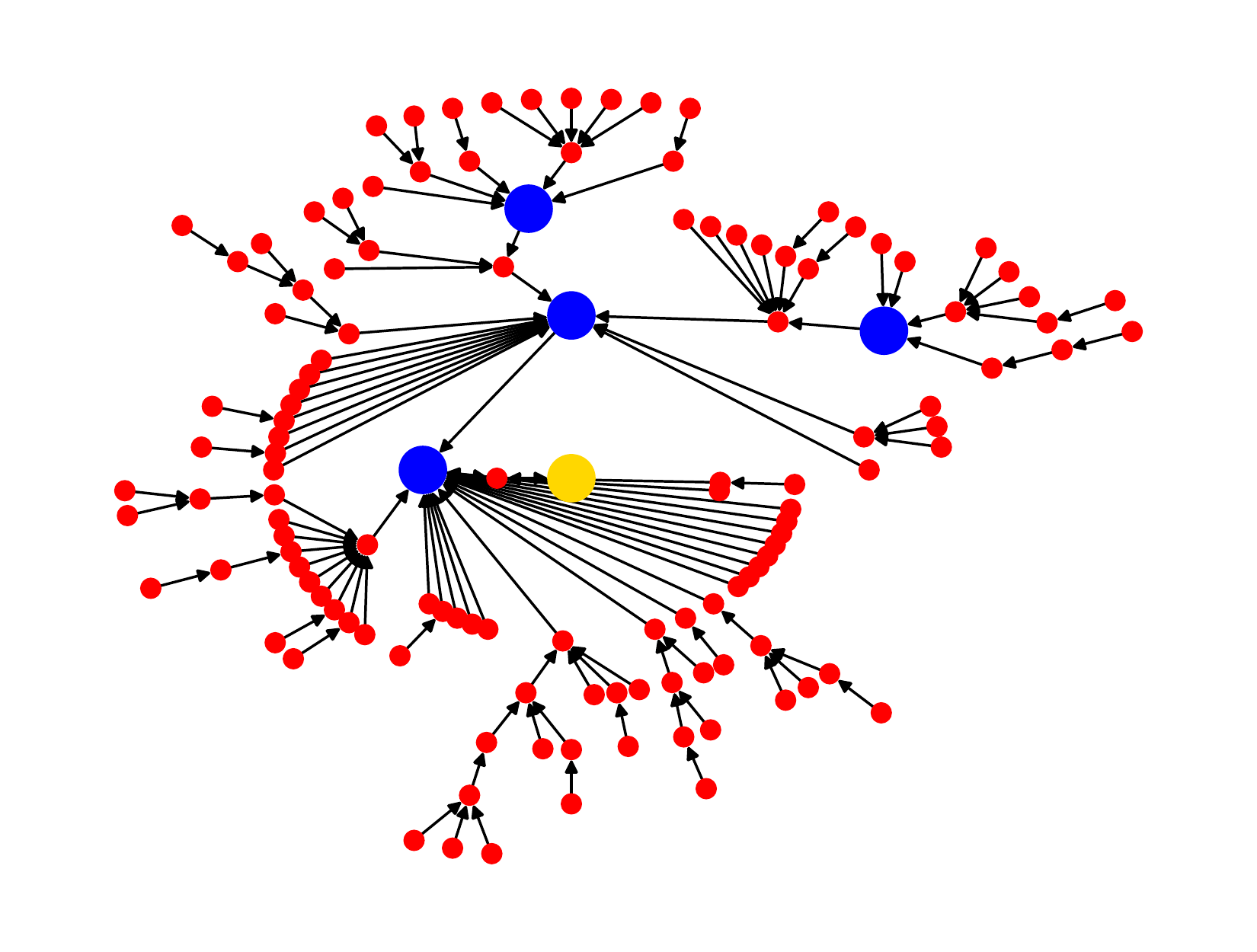}
            } &
        \subcaptionbox{Improvement relative to red, for $1\%, \log n, \sqrt{n}$ \label{fig:SF-networks_utilization}}{
            \includegraphics[width=0.30\textwidth]{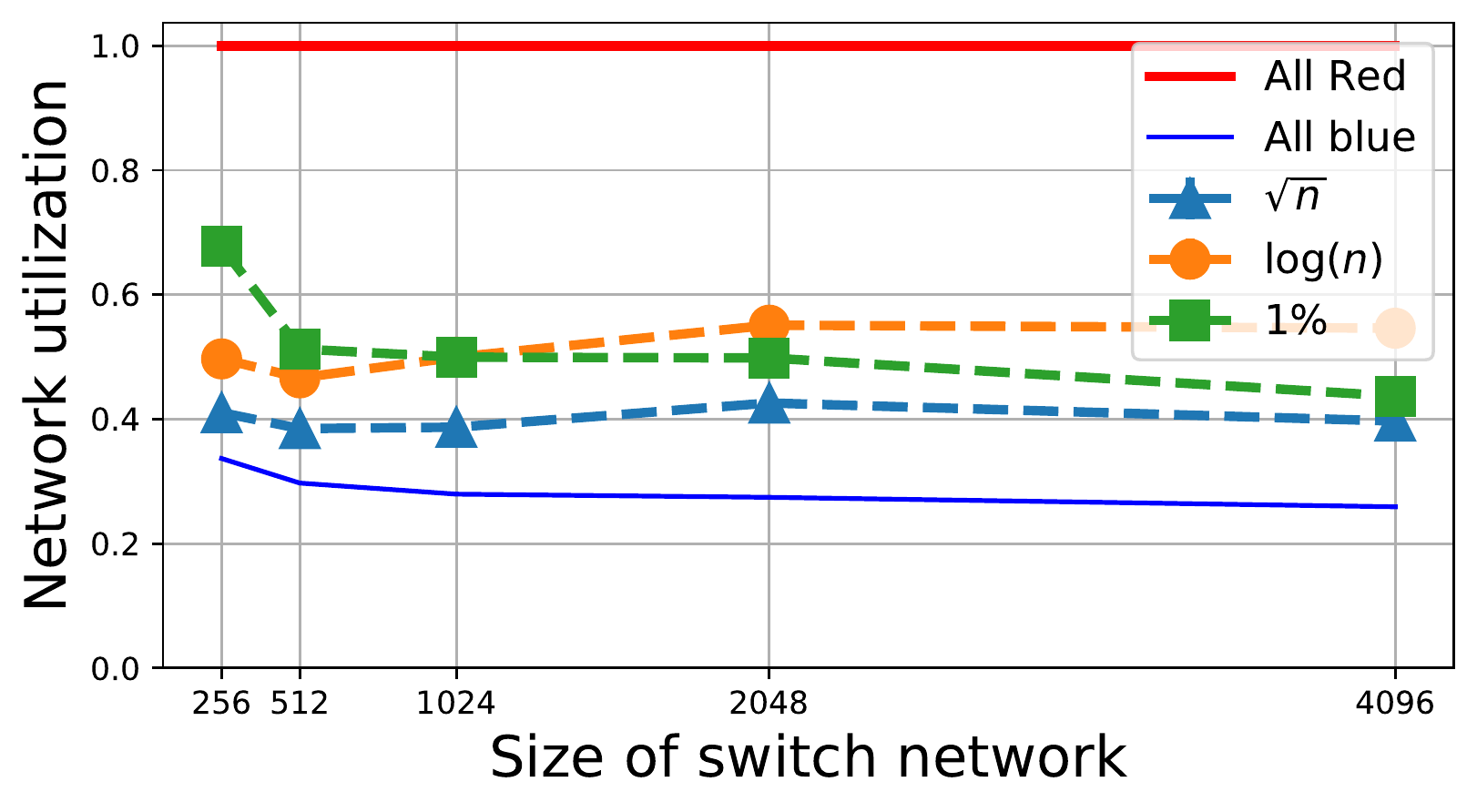}
            }
    \end{tabular}
\caption{\alg\ in scale-free networks.}
\label{fig:SF-networks}
\end{figure*}

\section{Proofs}
\label{sec:appendix:proofs}

\subsection{Proof of Lemma~\ref{lem:pidecompose}}
\label{sec:appendix:proof:pidecompose}

\begin{proof}
By Eq.~\eqref{eq:potential:non_leaf}, we can formulate the recursive formula
\begin{align}
\aggc_v^m(\ell,\blueset)
&= 
\begin{cases}
\aggc_{c_1}(1,\blueset) + 1 \cdot \rate(v,A_v^\ell)
&\quad \mbox{if $m = 1$} \\
\aggc_v^{m-1}(\ell,\blueset) + \aggc_{c_m}(1,\blueset)
&\quad \mbox{if $m \geq 2$}.
\label{eq:potential:blue:recursive}
\end{cases}
\end{align}
for the case where $v \in \blueset$, and the recursive formula
\begin{align}
\aggc_v^m(\ell,\blueset)
&= 
\begin{cases}
\aggc_{c_1}(\ell+1,\blueset) +  \load(v) \cdot \rate(v,A_v^\ell)
&\quad \mbox{if $m = 1$} \\
\aggc_v^{m-1}(\ell,\blueset) + \aggc_{c_m}(\ell+1,\blueset)
&\quad \mbox{if $m \geq 2$}.
\label{eq:potential:red:recursive}
\end{cases}
\end{align}
for the case where $v \notin \blueset$.
The result follows from solving these recursions.
\end{proof}

% \section{TODO:}
% \begin{enumerate}
% \item consider adding a "takeway" sentence/bottom line at the end of each section (like Understanding Lifecycle Management Complexity
% of Datacenter Topologies)
% \item what we DON'T do: choosing the tree/root, implementing the Process function in the switches, where to implement the functionality?
% \item can we say something about load-oblivious placements? E.g., can we use the uniform?
% \item consider adding a small section of future directions about delay. Specifically, show that they are not the same...
% \item simulations:
%     \begin{enumerate}
%     \item MapReduce, increasing $\numblue$, multiple distributions, single regular graph, multiple algorithms
%     \item Parameter server
%     \gabi{find use case parameters }
%     \item Different topologies-
%     \end{enumerate}
% \end{enumerate}

% ==================================================
% OLD APPENDIX & OLD ALGORITHMS AND PROOFS
% ==================================================

\subsection{Proof of Lemma~\ref{lem:gather_correctness:induction}}
\label{sec:appendix:proof:gather_correctness}

\begin{proof}
The proof is by double induction on the height of $\network_v$ and the number of children $m$ for which 
$\dpellicol_v^m(\ell,i,\red)$ and $\dpellicol_v^m(\ell,i,\blue)$ have been computed correctly.
\paragraph{$v$ is a leaf in $\network$:}
If $i=0$ we have
% \gabi{we actually prove here the basis of the induction on $i$ for $i=0$.}
\begin{align}
\dpelli_v(\ell,i)
=\ell \cdot \load(v)
=\min_{\abs{\blueset}=i} \aggc_v^{\childnum(v)}(\ell,\blueset),
\end{align}
where the first equality follows from line~\ref{alg:gather:leaf:i_equals_0} in Alg.~\ref{alg:alg:gather}, and the last equality follows from the definition of the $(v,m)$-potential for a leaf node in Eq.~\eqref{eq:potential:leaf} by taking $\blueset=\emptyset$.
% \gabi{next, we effectively prove the basis of the induction on $i$ for $i>0$.}
If $i>0$ then having $\abs{\blueset}=i$ implies that $v$ is blue (since it is a leaf in $\network$), leading to
\begin{align}
\dpelli_v(\ell,i)
=\ell
=\min_{\abs{\blueset}=i} \aggc_v^{\childnum(v)}(\ell,\blueset),
\end{align}
where again, the equality follows from Eq.~\eqref{eq:potential:leaf}.
This completes the base case for our induction on the height of $\network_v$. We henceforth assume that $\dpelli_{v'}(\ell,i)$ has been computed correctly for all nodes $v'$ below $v$, and for all $i$. In particular, this is true for every child $c_m$ of a non-leaf node $v$, $m=1,\ldots,\childnum(v)$.

\paragraph{$v$ is a non-leaf, $m=1$:}
Assume $v$ is blue and $i \geq 1$. It follows that
\begin{align}
\min_{\substack{\abs{\blueset}=i\\v \in \blueset}} \aggc_v^1(\ell,\blueset)
&= \min_{\substack{\abs{\blueset}=i\\v \in \blueset}} \myls \mylp \sum_{\link \in \network_v^1} \msg_{\link}(\network_v^1,\load,\blueset) \myrp \notag \\
& \quad\quad\quad + \ell \cdot \msg_{(v,\parent(v))}(\tilde{\network}_v^1,\load,\blueset)  \myrs
\label{eq:induction:m_1:blue:1} \\
&= \min_{\substack{\abs{\blueset}=i\\v \in \blueset}} \myls \sum_{\link \in \network_v^1} \msg_{\link}(\network_v^1,\load,\blueset) \myrs + \ell
\label{eq:induction:m_1:blue:2} \\
&= \min_{\abs{\blueset'}=i-1} \myls \mylp \sum_{\link \in \network_{c_1}} \msg_{\link}(\network_{c_1},\load,\blueset') \myrp \notag \\
& \quad\quad\quad + \msg_{(c_1,v)}(\tilde{\network}_{c_1},\load,\blueset') \myrs + \ell
\label{eq:induction:m_1:blue:3} \\
&= \min_{\abs{\blueset'}=i-1} \aggc_{c_1}(1,\blueset') + \ell
\label{eq:induction:m_1:blue:4} \\
&= \dpelli_{c_1}(1,i-1) + \ell
\label{eq:induction:m_1:blue:5} \\
&= \dpellicol_v^1(\ell,i,\blue)
\label{eq:induction:m_1:blue:6} \end{align}
Eq.~\eqref{eq:induction:m_1:blue:1} follows from the definition of the $(v,m)$-potential in Eq.~\eqref{eq:potential:non_leaf}.
Eq.~\eqref{eq:induction:m_1:blue:2} and Eq.~\eqref{eq:induction:m_1:blue:3} follow from the definition of $\network_v^1$ (which also contains link $(c_1,v)$), along with the fact that $v$ is blue in the current case considered (hence the move from $\blueset$ to $\blueset'$ with requiring that $v \in \blueset'$), and therefore forwards a single message to its parent $\parent(v)$.
Eq.~\eqref{eq:induction:m_1:blue:4} again follows from the definition of the $(v,m)$-potential in Eq.~\eqref{eq:potential:non_leaf}, and
Eq.~\eqref{eq:induction:m_1:blue:5} follows from the induction hypothesis on $c_1$ (where the height of $\network_{c_1}$ is strictly smaller than that of $\network_v$).
Finally, Eq.~\eqref{eq:induction:m_1:blue:6} follows from line~\ref{alg:gather:Y_v^m:first_child:blue} in Alg.~\ref{alg:alg:gather}.

Assume next that $v$ is red and $i \geq 0$.
It follows that
\begin{align}
\min_{\substack{\abs{\blueset}=i\\v \notin \blueset}} \aggc_v^1(\ell,\blueset)
&= \min_{\substack{\abs{\blueset}=i\\v \notin \blueset}} \myls \mylp \sum_{\link \in \network_v^1} \msg_{\link}(\network_v^1,\load,\blueset) \myrp \notag \\
& \quad\quad\quad + \ell \cdot \msg_{(v,\parent(v))}(\tilde{\network}_v^1,\load,\blueset)  \myrs
\label{eq:induction:m_1:red:1} \\
&= \min_{\abs{\blueset}=i} \myls \mylp \sum_{\link \in \network_{c_1}} \msg_{\link}(\network_{c_1},\load,\blueset) \myrp \notag \\
& \quad\quad\quad + \msg_{(c_1,v)}(\tilde{\network}_{c_1},\load,\blueset) \notag \\
& \quad\quad\quad + \ell \cdot \mylp \msg_{(c_1,v)}(\tilde{\network}_{c_1},\load,\blueset) \notag \\
& \quad\quad\quad + \load(v) \myrp \myrs
\label{eq:induction:m_1:red:2} \\
&= \min_{\abs{\blueset}=i} \myls \mylp \sum_{\link \in \network_{c_1}} \msg_{\link}(\network_{c_1},\load,\blueset) \myrp \notag \\
& \quad\quad\quad + (\ell+1) \cdot \msg_{(c_1,v)}(\tilde{\network}_{c_1},\load,\blueset') \myrs \notag \\
& \quad\quad\quad + \ell \cdot \load(v)
\label{eq:induction:m_1:red:3} \\
&= \min_{\abs{\blueset}=i} \aggc_{c_1}(\ell+1,\blueset) + \ell \cdot \load(v)
\label{eq:induction:m_1:red:4} \\
&= \dpelli_{c_1}(\ell+1,i) + \ell \cdot \load(v)
\label{eq:induction:m_1:red:5} \\
&= \dpellicol_v^1(\ell,i,\red)
\label{eq:induction:m_1:red:6}
\end{align}
Eq.~\eqref{eq:induction:m_1:red:1} follows from the definition of the $(v,m)$-potential in Eq.~\eqref{eq:potential:non_leaf}.
Eq.~\eqref{eq:induction:m_1:red:2} follows from the definition of $\network_v^1$ (which also contains link $(c_1,v)$), along with the fact that $v$ is red in the current case considered, and forwards $\msg_{(c_1,v)}(\tilde{\network}_{c_1},\load,\blueset) +  \load(v)$ messages across link $(v,\parent(v))$ in $\tilde{\network}_v^1$.
Eq.~\eqref{eq:induction:m_1:red:3} follows from simple algebraic manipulation.
Eq.~\eqref{eq:induction:m_1:red:4} again follows from the definition of the $(v,m)$-potential in Eq.~\eqref{eq:potential:non_leaf}, and
Eq.~\eqref{eq:induction:m_1:red:5} follows from the induction hypothesis on $c_1$ (where the height of $\network_{c_1}$ is strictly smaller than that of $\network_v$).
Finally, Eq.~\eqref{eq:induction:m_1:red:6} follows from line~\ref{alg:gather:Y_v^m:first_child:red} in Alg.~\ref{alg:alg:gather}.

\paragraph{$v$ is a non-leaf, $m>1$:}
We assume the claim holds for all nodes $u$ for which the height of $\network_u$ is strictly less than that of $\network_v$, and that for all $i$ both $\dpelli_v(\ell,i)$ have been computed correctly, and for all $m'<m$, $\dpellicol_v^{m'}(\ell,i,\red)$ and $\dpellicol_v^{m'}(\ell,i,\blue)$ have been computed correctly.

Assume first that $v$ is blue and $i \geq 1$. It follows that
\begin{align}
\min_{\substack{\abs{\blueset}=i\\v \in \blueset}} \aggc_v^m(\ell,\blueset)
&= \min_{\substack{\abs{\blueset}=i\\v \in \blueset}} \myls \mylp \sum_{\link \in \network_v^m} \msg_{\link}(\network_v^m,\load,\blueset) \myrp \notag \\
& \quad\quad\quad + \ell \cdot \msg_{(v,\parent(v))}(\tilde{\network}_v^m,\load,\blueset) \myrs
\label{eq:induction:m_geq_2:blue:1} \\
&= \min_{\substack{\abs{\blueset}=i\\v \in \blueset}} \myls \mylp \sum_{\link \in \network_v^{m-1}} \msg_{\link}(\network_v^{m-1},\load,\blueset) \myrp \notag \\
& \quad\quad\quad + \ell + \mylp \sum_{\link \in \network_{c_m}} \msg_{\link}(\network_{c_m},\load,\blueset) \myrp \notag \\
& \quad\quad\quad + \msg_{(c_m,v)}(\tilde{\network}_{c_m},\load,\blueset) \myrs
\label{eq:induction:m_geq_2:blue:2} \\
&= \mylp \sum_{\link \in \network_v^{m-1}} \msg_{\link}(\network_v^{m-1},\load,\blueset^*) \myrp \notag \\
& \quad\quad\quad + \ell + \mylp \sum_{\link \in \network_{c_m}} \msg_{\link}(\network_{c_m},\load,\blueset^*) \myrp \notag \\
& \quad\quad\quad + \msg_{(c_m,v)}(\tilde{\network}_{c_m},\load,\blueset^*),
\label{eq:induction:m_geq_2:blue:3} \end{align}
where $\blueset^*$ is the set of blue nodes attaining the minimum in Eq.~\eqref{eq:induction:m_geq_2:blue:2}.
Transition Eq.~\eqref{eq:induction:m_geq_2:blue:1} follows from the definition of the $(v,m)$-potential in Eq.~\eqref{eq:potential:non_leaf}. Eq.~\eqref{eq:induction:m_geq_2:blue:2} follows from the definition of $\network_v^m$ (which also contains link $(c_m,v)$), along with the fact that $v$ is blue in the current case considered, implying that  both $\msg_{(v,\parent(v))}(\tilde{\network}_v^m,\load,\blueset^*)$ and $\msg_{(v,\parent(v))}(\tilde{\network}_v^{m-1},\load,\blueset^*)$ are just 1.

Let $\blueset'=\blueset^* \cap T_{c_m}$ and let $j'=\abs{\blueset'}$. Further, let $\blueset''=\blueset^* \cap T_v^{m-1}$ and let $j''=\abs{\blueset''}=i-j'$.
Note that $\blueset^*=\blueset' \cup \blueset''$, $j'+j''=i$, and $v \in \blueset''$.
It follows that
\begin{align}
\min_{\substack{\abs{\blueset}=i\\v \in \blueset}} \aggc_v^m(\ell,\blueset)
&= \mylp \sum_{\link \in \network_v^{m-1}} \msg_{\link}(\network_v^{m-1},\load,\blueset'') \myrp + \ell \notag \\
& \quad\quad\quad + \mylp \sum_{\link \in \network_{c_m}} \msg_{\link}(\network_{c_m},\load,\blueset') \myrp \notag \\
& \quad\quad\quad + \msg_{(c_m,v)}(\tilde{\network}_{c_m},\load,\blueset')
\label{eq:induction:m_geq_2:blue:4} \\
&= \min_{\substack{\abs{\tilde{\blueset}''}=i-j'\\v \in \tilde{\blueset}''}} \myls \mylp \sum_{\link \in \network_v^{m-1}} \msg_{\link}(\network_v^{m-1},\load,\tilde{\blueset}'') \myrp \notag \\
& \quad\quad\quad + \ell \myrs \notag \\
& \quad\quad\quad + \min_{\abs{\tilde{\blueset}'}=j'} \myls \mylp \sum_{\link \in \network_{c_m}} \msg_{\link}(\network_{c_m},\load,\tilde{\blueset}') \myrp \notag \\
& \quad\quad\quad + \msg_{(c_m,v)}(\tilde{\network}_{c_m},\load,\tilde{\blueset}') \myrs
\label{eq:induction:m_geq_2:blue:5} \\
&= \min_{\substack{\abs{\tilde{\blueset}''}=i-j'\\v \in \tilde{\blueset}''}} \aggc_v^{m-1}(\ell,\tilde{\blueset}'') \notag \\
& \quad\quad\quad + \min_{\abs{\tilde{\blueset}'}=j'} \aggc_{c_m}(1,\tilde{\blueset}') \label{eq:induction:m_geq_2:blue:6} \\
&= \dpellicol_v^{m-1}(\ell,i-j',\blue) + \dpelli_{c_m}(1,j') \label{eq:induction:m_geq_2:blue:7} \\
&= \min_{0 \leq j \leq i} \left[ \dpellicol_v^{m-1}(\ell,i-j,\blue) + \dpelli_{c_m}(1,j) \right]. \label{eq:induction:m_geq_2:blue:8} \end{align}
By the definition of $\blueset',\blueset''$, substituting these terms in~\eqref{eq:induction:m_geq_2:blue:3}, we obtain \eqref{eq:induction:m_geq_2:blue:4}.
Next, we show the equality of \eqref{eq:induction:m_geq_2:blue:4} and \eqref{eq:induction:m_geq_2:blue:5}.
First note that by definition, \eqref{eq:induction:m_geq_2:blue:4} is no smaller than \eqref{eq:induction:m_geq_2:blue:5}. Assume by contradiction that \eqref{eq:induction:m_geq_2:blue:4} is strictly larger than \eqref{eq:induction:m_geq_2:blue:5}, and let $\bar{\blueset}'$, and $\bar{\blueset}''$ be the sets obtaining the minimum for the first and second term in~\eqref{eq:induction:m_geq_2:blue:5}, respectively, satisfying $\bar{\blueset}'' \subseteq \network_v^{m-1}$, $\bar{\blueset}' \subseteq \network_{c_m}$,  $\abs{\bar{\blueset}''}=i-j'$, $\abs{\bar{\blueset}'}=j'$, and 
$v \in \bar{\blueset}''$.
Since $\network_v^{m-1} \cap \network_{c_m} = \emptyset$, it follows that $\bar{\blueset}'' \cap \bar{\blueset}'' = \emptyset$. This, in turn, implies that $\abs{\bar{\blueset}}=i$, and by our derivation $\aggc_v^m(\ell,\bar{\blueset}) < \aggc_v^m(\ell,\blueset^*)$, contradicting the minimality of $\blueset^*$.
Eq.~\eqref{eq:induction:m_geq_2:blue:6} follows from the definition of $(v,m)$-potential in Eq.~\eqref{eq:potential:non_leaf}.
Eq.~\eqref{eq:induction:m_geq_2:blue:6} follows from the induction hypothesis on $m$ and $c_m$.
Finally, we show that Eq.~\eqref{eq:induction:m_geq_2:blue:6} equals \eqref{eq:induction:m_geq_2:blue:7}.
Clearly~\eqref{eq:induction:m_geq_2:blue:6} is no smaller than~\eqref{eq:induction:m_geq_2:blue:7}. Assume by contradiction that~\eqref{eq:induction:m_geq_2:blue:6} is strictly larger than~\eqref{eq:induction:m_geq_2:blue:7}, and let $j^*$ be the value for which the minimim in Eq.~\eqref{eq:induction:m_geq_2:blue:7} is obtained.
By the induction hypothesis on $m$ (for $\dpellicol_v^{m-1}$) and $c_m$ (for $\dpelli_{c_m}$), with $i-j^*$ and $j^*$, respectively, there exist disjoint sets 
$\bar{\blueset}''$ and $\bar{\blueset}'$ 
of sizes $i-j^*$ and $j^*$, respectively, such that
$\bar{\blueset}'' \subseteq T_v^{m-1}$, $\bar{\blueset}' \subseteq T_{c_m}$, and $v \in \bar{\blueset}''$.
It follows that taking $\bar{\blueset}=\bar{\blueset}'' \cup \bar{\blueset}'$ we obtain using our derivation that
$\aggc_v^m(\ell,\bar{\blueset}) < \aggc_v^m(\ell,\blueset^*)$, contradicting the minimality of $U^*$.
This completes the proof for the case where $v$ is blue.

Assume next that $v$ is red and $i \geq 0$. It follows that
\begin{align}
\min_{\substack{\abs{\blueset}=i\\v \notin \blueset}} \aggc_v^m(\ell,\blueset)
&= \min_{\substack{\abs{\blueset}=i\\v \notin \blueset}} \myls \mylp \sum_{\link \in \network_v^m} \msg_{\link}(\network_v^m,\load,\blueset) \myrp \notag \\
& \quad\quad\quad + \ell \cdot \msg_{(v,\parent(v))}(\tilde{\network}_v^m,\load,\blueset) \myrs
\label{eq:induction:m_geq_2:red:1} \\
&= \min_{\substack{\abs{\blueset}=i\\v \notin \blueset}} \myls \mylp \sum_{\link \in \network_v^{m-1}} \msg_{\link}(\network_v^{m-1},\load,\blueset) \myrp \notag \\
& \quad\quad\quad + \mylp \sum_{\link \in \network_{c_m}} \msg_{\link}(\network_{c_m},\load,\blueset) \myrp \notag \\
& \quad\quad\quad + (\ell+1) \cdot \msg_{(c_m,v)}(\tilde{\network}_{c_m},\load,\blueset) \notag \\
& \quad\quad\quad + \ell \cdot \msg_{(v,\parent(v))}(\tilde{\network}_v^{m-1},\load,\blueset) \myrs
\label{eq:induction:m_geq_2:red:2} \\
&= \mylp \sum_{\link \in \network_v^{m-1}} \msg_{\link}(\network_v^{m-1},\load,\blueset^*) \myrp \notag \\
& \quad\quad\quad + \mylp \sum_{\link \in \network_{c_m}} \msg_{\link}(\network_{c_m},\load,\blueset^*) \myrp \notag \\
& \quad\quad\quad + (\ell+1) \cdot \msg_{(c_m,v)}(\tilde{\network}_{c_m},\load,\blueset^*) \notag \\
& \quad\quad\quad + \ell \cdot \msg_{(v,\parent(v))}(\tilde{\network}_v^{m-1},\load,\blueset^*)
\label{eq:induction:m_geq_2:red:3} \end{align}
where $\blueset^*$ is the set of blue nodes attaining the minimum in Eq.~\eqref{eq:induction:m_geq_2:red:2}.
Transition Eq.~\eqref{eq:induction:m_geq_2:red:1} follows from the definition of the $(v,m)$-potential in Eq.~\eqref{eq:potential:non_leaf}. Eq.~\eqref{eq:induction:m_geq_2:red:2} follows from the definition of $\network_v^m$, which  can be decomposed into $\network_v^{m-1} \cup \network_{c_m} \cup \set{(c_m,v)}$.
Most of the derivation of Eq.~\eqref{eq:induction:m_geq_2:red:2} trivially follows from this decomposition. One non trivial observation follows from noting that the messages traversing $(c_m,v)$ are counted in Eq.~\eqref{eq:induction:m_geq_2:red:2} within the messages traversing $(v,\parent(v)$, as well as in the summation over all edges in $\network_v^m$.
Since  $v$ is red in the current case considered, this implies that the messages traversing $(c_m,v)$ are accounted for $(\ell+1)$ times.

Let $\blueset'=\blueset^* \cap T_{c_m}$ and let $j'=\abs{\blueset'}$. Further, let $\blueset''=\blueset^* \cap T_v^{m-1}$ and let $j''=\abs{\blueset''}=i-j'$.
Note that $\blueset^*=\blueset' \cup \blueset''$, $j'+j''=i$, and $v \notin \blueset''$.
Using the same arguments used for proving the case where $v$ is blue, one can show that 
\begin{align}
\min_{\substack{\abs{\blueset}=i\\v \notin \blueset}} \aggc_v^m(\ell,\blueset)
&= \mylp \sum_{\link \in \network_v^{m-1}} \msg_{\link}(\network_v^{m-1},\load,\blueset'') \myrp \notag \\
& \quad\quad\quad + \mylp \sum_{\link \in \network_{c_m}} \msg_{\link}(\network_{c_m},\load,\blueset') \myrp \notag \\
& \quad\quad\quad + (\ell+1) \cdot \msg_{(c_m,v)}(\tilde{\network}_{c_m},\load,\blueset') \notag \\
& \quad\quad\quad + \ell \cdot \msg_{(v,\parent(v))}(\tilde{\network}_v^{m-1},\load,\blueset'')
\label{eq:induction:m_geq_2:red:4} \\
&= \min_{0 \leq j \leq i} \left[ \dpellicol_v^{m-1}(\ell,i-j,\blue) + \dpelli_{c_m}(\ell+1,j) \right], \label{eq:induction:m_geq_2:red:5} \end{align}
thus completing the proof for the case where $v$ is red.
The lemma now follows.
\end{proof}

\subsection{Proof of Lemma~\ref{lem:color_correctness}}
\label{sec:appendix:proof:color_correctness}

\begin{proof}
In what follows, we say a node $v$ is \emph{correctly assigned} if: (i) it is colored so as to minimize the utilization cost of the entire system,
(ii) it is allotted the number of blue nodes for $\network_v$ so as to minimize the utilization cost of the entire system and (iii) 
$\ell^*_v$, the distance from $v$ to its closets blue ancestor or $d$, is assigned so as to minimize the utilization cost of the entire system.
We prove by induction on the order of handling nodes by \algcolor\ that if node $v$ is correctly assigned 
then each of its children $c_m$, $m=1,\ldots,\childnum(v)$ is correctly assigned.
%and is correctly allotted the number of blue nodes that should be allocated in $\network_v$, 
%then each of its children $c_m$, $m=1,\ldots,\childnum(v)$ is colored correctly and is correctly allotted the number of blue nodes that should be allocated in $\network_{c_m}$. Note that if for a node $v$ all of its ancestors are colored correctly then, $\ell^*$, the distance to is closets blue ancestor or $d$ is well defined.

For the base case, consider node $\destination$, which should have $\numblue$ blue nodes in its subtree, it's color is trivially not blue (since $\destination$ is a server) and $\ell^*_d=0$. So $\destination$ is correctly assigned. $\destination$ has a single child, $\rootswitch$, and by line~\ref{alg:algcolor:condition_for_blue} of \algcolor, along with Eq.~\eqref{eq:lem:alggather:correctness:red} and Eq.~\eqref{eq:lem:alggather:correctness:blue} $\rootswitch$ is colored correctly, since by line~\ref{alg:gather:X_v:min_red_blue} %and~\ref{alg:gather:Y_v^m:first_child:red} 
of \alggather, its color is the one minimizing $\dpelli_{\rootswitch}(1,\numblue) =\dpelli_{\destination}(0,\numblue)$. Clearly by line \ref{alg:algcolor:dest} in \algcolor\ $r$ is correctly assigned with $\ell^*_r=1$ and $\numblue$ blue nodes.

Assume the claim holds for all nodes handled before node $v$, and consider node $v$ which is correctly assigned.
First, since $v$ is correctly colored, then by Lemma~\ref{lem:pidecompose} and Equations \eqref{eq:ellBlue} and \eqref{eq:ellRed} each child $c$ will be correctly assigned it's $\ell^*_c$ in \algcolor;
if $v$ is blue then for each child $c$, $\ell^*_c=1$ and if $v$ is red $\ell^*_c=\ell^*_v + 1$ for every child $c$.
Next, by induction on the number of children of $v$ from $\childnum(v)$ to 1,
it is easy to show
% from the proof of Lemma~\ref{lem:gather_correctness:induction}
that each child $c$ is assigned the correct number of blue nodes to be distributed in its subtree $\network_c$.
This follows from the fact that the $\minsplit$ procedure in lines~\ref{alg:minsplit:start}-\ref{alg:minsplit:end} of \algcolor\ essentially extract the value $j$ obtaining the minimum considered also by the $\mincost$ procedure in lines~\ref{alg:mincost:start}-\ref{alg:mincost:end} in \alggather, when applied the same value of $\ell=\ell_c$.
% The result follows by lines \ref{alg:minsplit:blue} and \ref{alg:minsplit:red} and Claim \ref{lem:pidecompose}, Equations \eqref{eq:potential:blue} and \eqref{eq:potential:red}
% when $v$ it either blue or red and since $\ell^*_c$ is correct.
Lastly, since each child $c$ is assigned correctly $\ell^*_c$ and the correct number of blue nodes, by Lemma \ref{lem:gather_correctness:induction} and
line~\ref{alg:gather:X_v:min_red_blue} of \alggather, $c$ will be also colored correctly.
\end{proof}

\label{end_total}
\end{document}

%% file: aggregation_allred.tex
% \documentclass{standalone}

% \usepackage{xcolor}
% \usepackage{pgfplotstable}
% \usepackage{pgfplots}
% \usepackage{tikz}
% \usetikzlibrary{shapes,arrows,fit,calc,positioning,patterns}
% \usepgfplotslibrary{groupplots}
% \pgfplotsset{compat=1.15}

% \newcommand{\destination}{d}
% \newcommand{\rootswitch}{r}

% \tikzset{
%   treenode/.style = {align=center, inner sep=0pt, text centered},
%   node_b/.style = {treenode, circle, white, draw=black, font=\Huge, fill=blue, text width=3em},
%   node_r/.style = {treenode, circle, white, draw=black, font=\Huge,fill=red ,text width=3em, very thick},
%   node_ws/.style = {treenode, circle, draw=white, font=\Huge, text width=3em, very thick},
%   node_p/.style = {treenode, rectangle, draw=black, font=\Huge, minimum height = 1.1cm, text width=3em, very thick},
%   node_h/.style = {treenode, rectangle, white, draw=black, font=\Huge, minimum height = 1.1cm, text width=3em, fill=gray, very thick},
%   node_wh/.style = {treenode, rectangle, draw=white, font=\Huge, minimum height = 1.1cm, text width=3em, very thick},
%   edge from parent/.style = {font=\Huge, line width=1mm, draw, <-, >=stealth'},
%   level 2/.style = {sibling distance=40mm},
%   level 3/.style = {sibling distance=20mm},
%   level/.style = {level distance=2.5cm},
% }
% \begin{document}
%\title{$f(x_1,x_2,x_3,x_4,x_5x_6)$}
\begin{tikzpicture}
\node
    [node_p] {$\destination$}
        child{ node [node_r_full] {$\rootswitch$}
             child{ node [node_r_full] (S1) {}
                 child{  node [node_p] (h1) {$x_1$}
                    edge from parent[]
                    				    % edge from parent[style={-,dashed}] node[left,xshift=-2mm] {}
                }
                child{  node [node_p] (h2) {$x_2$}
                    edge from parent[]
                    				    % edge from parent[style={-,dashed}] node[left,xshift=-2mm] {}
                }
                edge from parent node[left,xshift=-3mm] {2}
             }
            child{node [node_r_full] (s2) {}
                 child{  node [node_p] (h3) {$x_3$}
                    edge from parent[] }
            edge from parent node[left,xshift=-1mm] {1}
            }
            child{ node [node_r_full] (s3){}
                child{  node [node_p] (h4) {$x_4$}
                    edge from parent[]}
                child{ node [node_r_full](s4) {}
                    child{  node [node_p] (h5) {$x_5$}
                    edge from parent[]}
                    child{  node [node_p] (h6) {$x_6$}
                    edge from parent[]}
                edge from parent node[left,xshift=-1mm] {2}
                }
            edge from parent node[left,xshift=-3mm] {3}
            }
        edge from parent node[left,xshift=-2mm] {6}
        }
;
\end{tikzpicture}

%% file: aggregation_allblue.tex
% \documentclass{standalone}

% \usepackage{xcolor}
% \usepackage{pgfplotstable}
% \usepackage{pgfplots}
% \usepackage{tikz}
% \usetikzlibrary{shapes,arrows,fit,calc,positioning,patterns}
% \usepgfplotslibrary{groupplots}
% \pgfplotsset{compat=1.15}

% \newcommand{\destination}{d}
% \newcommand{\rootswitch}{r}

% \tikzset{
%   treenode/.style = {align=center, inner sep=0pt, text centered},
%   node_b/.style = {treenode, circle, white, draw=black, font=\Huge, fill=blue, text width=3em},
%   node_r/.style = {treenode, circle, white, draw=black, font=\Huge,fill=red ,text width=3em, very thick},
%   node_ws/.style = {treenode, circle, draw=white, font=\Huge, text width=3em, very thick},
%   node_p/.style = {treenode, rectangle, draw=black, font=\Huge, minimum height = 1.1cm, text width=3em, very thick},
%   node_h/.style = {treenode, rectangle, white, draw=black, font=\Huge, minimum height = 1.1cm, text width=3em, fill=gray, very thick},
%   node_wh/.style = {treenode, rectangle, draw=white, font=\Huge, minimum height = 1.1cm, text width=3em, very thick},
%   edge from parent/.style = {font=\Huge, line width=1mm, draw, <-, >=stealth'},
%   level 2/.style = {sibling distance=40mm},
%   level 3/.style = {sibling distance=20mm},
%   level/.style = {level distance=2.5cm},
% }
% \begin{document}
% \title{$f(x_1,x_2,x_3,x_4,x_5x_6)$}
\begin{tikzpicture}
\node
    [node_p] {$\destination$}
        child{ node [node_b] {$\rootswitch$}
             child{ node [node_b] (S1) {}
                 child{  node [node_p] (h1) {$x_1$}
                    edge from parent[]
                    				    % edge from parent[style={-,dashed}] node[left,xshift=-2mm] {}
                }
                child{  node [node_p] (h2) {$x_2$}
                    edge from parent[]
                    				    % edge from parent[style={-,dashed}] node[left,xshift=-2mm] {}
                }
                edge from parent node[left,xshift=-3mm] {1}
             }
            child{node [node_b] (s2) {}
                 child{  node [node_p] (h3) {$x_3$}
                    edge from parent[] }
            edge from parent node[left,xshift=-1mm] {1}
            }
            child{ node [node_b] (s3){}
                child{  node [node_p] (h4) {$x_4$}
                    edge from parent[]}
                child{ node [node_b](s4) {}
                    child{  node [node_p] (h5) {$x_5$}
                    edge from parent[]}
                    child{  node [node_p] (h6) {$x_6$}
                    edge from parent[]}
                edge from parent node[left,xshift=-1mm] {1}
                }
            edge from parent node[left,xshift=-3mm] {1}
            }
        edge from parent node[left,xshift=-2mm] {1}
        }
;
\end{tikzpicture}

% \end{document}

%% file: toy1_1_top.tex
\begin{tikzpicture}
\node
    [node_p] {$\destination$}
        child{ node [node_b] {$\rootswitch$} 
                child{  node [node_r] {} 
                        child{  node [node_r] (s1) {}
                                child{  node [node_h] (h1) {2}
                                        edge from parent[draw=none]
                    				    % edge from parent[style={-,dashed}] node[left,xshift=-2mm] {}
                                }
            				    edge from parent node[left,xshift=-2mm] {2}
                        }
                        child{  node [node_r] (s2) {}
                                child{  node [node_h] (h2) {6}
                                        edge from parent[draw=none]
                    				    % edge from parent[style={-,dashed}] node[left,xshift=-2mm] {}
                                }
            				    edge from parent node[right,xshift=2mm] {6}
                        }
                        edge from parent node[left,xshift=-2mm] {8}
                }
                child{  node [node_b] {} 
                        child{  node [node_r] (s3) {}
                                child{  node [node_h] (h3) {5}
                                        edge from parent[draw=none]
                    				    % edge from parent[style={-,dashed}] node[left,xshift=-2mm] {}
                                }
            				    edge from parent node[left,xshift=-2mm] {5}
                        }
                        child{  node [node_r] (s4) {}
                                child{  node [node_h] (h4) {4}
                                        edge from parent[draw=none]
                    				    % edge from parent[style={-,dashed}] node[left,xshift=-2mm] {}
                                }
            				    edge from parent node[right,xshift=2mm] {4}
                        }
                        edge from parent node[right,xshift=2mm] {1}
                }
                edge from parent node[right,xshift=2mm] {1}
        }
;

\path[-,dashed,line width=1mm,]
    (s1) edge (h1)
    (s2) edge (h2)
    (s3) edge (h3)
    (s4) edge (h4)
;
    
\end{tikzpicture}

%% file: toy1_2_max.tex
\begin{tikzpicture}
\node
    [node_p] {$\destination$}
        child{ node [node_r] {$\rootswitch$} 
                child{  node [node_r] {} 
                        child{  node [node_r] (s1) {} 
                                child{  node [node_h] (h1) {2}
                                        edge from parent[draw=none]
                    				    % edge from parent[style={-,dashed}] node[left,xshift=-2mm] {}
                                }
            				    edge from parent node[left,xshift=-2mm] {2}
                        }
                        child{  node [node_b] (s2) {}
                                child{  node [node_h] (h2) {6}
                                        edge from parent[draw=none]
                    				    % edge from parent[style={-,dashed}] node[left,xshift=-2mm] {}
                                }
            				    edge from parent node[right,xshift=2mm] {1}
                        }
                        edge from parent node[left,xshift=-2mm] {3}
                }
                child{  node [node_r] {} 
                        child{  node [node_b] (s3) {}
                                child{  node [node_h] (h3) {5}
                                        edge from parent[draw=none]
                    				    % edge from parent[style={-,dashed}] node[left,xshift=-2mm] {}
                                }
            				    edge from parent node[left,xshift=-2mm] {1}
                        }
                        child{  node [node_r] (s4) {}
                                child{  node [node_h] (h4) {4}
                                        edge from parent[draw=none]
                    				    % edge from parent[style={-,dashed}] node[left,xshift=-2mm] {}
                                }
            				    edge from parent node[right,xshift=2mm] {4}
                        }
                        edge from parent node[right,xshift=2mm] {5}
                        }
                edge from parent node[right,xshift=2mm] {8}
        }
;

\path[-,dashed,line width=1mm,]
    (s1) edge (h1)
    (s2) edge (h2)
    (s3) edge (h3)
    (s4) edge (h4)
;

\end{tikzpicture}

%% file: toy1_3_level.tex
\begin{tikzpicture}
\node
    [node_p] {$\destination$}
        child{ node [node_r] {$\rootswitch$} 
                child{  node [node_b] {} 
                        child{  node [node_r] (s1) {}
                                child{  node [node_h] (h1) {2}
                                        edge from parent[draw=none]
                    				    % edge from parent[style={-,dashed}] node[left,xshift=-2mm] {}
                                }
            				    edge from parent node[left,xshift=-2mm] {2}
                        }
                        child{  node [node_r] (s2) {}
                                child{  node [node_h] (h2) {6}
                                        edge from parent[draw=none]
                    				    % edge from parent[style={-,dashed}] node[left,xshift=-2mm] {}
                                }
            				    edge from parent node[right,xshift=2mm] {6}
                        }
                        edge from parent node[left,xshift=-2mm] {1}
                }
                child{  node [node_b] {} 
                        child{  node [node_r] (s3) {}
                                child{  node [node_h] (h3) {5}
                                        edge from parent[draw=none]
                    				    % edge from parent[style={-,dashed}] node[left,xshift=-2mm] {}
                                }
            				    edge from parent node[left,xshift=-2mm] {5}
                        }
                        child{  node [node_r] (s4) {}
                                child{  node [node_h] (h4) {4}
                                        edge from parent[draw=none]
                    				    % edge from parent[style={-,dashed}] node[left,xshift=-2mm] {}
                                }
            				    edge from parent node[right,xshift=2mm] {4}
                        }
                        edge from parent node[right,xshift=2mm] {1}                }
                edge from parent node[right,xshift=2mm] {2}
        }
;

\path[-,dashed,line width=1mm,]
    (s1) edge (h1)
    (s2) edge (h2)
    (s3) edge (h3)
    (s4) edge (h4)
;

\end{tikzpicture}

%% file: toy1_4_alg.tex
\begin{tikzpicture}
\node
    [node_p] {$\destination$}
        child{ node [node_r] {$\rootswitch$} 
                child{  node [node_r] {} 
                        child{  node [node_r] (s1) {} 
                                child{  node [node_h] (h1) {2}
                                        edge from parent[draw=none]
                    				    % edge from parent[style={-,dashed}] node[left,xshift=-2mm] {}
                                }
            				    edge from parent node[left,xshift=-2mm] {2}
                        }
                        child{  node [node_b] (s2) {}
                                child{  node [node_h] (h2) {6}
                                        edge from parent[draw=none]
                    				    % edge from parent[style={-,dashed}] node[left,xshift=-2mm] {}
                                }
            				    edge from parent node[right,xshift=2mm] {1}
                        }
                        edge from parent node[left,xshift=-2mm] {3}
                }
                child{  node [node_b] {} 
                        child{  node [node_r] (s3) {} 
                                child{  node [node_h] (h3) {5}
                                        edge from parent[draw=none]
                    				    % edge from parent[style={-,dashed}] node[left,xshift=-2mm] {}
                                }
            				    edge from parent node[left,xshift=-2mm] {5}
                        }
                        child{  node [node_r] (s4) {}
                                child{  node [node_h] (h4) {4}
                                        edge from parent[draw=none]
                    				    % edge from parent[style={-,dashed}] node[left,xshift=-2mm] {}
                                }
            				    edge from parent node[right,xshift=2mm] {4}
                        }
                        edge from parent node[right,xshift=2mm] {1}
                }
                edge from parent node[right,xshift=2mm] {4}
        }
;

\path[-,dashed,line width=1mm,]
    (s1) edge (h1)
    (s2) edge (h2)
    (s3) edge (h3)
    (s4) edge (h4)
;
\end{tikzpicture}

%% file: toy_2_k1.tex
\begin{tikzpicture}
\node
    [node_p] {$\destination$}
        child{ node [node_r] {$\rootswitch$} 
                child{  node [node_r] {} 
                        child{  node [node_r] (s1) {} 
                                child{  node [node_h] (h1) {2}
                                        edge from parent[draw=none]
                    				    % edge from parent[style={-,dashed}] node[left,xshift=-2mm] {}
                                }
            				    edge from parent node[left,xshift=-2mm] {2}
                        }
                        child{  node [node_r] (s2) {}
                                child{  node [node_h] (h2) {6}
                                        edge from parent[draw=none]
                    				    % edge from parent[style={-,dashed}] node[left,xshift=-2mm] {}
                                }
            				    edge from parent node[right,xshift=2mm] {6}
                        }
                        edge from parent node[left,xshift=-2mm] {8}
                }
                child{  node [node_b] {} 
                        child{  node [node_r] (s3) {} 
                                child{  node [node_h] (h3) {5}
                                        edge from parent[draw=none]
                    				    % edge from parent[style={-,dashed}] node[left,xshift=-2mm] {}
                                }
            				    edge from parent node[left,xshift=-2mm] {5}
                        }
                        child{  node [node_r] (s4) {}
                                child{  node [node_h] (h4) {4}
                                        edge from parent[draw=none]
                    				    % edge from parent[style={-,dashed}] node[left,xshift=-2mm] {}
                                }
            				    edge from parent node[right,xshift=2mm] {4}
                        }
                        edge from parent node[right,xshift=2mm] {1}
                }
                edge from parent node[right,xshift=2mm] {9}
        }
;

\path[-,dashed,line width=1mm,]
    (s1) edge (h1)
    (s2) edge (h2)
    (s3) edge (h3)
    (s4) edge (h4)
;

\end{tikzpicture}

%% file: toy_2_k2.tex
\begin{tikzpicture}
\node
    [node_p] {$\destination$}
        child{ node [node_r] {$\rootswitch$} 
                child{  node [node_r] {} 
                        child{  node [node_r] (s1) {} 
                                child{  node [node_h] (h1) {2}
                                        edge from parent[draw=none]
                    				    % edge from parent[style={-,dashed}] node[left,xshift=-2mm] {}
                                }
            				    edge from parent node[left,xshift=-2mm] {2}
                        }
                        child{  node [node_b] (s2) {}
                                child{  node [node_h] (h2) {6}
                                        edge from parent[draw=none]
                    				    % edge from parent[style={-,dashed}] node[left,xshift=-2mm] {}
                                }
            				    edge from parent node[right,xshift=2mm] {1}
                        }
                        edge from parent node[left,xshift=-2mm] {3}
                }
                child{  node [node_b] {} 
                        child{  node [node_r] (s3) {} 
                                child{  node [node_h] (h3) {5}
                                        edge from parent[draw=none]
                    				    % edge from parent[style={-,dashed}] node[left,xshift=-2mm] {}
                                }
            				    edge from parent node[left,xshift=-2mm] {5}
                        }
                        child{  node [node_r] (s4) {}
                                child{  node [node_h] (h4) {4}
                                        edge from parent[draw=none]
                    				    % edge from parent[style={-,dashed}] node[left,xshift=-2mm] {}
                                }
            				    edge from parent node[right,xshift=2mm] {4}
                        }
                        edge from parent node[right,xshift=2mm] {1}
                }
                edge from parent node[right,xshift=2mm] {4}
        }
;

\path[-,dashed,line width=1mm,]
    (s1) edge (h1)
    (s2) edge (h2)
    (s3) edge (h3)
    (s4) edge (h4)
;

\end{tikzpicture}

%% file: toy_2_k3.tex
\begin{tikzpicture}
\node
    [node_p] {$\destination$}
        child{ node [node_r] {$\rootswitch$} 
                child{  node [node_r] {} 
                        child{  node [node_r] (s1) {} 
                                child{  node [node_h] (h1) {2}
                                        edge from parent[draw=none]
                    				    % edge from parent[style={-,dashed}] node[left,xshift=-2mm] {}
                                }
            				    edge from parent node[left,xshift=-2mm] {2}
                        }
                        child{  node [node_b] (s2) {}
                                child{  node [node_h] (h2) {6}
                                        edge from parent[draw=none]
                    				    % edge from parent[style={-,dashed}] node[left,xshift=-2mm] {}
                                }
            				    edge from parent node[right,xshift=2mm] {1}
                        }
                        edge from parent node[left,xshift=-2mm] {3}
                }
                child{  node [node_r] {} 
                        child{  node [node_b] (s3) {} 
                                child{  node [node_h] (h3) {5}
                                        edge from parent[draw=none]
                    				    % edge from parent[style={-,dashed}] node[left,xshift=-2mm] {}
                                }
            				    edge from parent node[left,xshift=-2mm] {1}
                        }
                        child{  node [node_b] (s4) {}
                                child{  node [node_h] (h4) {4}
                                        edge from parent[draw=none]
                    				    % edge from parent[style={-,dashed}] node[left,xshift=-2mm] {}
                                }
            				    edge from parent node[right,xshift=2mm] {1}
                        }
                        edge from parent node[right,xshift=2mm] {2}
                }
                edge from parent node[right,xshift=2mm] {5}
        }
;

\path[-,dashed,line width=1mm,]
    (s1) edge (h1)
    (s2) edge (h2)
    (s3) edge (h3)
    (s4) edge (h4)
;

\end{tikzpicture}

%% file: toy_2_k4.tex
\begin{tikzpicture}
\node
    [node_p] {$\destination$}
        child{ node [node_r] {$\rootswitch$} 
                child{  node [node_b] {} 
                        child{  node [node_r] (s1) {} 
                                child{  node [node_h] (h1) {2}
                                        edge from parent[draw=none]
                    				    % edge from parent[style={-,dashed}] node[left,xshift=-2mm] {}
                                }
            				    edge from parent node[left,xshift=-2mm] {2}
                        }
                        child{  node [node_b] (s2) {}
                                child{  node [node_h] (h2) {6}
                                        edge from parent[draw=none]
                    				    % edge from parent[style={-,dashed}] node[left,xshift=-2mm] {}
                                }
            				    edge from parent node[right,xshift=2mm] {1}
                        }
                        edge from parent node[left,xshift=-2mm] {1}
                }
                child{  node [node_r] {} 
                        child{  node [node_b] (s3) {} 
                                child{  node [node_h] (h3) {5}
                                        edge from parent[draw=none]
                    				    % edge from parent[style={-,dashed}] node[left,xshift=-2mm] {}
                                }
            				    edge from parent node[left,xshift=-2mm] {1}
                        }
                        child{  node [node_b] (s4) {}
                                child{  node [node_h] (h4) {4}
                                        edge from parent[draw=none]
                    				    % edge from parent[style={-,dashed}] node[left,xshift=-2mm] {}
                                }
            				    edge from parent node[right,xshift=2mm] {1}
                        }
                        edge from parent node[right,xshift=2mm] {2}
                }
                edge from parent node[right,xshift=2mm] {3}
        }
;

\path[-,dashed,line width=1mm,]
    (s1) edge (h1)
    (s2) edge (h2)
    (s3) edge (h3)
    (s4) edge (h4)
;

\end{tikzpicture}

%% file: fig_barrier_3.tex
\begin{tikzpicture}
\node
    [node_wh] {}
        child{ node [node_ws] {} 
                child{  node [node_ws] {} 
                        child{  node [node_ws] (s1) {} 
                                child{  node [node_wh] (h1) {}
                                        edge from parent[draw=none]
                    				    % edge from parent[style={-,dashed}] node[left,xshift=-2mm] {}
                                }
                                edge from parent[draw=none]
            				    % edge from parent node[left,xshift=-2mm] {2}
                        }
                        child{  node [node_ws] (s2) {}
                                child{  node [node_wh] (h2) {}
                                        edge from parent[draw=none]
                    				    % edge from parent[style={-,dashed}] node[left,xshift=-2mm] {}
                                }
                                edge from parent[draw=none]
            				    % edge from parent node[right,xshift=2mm] {1}
                        }
                        edge from parent[draw=none]
                        % edge from parent node[left,xshift=-2mm] {3}
                }
                child{  node [node_b] {} 
                        child{  node [node_r] (s3) {} 
                                child{  node [node_h] (h3) {5}
                                        edge from parent[draw=none]
                    				    % edge from parent[style={-,dashed}] node[left,xshift=-2mm] {}
                                }
            				    edge from parent node[left,xshift=-2mm] {5}
                        }
                        child{  node [node_r] (s4) {}
                                child{  node [node_h] (h4) {4}
                                        edge from parent[draw=none]
                    				    % edge from parent[style={-,dashed}] node[left,xshift=-2mm] {}
                                }
            				    edge from parent node[right,xshift=2mm] {4}
                        }
                        edge from parent[draw=none]
                        % edge from parent node[right,xshift=2mm] {1}
                }
                edge from parent[draw=none]
                % edge from parent node[right,xshift=2mm] {4}
        }
;

\path[-,dashed,line width=1mm,]
    % (s1) edge (h1)
    % (s2) edge (h2)
    (s3) edge (h3)
    (s4) edge (h4)
;

\end{tikzpicture}

%% file: fig_barrier_2.tex
\begin{tikzpicture}
\node
    [node_wh] {}
        child{ node [node_ws] {} 
                child{  node [node_ws] {} 
                        child{  node [node_ws] (s1) {} 
                                child{  node [node_wh] (h1) {}
                                        edge from parent[draw=none]
                    				    % edge from parent[style={-,dashed}] node[left,xshift=-2mm] {}
                                }
                                edge from parent[draw=none]
            				    % edge from parent node[left,xshift=-2mm] {2}
                        }
                        child{  node [node_b] (s2) {}
                                child{  node [node_h] (h2) {6}
                                        edge from parent[draw=none]
                    				    % edge from parent[style={-,dashed}] node[left,xshift=-2mm] {}
                                }
                                edge from parent[draw=none]
            				    % edge from parent node[right,xshift=2mm] {1}
                        }
                        edge from parent[draw=none]
                        % edge from parent node[left,xshift=-2mm] {3}
                }
                child{  node [node_ws] {} 
                        child{  node [node_ws] (s3) {} 
                                child{  node [node_wh] (h3) {}
                                        edge from parent[draw=none]
                    				    % edge from parent[style={-,dashed}] node[left,xshift=-2mm] {}
                                }
                                edge from parent[draw=none]
            				    % edge from parent node[left,xshift=-2mm] {5}
                        }
                        child{  node [node_ws] (s4) {}
                                child{  node [node_wh] (h4) {}
                                        edge from parent[draw=none]
                    				    % edge from parent[style={-,dashed}] node[left,xshift=-2mm] {}
                                }
                                edge from parent[draw=none]
            				    % edge from parent node[right,xshift=2mm] {4}
                        }
                        edge from parent[draw=none]
                        % edge from parent node[right,xshift=2mm] {1}
                }
                edge from parent[draw=none]
                % edge from parent node[right,xshift=2mm] {4}
        }
;

\path[-,dashed,line width=1mm,]
    % (s1) edge (h1)
    (s2) edge (h2)
    % (s3) edge (h3)
    % (s4) edge (h4)
;

\end{tikzpicture}

%% file: fig_barrier_1.tex
\begin{tikzpicture}
\node
    [node_p] {$\destination$}
        child{ node [node_r] {$\rootswitch$} % root
                child{  node [node_r] {} % left child of root
                        child{  node [node_r] (s1) {} % left child of left child of root
                                child{  node [node_h] (h1) {2}
                                        edge from parent[draw=none]
                    				    % edge from parent[style={-,dashed}] node[left,xshift=-2mm] {}
                                }
            				    edge from parent node[left,xshift=-2mm] {2}
                        }
                        child{  node [node_b] (s2) {} % right child of left child of root
                                child{  node [node_h] (h2) {1}
                                        edge from parent[draw=none]
                    				    % edge from parent[style={-,dashed}] node[left,xshift=-2mm] {}
                                }
            				    edge from parent node[right,xshift=2mm] {1}
                        }
                        edge from parent node[left,xshift=-2mm] {3}
                }
                child{  node [node_b] (rcr) {} % right child of root
                        child{  node [node_h] (h3) {1}
                                edge from parent[draw=none]
            				    % edge from parent[style={-,dashed}] node[left,xshift=-2mm] {}
                        }
                        edge from parent node[right,xshift=2mm] {1}
                }
                edge from parent node[right,xshift=2mm] {4}
        }
;

\path[-,dashed,line width=1mm,]
    (s1) edge (h1)
    (s2) edge (h2)
    (rcr) edge (h3)
    % (s4) edge (h4)
;

\end{tikzpicture}